\documentclass[11pt]{article}
\usepackage{graphicx,amsmath,amssymb,
bbm,amsfonts,mathrsfs,enumerate,color,url}
\usepackage{algorithm}
\usepackage{fullpage}
\usepackage{amsthm}
\usepackage{caption} 
\usepackage{algorithmic}
\usepackage{natbib}
\usepackage{chapterbib}

\newtheorem{thm}{Theorem}
\newtheorem{lem}[thm]{Lemma}
\newtheorem{prop}[thm]{Proposition}
\newtheorem{cor}[thm]{Corollary}

\newcommand{\abs}[1]{\left| #1\right|}

\newcommand{\norms}[1]{\| #1 \|} 
\newcommand{\ind}{\mathbbm{1}}
\newcommand{\pr}{\mathbb{P}}
\newcommand{\R}{\mathbb{R}}
\newcommand{\E}{\mathbb{E}}

\newcommand{\floor}[1]{\left\lfloor #1 \right\rfloor}

\newcommand{\Var}{\mathrm{Var}}

\newcommand{\argmin}[1]{\underset{#1}{\operatorname{arg}\operatorname{min}}\;}
\newcommand{\sgn}{\mathrm{sgn}}

\newcommand{\eqdist}{\stackrel{d}{=}}
\newcommand{\inprob}{\stackrel{p}{\to}}
\newcommand{\vertiii}[1]{{\left\vert\kern-0.25ex\left\vert\kern-0.25ex\left\vert #1 
    \right\vert\kern-0.25ex\right\vert\kern-0.25ex\right\vert}}

\newcommand{\mb}{\mathbf}
\newcommand{\mbb}{\boldsymbol}

\def\ci{\perp\!\!\!\perp}

\title{Goodness of fit tests for high-dimensional linear models}

\author{Rajen D. Shah\thanks{Supported in part by the
    Forschungsintitut f\"ur Mathematik (FIM) at ETH Z\"urich.} \\ 
University of Cambridge
\and Peter B\"uhlmann \\ 
ETH Z\"urich}

\begin{document}
\maketitle
\begin{cbunit}

\begin{abstract}
In this work we propose a framework for constructing goodness of fit tests
in both low and high-dimensional linear models. We advocate applying
regression methods to the scaled residuals following either an ordinary
least squares or Lasso fit to the data, and using some proxy for prediction
error as the final test statistic. We call this family  Residual Prediction (RP) tests. We
show that simulation can be used to obtain the critical values for such
tests in the low-dimensional setting, and demonstrate using both
theoretical results and  extensive numerical studies that 
some form of the parametric bootstrap
can do the same when the high-dimensional linear 
model is under consideration. 
We show that RP tests can be used to test for significance of groups or
individual variables as special cases, and here they compare favourably
with state of the art methods, but we also argue that they can be
designed to test for as diverse model 
misspecifications as heteroscedasticity and nonlinearity. 
\end{abstract}

\section{Introduction}
High-dimensional data, where the number of variables may greatly exceed the number of observations, has become increasingly more prevalent across a variety of disciplines. While such data pose many challenges to statisticians, we now have a variety of methods for fitting models to high-dimensional data, many of which are based on the Lasso \citep{tibshirani96regression}; see
\citet{buhlmann2011statistics} for a review of some of the developments.

More recently, huge strides have been made in quantifying uncertainty about parameter estimates.  For the important special case of the high-dimensional linear model, frequentist $p$-values for individual parameters or groups of parameters can now be obtained through an array of different techniques \citep{wasserman2009high, meinshausen09pvalues, pb13, zhangzhang14, covtest14, optimalconf14, jamo13b, Meinshausen2014, Ning2014, voorman2014inference, Zhou2015}---see \citet{Dezeure2014} for an overview of some of these methods.
Subsampling techniques such as Stability Selection \citep{meinshausen2008ss} and its variant Complementary Pairs Stability Selection (CPSS) \citep{shah2013} can also be used to select important variables whilst preserving error control in a wider variety of settings.

Despite these advances, something still lacking from the practitioner's
toolbox is a corresponding set of diagnostic checks to help assess the
validity of, for example, the high-dimensional linear model. For instance,
there are no well-established methods for detecting heteroscedasticity in
high-dimensional linear models, or whether a nonlinear model may be more
appropriate.

In this paper, we introduce an approach for creating diagnostic measures or
goodness of fit tests that are sensitive to different sorts of departures
from the `standard' high-dimensional linear model.  
As the measures are
derived from examining the residuals following e.g.\ a Lasso fit to the
data, we use the name Residual Prediction (RP) tests. To the best of
  our knowledge, it is the first methodology for deriving confirmatory
  statistical conclusions, in terms of $p$-values, to test for a broad range
  of deviations from a high-dimensional linear model. In
Section~\ref{sec:contrib} we give a brief overview of the idea, but first
we discuss
what we mean by goodness of fit in a high-dimensional setting.

\subsection{Model misspecification in high-dimensional linear models}
Consider the Gaussian linear model
\begin{equation} \label{eq:lin_mod}
 \mb y = \mb X \mbb\beta + \sigma\mbb\varepsilon,
\end{equation}
where $\mb y \in \R^n$ is a response vector, $\mb X \in \R^{n \times p}$ is the
fixed design matrix, $\mbb\beta \in \R^p$ is the unknown vector of
coefficients, $\mbb\varepsilon \sim \mathcal{N}_n(\mb 0, \mb I)$ is a vector of
uncorrelated Gaussian errors, and $\sigma^2 >0$ is the variance of the
noise. In the low-dimensional situation where $p<n$, we may speak of
\eqref{eq:lin_mod} being misspecified such that $\E(\mb y)\neq \mb
X\mbb\beta$. When $\mb X$ has full row rank however, any vector in $\R^n$
can be expressed as $\mb X\mbb\beta$ for some $\mbb\beta \in \R^p$, leaving
in general no room for nonlinear alternatives. When
  restricting to sparse linear models specified by
  \eqref{eq:lin_mod}, the situation is different though and
  misspecification can happen \citep{pbvdg15}; we will take a sparse
  Gaussian linear model as our null hypothesis (see also Theorems~\ref{thm:maximise_pval} 
and \ref{thm:single_pval_null}).
We discuss an approach to handle a relaxation of the Gaussian error assumption in Section~\ref{sec:non-Gaussian} of the supplementary material.

When there is no good sparse approximation to $\mb X\mbb\beta$,  
a high-dimensional linear model may not be an appropriate model for the
data-generating process: a sparse nonlinear model might be more
interpretable and may generalise better, for example.
Moreover, the Lasso
and other sparse estimation procedures may have poor performance, undermining the
various different high-dimensional inference methods mentioned above that
make use of them. 
Our proposed RP tests investigate whether the Lasso is
a good estimator of the signal.


%

\subsection{Overview of Residual Prediction (RP) tests and main contributions}
\label{sec:contrib}
Let $\hat{\mb R}$ be the residuals following a Lasso fit to $\mb X$. If $\mb X \mbb\beta$ is such that it can be well-estimated by the Lasso, then the residuals should contain very little signal and instead should behave roughly like the noise term $\sigma\mbb\varepsilon$.
On the other hand, if the signal is such that the Lasso performs poorly and instead a nonlinear model were more appropriate, for example, some of the (nonlinear) signal should be present in the residuals, as the Lasso would be incapable of fitting to it.

Now if we use a regression procedure that is well-suited to predicting the
nonlinear signal (an example may be Random Forest
\citep{breiman01random}), applying this to the residuals and computing the
resulting mean residual sum of squares (RSS) or any other proxy for
prediction error will give us a test statistic that under the null hypothesis
of a sparse linear model
we expect to be relatively large, and under the alternative we expect to be
relatively small. 
Different regression procedures applied to the residuals can be used to
test for different sorts of departures from the Gaussian linear model. 
Thus RP tests consist of three components.
\begin{enumerate}
\item An initial procedure that regresses $\mb y$ on $\mb X$ to give a set
  of residuals; this is typically the Lasso if $p>n$ or could be ordinary least squares if $\mb X$ is low-dimensional.
\item A \emph{residual prediction method} (RP method) that is suited to predicting the particular signal expected in the residuals under the alternative(s) under consideration.
\item Some measure of the predictive capability of the RP method. Typically this would be the residual sum of squares (RSS), but in certain situations a cross-validated estimate of prediction error may be more appropriate, for example.
\end{enumerate}
We will refer to the composition of a residual prediction method and an estimator of prediction error as a \emph{residual prediction function}.
This must be a (measurable) function $f$ of the residuals and all available
predictors, $p_{\text{all}}$ of them in total, to the reals $f : \R^n \times \R^{n \times p_{\text{all}}} \to \R$.
For example, if rather than testing for nonlinearity, we wanted to
ascertain whether any additional variables 
were significant
after accounting for those in $\mb X$, we could consider the mean RSS after
regressing the residuals on a matrix of predictors containing both $\mb X$
and the additional
variables, using the Lasso. 
If the residuals can be predicted better than one would expect under the
null hypothesis with a model as in \eqref{eq:lin_mod}, this provides
evidence against the null.

Clearly in order to use RP tests to perform formal hypothesis tests, one needs knowledge of the distribution of the test statistic under the null, in order to calculate $p$-values. Closed form expressions are difficult if not impossible to come by, particularly when the residual prediction method is something as intractable as Random Forest.

In this work, we show that under certain conditions, the parametric bootstrap \citep{Efron1994} can be used, with some modifications, to calibrate \emph{any} RP test. Thus the RP method can be as exotic as needed in order to detect the particular departure from the null hypothesis that is of interest, and there are no restrictions requiring it to be a smooth function of the data, for example. In order to obtain such a general result, the conditions are necessarily strong; nevertheless, we demonstrate empirically that for a variety of interesting RP tests, bootstrap calibration tends to be rather accurate even when the conditions cannot be expected to be met.
As well as providing a way of calibrating RP tests, we also introduce a framework for combining several RP tests in order to have power against a diverse set of alternatives. 

Although formally the null hypothesis tested by our approach is that of the sparse Gaussian linear model \eqref{eq:lin_mod}, an RP test geared towards nonlinearity is unlikely to reject purely due to non-Gaussianity of the errors, and so the effective null hypothesis typically allows for more general error distributions. By using the nonparametric bootstrap rather than the parametric bootstrap, we can allow for non-Gaussian error distributions more explicitly. We discuss this approach in Section~\ref{sec:non-Gaussian} of the supplementary material, where we see that type I error is very well controlled even in settings with $t_3$ and exponential errors.

Some work related to ours here is that of \citet{Chatterjee2010},
\citet{Chatterjee2011}, \citet{Camponovo2014} and \citet{Zhou2014} who study the use of the
bootstrap with the (adaptive) Lasso for constructing confidence sets
for the regression coefficients. Work that is more closely aligned to
our aim of creating diagnostic measures for high-dimensional models is that
of \citet{Nan2014}, though their approach is specifically geared towards variable selection and they do not provide theoretical guarantees within a hypothesis testing framework as we do.

\subsection{Organisation of the paper}
Simulating the residuals under the null is particularly simple when rather than using the Lasso residuals, ordinary least squares residuals are used.
We study this simple situation in Section~\ref{sec:OLS} not only to help motivate our approach in the high-dimensional setting, but also to present what we believe is a useful method in its own right. In Section~\ref{sec:Combine} we explain how several RP tests can be aggregated into a single test that combines the powers of each of the tests.
In Section~\ref{sec:Lasso} we describe the use of RP tests in the high-dimensional setting, and prove the
validity of a calibration procedure based on the parametric bootstrap. We give several applications of RP tests in Section~\ref{sec:apps} along with the results of extensive numerical experiments, and conclude with a discussion in Section~\ref{sec:discuss}. The supplementary material contains further discussion of the power of RP tests; 
a proposal for how to test null hypotheses of the form \eqref{eq:lin_mod} allowing for more general error distributions; additional numerical results; a short comment concerning the interpretation of $p$-values; and all of the proofs. The \texttt{R} \citep{R} package \texttt{RPtests} provides an implementation of the methodology.

\section{Ordinary least squares RP tests} \label{sec:OLS}
A simple but nevertheless important version of RP tests uses residuals from ordinary least squares (OLS) in the first stage. For this, we require $p < n$ in the set-up of \eqref{eq:lin_mod}. Let $\mb P$ denote the orthogonal projection on to the column space of $\mb X$. Then under the null hypothesis that the model \eqref{eq:lin_mod} is correct, the scaled residuals $\hat{\mb R}$ are
\[
 \hat{\mb R} :=\frac{(\mb I - \mb P) \mb y}{\| (\mb I - \mb P) \mb y\|_2} = \frac{(\mb I - \mb P) \mbb \varepsilon}{\| (\mb I - \mb P) \mbb\varepsilon\|_2},
\]
and so their distribution does not depend on any unknown parameters: they form an ancillary statistic.
Note that the scaling of the residuals eliminates the dependence on $\sigma^2$.
It is thus simple to simulate from the distribution of any function of the scaled residuals, and this allows critical values to be calculated for tests using any RP method.

We note that by using OLS applied to a larger set of variables as the RP method, and the RSS from the resulting fit as the estimate of prediction error, the overall test is equivalent to a partial $F$-test for the significance of the additional group of variables. To see this let us write $\mb Z \in \R^{n \times q}$ for an additional group of variables. 
Let $\mb P_{\text{all}}$ be the orthogonal projection on to all available predictors, that is projection on to $\mb X_{\text{all}} = (\mb X, \, \mb Z) \in \R^{n \times p_{\text{all}}}$, where $p_{\text{all}}=p+q$.
When the RP method is OLS regression of the scaled residuals on to $\mb
X_{\text{all}}$, the resulting RSS is 
\begin{align*}
 \|(\mb I - \mb P_{\text{all}})\hat{\mb R}\|_2^2 = \frac{\|(\mb I - \mb P_{\text{all}})(\mb I - \mb P) \mb y\|_2 ^2}{\| (\mb I - \mb P) \mb y\|_2 ^2} &= \frac{\|(\mb I - \mb P_{\text{all}}) \mb y\|_2 ^2}{\| (\mb I - \mb P) \mb y\|_2 ^2},
\end{align*}
since $(\mb I - \mb P_{\text{all}})\mb P=\mb 0$. We reject for small values of the quantity above, or equivalently large values of
\[
 \frac{\| (\mb P_{\text{all}} - \mb P) \mb y\|_2 ^2}{\|(\mb I - \mb P_{\text{all}}) \mb y\|_2 ^2} \times \frac{n-p_{\text{all}}}{p_{\text{all}}-p},
\]
which is precisely the $F$-statistic for testing the hypothesis in question.

An alternative way to arrive at the $F$-test is to first residualise $\mb Z$ with respect to $\mb X$ and define new variables $\tilde{\mb Z} = (\mb I - \mb P)\mb Z$. 
Let us write $\tilde{\mb P}$ for the orthogonal projection on to $\tilde{\mb Z}$.
Now if our RP method is OLS regression of the scaled residuals on to
$\tilde{\mb Z}$, we may write our RSS as 
\begin{align*}
 \|(\mb I - \tilde{\mb P})\hat{\mb R}\|_2^2 = \frac{\|(\mb I - \tilde{\mb P})(\mb I - \mb P) \mb y\|_2 ^2}{\| (\mb I - \mb P) \mb y\|_2 ^2} &= \frac{\|\{\mb I - (\mb P + \tilde{\mb P})\} \mb y\|_2 ^2}{\| (\mb I - \mb P) \mb y\|_2 ^2},
\end{align*}
the final equality following from the fact that the column spaces of $\mb X$ and $\tilde{\mb Z}$ and hence $\mb P$ and $\tilde{\mb P}$ are orthogonal. It is easy to see that $\mb P + \tilde{\mb P} = \mb P_{\text{all}}$, and so we arrive at the $F$-test once more.

We can use each of the two versions of the $F$-test above as starting
points for generalisation, where rather than using OLS as a prediction
method, we use other RP methods more tailored to specific alternatives of
interest. The distribution of the output of an RP method under the null hypothesis
  of a linear model can be computed via simulation as follows. For a given $B>1$ we generate independent $n$-vectors with i.i.d.\ standard normal components $\mbb\zeta^{(1)}, \ldots, \mbb\zeta^{(B)}$. From these we form scaled residuals
\begin{equation} \label{eq:OLS_RP}
\hat{\mb R}^{(b)} = \frac{(\mb I-\mb P)\mbb\zeta^{(b)}}{\|(\mb I-\mb P)\mbb\zeta^{(b)}\|_2}.
\end{equation}
Let $\mb X_{\text{all}}$ be the full matrix of predictors.
Writing the original scaled residuals as $\hat{\mb R}$ we apply our chosen RP function $f$ to all of the scaled residuals to obtain a $p$-value
\begin{equation} \label{eq:MC_pval}
\frac{1}{B+1}\bigg(1 + \sum_{b=1}^B \ind_{\{f(\hat{\mb R}^{(b)}, \mb X_{\text{all}}) \leq f(\hat{\mb R}, \mb X_{\text{all}})\}}\bigg).
\end{equation}
See also Section \ref{sec:Lasso} for the extension to the case using
Lasso residuals.

Even in situations where the usual $F$-test may seem the natural choice, an
RP test with a carefully chosen RP method can often be more powerful
against alternatives of interest. This is particularly true when we
aggregate the results of various different RP methods to gain power over a
diverse set of alternatives, as we describe in the next section. 

\section{Aggregating RP tests} \label{sec:Combine}
In many situations, we would like to try a variety of different RP methods, in order to have power against various different alternatives. A key example is when an RP method involves a tuning parameter such as the Lasso. Each different value of the tuning parameter effectively gives a different RP method. One could also aim to create a generic omnibus test to test for, say, nonlinearity, heteroscedasticity and correlation between the errors, simultaneously.

To motivate our approach for combining the results of multiple RP tests, we consider the famous diabetes dataset of \citet{efron04least}. This has $p=10$ predictors measured for $n=442$ diabetes patients and includes a response that is a quantitative measure of disease progression one year after baseline. Given the null hypothesis of a Gaussian linear model, we wish to test for the presence of interactions and quadratic effects. In order to have power against alternatives composed of sparse coefficients for these effects, we consider as RP methods the Lasso applied to quadratic effects residualised with respect to the linear terms via OLS. We regress the OLS scaled residuals onto the transformed quadratic effects using the Lasso with tuning parameters on a grid of $\lambda$ values, giving a family of RP tests.

We plot the residual sums of squares from the Lasso fits to the scaled
residuals in Figure~\ref{fig:multi_lambda}, as a function of $\lambda$. 
Also shown are the residual sums of squares from Lasso fits to
scaled residuals simulated under the null hypothesis of a Gaussian
linear model, as simulation under the null hypothesis is the general
  principle which we use for deriving $p$-values.

\begin{figure}[!h]
\centering
\makebox{\includegraphics[scale=0.4]{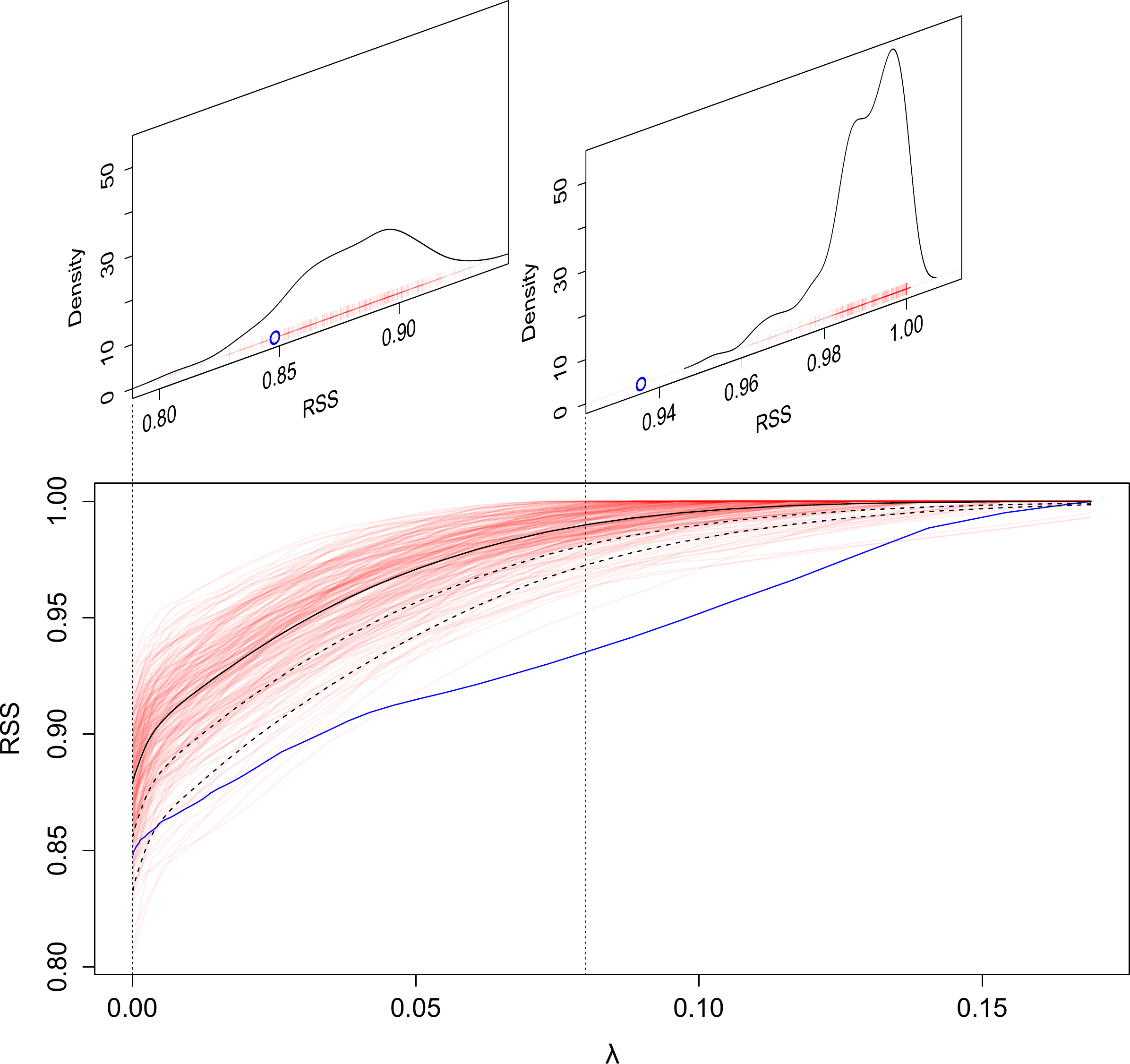}}
\caption{\label{fig:multi_lambda}Bottom plot: the residual sums of squares from Lasso fits to the original scaled residuals (blue) and simulated residuals (red), as well as the mean of the latter (black) and the mean displaced by one and two standard deviations (black, dotted). Top plots: kernel density plots for the simulated residual sums of squares at $\lambda=0$ (left) and $\lambda=0.8$ (right) with the original residual sums of squares in blue.}
\end{figure}

At the point $\lambda=0$, the observed RSS is not drastically smaller than
those of the simulated residuals, as the top left density plot
shows. Indeed, were we to calculate a $p$-value just based on the
$\lambda=0$ results corresponding to the $F$-test, we would obtain roughly 10\%.
The output at $\lambda=0.8$,
however, does provide compelling evidence against the null hypothesis, as
the top right density plot shows. Here the observed RSS is far to the left
of the support of the simulated residual sums of squares.
In order to create a $p$-value for the presence of interactions based on all of the output, we need a measure of how `extreme' the entire blue curve is, with respect to the red curves, in terms of carrying evidence against the null. Forming a $p$-value based on such a test statistic is straightforward, as we now explain.

Suppose we have residual 
prediction functions $f_l$, $l=1, \ldots, L$ (in our example these would be the RSS when using the Lasso with tuning parameter $\lambda_l$) and their evaluations on the true scaled residuals $\hat{\mb R}^{(0)} := \hat{\mb R}$ and simulated scaled residuals $\{\hat{\mb R}^{(b)}\}_{b=1}^B$. Writing $f_l^{(b)}=f_l(\hat{\mb R}^{(b)}, \mb X_{\text{all}})$, let $\mb f^{(b)}= \{f_l^{(b)}\}_{l=1}^L$ be the curve or vector of RP function evaluations at the $b$th scaled residuals, and denote by $\mb f^{(-b)} = \{\mb f^{(b')}\}_{b'\neq b}$ the entire collection of curves, potentially including the curve for the true scaled residuals $\mb R^{(0)}$, but excluding the $b$th curve. Let
\begin{align*}
\tilde{Q}:\; & \R^L \times \R^{L \times B}  \to \R \\
& (\mb f^{(b)}, \mb f^{(-b)}) \mapsto \tilde{Q}(\mb f^{(b)}, \mb f^{(-b)})
\end{align*}
be any measure of how extreme the curve $\mb f^{(b)}$ is compared to the rest of the curves $\mb f^{(-b)}$ (larger values indicating more extreme). Here $\tilde{Q}$ can be any function such that $\tilde{Q}_b := \tilde{Q}(\mb f^{(b)}, \mb f^{(-b)})$ does not depend on the particular ordering of the curves in $\mb f^{(-b)}$; we will give a concrete example below. We can use the $\{\tilde{Q}_b\}_{b\neq 0}$ to calibrate our test statistic $\tilde{Q}_0$  as detailed in the following proposition.
\begin{prop} \label{prop:Combine}
Suppose the simulated scaled residuals are constructed as in \eqref{eq:OLS_RP}.
Setting
\begin{equation*}
 Q=\frac{1}{B+1} \bigg(1+\sum_{b=1}^B \ind_{\{\tilde{Q}_b \geq \tilde{Q}_0\}}\bigg),
\end{equation*}
we have that under the null hypothesis \eqref{eq:lin_mod}, $\pr(Q\leq x) \leq x$ for all $x \in [0,1]$, so $Q$ constitutes a valid $p$-value.
\end{prop}
The result above is a straightforward consequence of the fact that under the null $\{\tilde{Q}_b\}_{b=0}^B$ form an exchangeable sequence, and standard results on Monte Carlo testing (see \citet{Davison1997} Ch.\ 4 for example). Under an alternative, we expect $\tilde{Q}_0$ to be smaller and $\tilde{Q}_b$ for $b \geq 1$ to be larger on average, than under the null. Thus this approach will have more power than directly comparing $\tilde{Q}_0$ to a sample from its null distribution.

We recommend constructing $\tilde{Q}$ as follows. Let $\hat{\mu}^{(-b)}_l$ and $\hat{\sigma}^{(-b)}_l$ respectively be the empirical mean and standard deviation of $\{f_l^{(b')}\}_{b' \neq b}$. We then set
\begin{equation} \label{eq:tildeQ}
\tilde{Q}_b = \max_{l}\,\{ (\hat{\mu}^{(-b)}_l - f_l^{(b)}) / \hat{\sigma}^{(-b)}_l\},
\end{equation}
the number of standard deviations by which the $b$th curve lies below the rest of the curves, maximised along the curve.
The intuition is that were $f_l^{(1)}$ to have a Gaussian distribution under the null for each $l$, $\Phi\{(f_l^{(0)} - \hat{\mu}^{(-0)}_l) / \hat{\sigma}^{(-0)}\}$ would be an approximate $p$-value based on the $l$th RP function, whence $\Phi(\tilde{Q}_0)$ would be the minimum of these $p$-values. Though it would be impossible to match the power of the most powerful test for the alternative in question (perhaps that corresponding to $\lambda=0.8$ in our diabetes example) among the $L$ tests considered, one would hope to come close. We stress however that this choice of $\tilde{Q}$ \eqref{eq:tildeQ} yields valid $p$-values regardless of the distribution of $f_l^{(1)}$ under the null.

Using this approach with a grid of $L=100$ $\lambda$ values,
we obtain a $p$-value of under 1\% for the diabetes example. As discussed in Section~\ref{sec:contrib}, this low $p$-value is unlikely to be due to a deviation from Gaussian errors, and indeed when we take our simulated errors $\mbb{\zeta}^{(b)}$ to be resamples from the vector of residuals (see Section~\ref{sec:non-Gaussian} of the supplementary material), we also obtain a $p$-value under 1\%; clear evidence that a model including only main effects is inappropriate for the data.
Further simulations demonstrating the power of this approach are presented in Section~\ref{sec:apps}.

\section{Lasso RP tests} \label{sec:Lasso}
When the null hypothesis is itself high-dimensional, we can use
Lasso residuals in the first stage of the RP testing
procedure. 
Although unlike scaled OLS residuals, scaled Lasso residuals are not ancillary, we will see that under certain conditions, the distribution of scaled Lasso residuals are not wholly sensitive to the parameters $\mbb\beta$ and $\sigma$ in \eqref{eq:lin_mod}.

Let us write $\hat{\mb R}_{\lambda}(\mbb\beta, \sigma\mbb\varepsilon)$ for
the scaled Lasso residuals when the tuning parameter is $\lambda$ (in square-root parametrisation, see below):
\begin{align}
\hat{\mbb{\beta}}_\lambda (\mbb\beta, \sigma\mbb\varepsilon) &\in \argmin{\mb b \in \R^p} \{\|\mb X (\mbb\beta - \mb b) + \sigma \mbb\varepsilon\|_2/\sqrt{n} + \lambda\|\mb b\|_1\} \label{eq:sqrt_lasso}\\
\hat{\mb R}_{\lambda}(\mbb\beta, \sigma\mbb\varepsilon) &= \frac{\mb X \{\mbb\beta - \hat{\mbb\beta}(\mbb\beta, \sigma\mbb\varepsilon) \}+ \sigma \mbb\varepsilon}{\|\mb X \{\mbb\beta - \hat{\mbb\beta}(\mbb\beta, \sigma\mbb\varepsilon) \}+ \sigma \mbb\varepsilon\|_2} \notag.
\end{align} 
Note that under \eqref{eq:lin_mod} $\hat{\mbb{\beta}}_\lambda (\mbb\beta,
  \sigma\mbb\varepsilon) \in \argmin{\mb b \in \R^p} \{\|\mb y - \mb X
  \mb b\|_2/\sqrt{n} + \lambda \|\mb b\|_1\}$ and $\hat{\mb R}_{\lambda}(\mbb\beta,
  \sigma\mbb\varepsilon) = \{\mb y - \mb X \hat{\mbb\beta}(\mbb\beta,
    \sigma\mbb\varepsilon)\}/\|\mb y- \mb X\hat{\mbb\beta}(\mbb\beta,
    \sigma\mbb\varepsilon)\|_2.$ 
Sometimes we will omit the first argument of $\hat{\mbb\beta}$ for convenience in which case it will always be the true parameter value under the null, $\mbb\beta$.
Here we are using the Lasso in the square-root parametrisation \citep{Belloni2011, Sun2012} rather than the conventional version where the term in the objective assessing the model fit would be $\|\mb X (\mbb\beta - \mb b) + \sigma \mbb\varepsilon\|_2^2$. We note that the two versions of the Lasso have identical solution paths but these will simply be parametrised differently. For this reason, we will simply refer to \eqref{eq:sqrt_lasso} as the Lasso solution.
Note that while the Lasso solution may potentially be non-unique, the residuals are always uniquely defined as the fitted values from a Lasso fit are unique (see \citet{Tibshirani2013}, for example).
Throughout we will assume that the columns of $\mb X$ have been
scaled to have $\ell_2$-norm $\sqrt{n}$.

We set out our proposal for calibrating RP tests based on Lasso residuals using the parametric bootstrap in Algorithm~\ref{alg:Lasso1} below.
\begin{algorithm}
\caption{Lasso RP tests}
\label{alg:Lasso1}
\begin{enumerate}
\item Let $\check{\mbb\beta}$ be an estimate of $\mbb\beta$, typically a Lasso estimate selected by cross-validation.
\item Set $\check{\sigma}=\|\mb y - \mb X \check{\mbb\beta}\|_2/\sqrt{n}$.
\item Form $B$ scaled simulated residuals $\{\hat{\mb
    R}_\lambda(\check{\mbb\beta}, \check{\sigma}
  \mbb\zeta^{(b)})\}_{b=1}^B$ where the $\mbb\zeta^{(b)}$ are i.i.d.\ draws
  from $\mathcal{N}_n(\mb 0, \mb I)$, and $\lambda$ chosen according to the proposal of \citet{Sun2013}.

\item Based on the scaled simulated residuals $\{\hat{\mb
    R}_\lambda(\check{\mbb\beta}, \check{\sigma}
  \mbb\zeta^{(b)})\}_{b=1}^B$, compute a $p$-value \eqref{eq:MC_pval} or use these to form an aggregated $p$-value as described in Section~\ref{sec:Combine}.
\end{enumerate}
\end{algorithm}
In the following section we aim to justify the use of the parametric bootstrap from a theoretical perspective and also discuss the particular choices $\check{\mbb\beta}, \check{\sigma}$ and $\lambda$ used above.

\subsection{Justification of Lasso RP tests}
Given $\mb b \in \R^p$ and a set $A \subseteq \{1, \ldots, p\}$, let $\mb b_A$ be the subvector of $\mb b$ with components consisting of those indexed by $A$. Also for a matrix $\mb M$, let $\mb M_A$ be the submatrix of $\mb M$ containing those columns indexed by $A$, and let $\mb M_k=\mb M_{\{k\}}$, the $k$th column.
The following result shows that
if $\sgn(\check{\mbb\beta}) = \sgn(\mbb\beta)$, with the sign function understood as being applied componentwise,
we have partial ancillarity of the scaled residuals. In the following we let $S = \{j:\beta_j \neq 0\}$ be the support set of $\mbb\beta$.
 \begin{thm} \label{thm:lambda_fixed}
  Suppose $\check{\mbb\beta}$ is such that $\sgn(\check{\mbb\beta}) =
  \sgn(\mbb\beta)$. For $t \in [0, 1)$ and $\lambda>0$, consider the
    deterministic set
  \[
  \Lambda_{\lambda, t} = \{\mbb\zeta \in \R^n : \sgn\big(\hat{\mbb\beta}_{\lambda, S}(\mbb\beta, \sigma\mbb\zeta)\big) = \sgn(\mbb\beta_S) \text{ and } \min_{j \in S} \hat{\beta}_j(\mbb\beta, \sigma\mbb\zeta) / \beta_j > t \}.
  \]
  Then we have that for all $\mbb \zeta \in \Lambda_{\lambda, t}$,
  $\hat{\mb R}_\lambda(\mbb \beta, \sigma\mbb \zeta) =\hat{\mb R}_\lambda(\check{\mbb\beta}, \check{\sigma}\mbb \zeta)$
 provided 
$0 < \check{\sigma}/\sigma < \min_{j \in S} \check{\beta}_j / \{(1-t) \beta_j\}$.
 \end{thm}
In words, provided the error $\mbb \zeta$ is in the set
$\Lambda_{\lambda,t}$ and conditions for $\check{\mbb\beta}$ and
$\check{\sigma}$ are met, the scaled residuals from a Lasso fit to
$\mb y= \mb X \mbb\beta + \sigma\mbb \zeta$ are precisely equal to the scaled residuals from a Lasso fit to $\mb X \check{\mbb\beta} + \check{\sigma}\mbb \zeta$.
Note that all of the quantities in the result are deterministic.  
Under reasonable conditions and for a sensible choice of $\lambda$ (see
Theorem~\ref{thm:maximise_pval}), 
when 
$\mbb\zeta \sim \mathcal{N}_n(\mb 0, \mb I)$, we can expect the event
$\mbb\zeta \in \Lambda_{\lambda, t}$ to have large probability. Thus Theorem~\ref{thm:lambda_fixed} shows that
the scaled residuals are not very sensitive to the parameter $\sigma$ or to
the magnitudes of the components of $\mbb\beta$, but instead depend largely
on the signs of the latter. It is the square-root parametrisation that
allows the result to hold for a large range of values of
$\check{\sigma}$, and in particular for all $\check{\sigma}$ sufficiently
small.  

 
Theorem~\ref{thm:lambda_fixed} does not directly justify a way to
  simulate from the distribution of the scaled
Lasso residuals as in Algorithm \ref{alg:Lasso1} since the sign
  pattern of $\check{\mbb\beta}$ must equal that of $\mbb \beta$.
Accurate estimation of the sign pattern of $\mbb\beta$ using the
Lasso requires a strong irrepresentable or neighbourhood stability
condition \citep{meinshausen04consistent, zhao05model}. 
Nevertheless, we now show that we can modify Algorithm~\ref{alg:Lasso1} to
yield provable error control under more reasonable conditions. In Section~\ref{sec:heuristic} we argue heuristically that the same error control
should hold for Algorithm~\ref{alg:Lasso1} in a wide range of settings. 

\subsubsection{Modified Lasso RP tests}
Under a so-called beta-min condition (see Theorem~\ref{thm:maximise_pval} below), with high probability we can arrive at an initial estimate of $\mbb\beta$, $\mbb\beta'$ via the Lasso for which $\sgn(\mbb\beta'_S) = \sgn(\mbb\beta_S)$, and where $\min_{j \in S} |\beta'_j| > \max_{j \in S^c} |\beta'_j|$. 
With such a $\mbb\beta'$, we can aim to seek a threshold $\tau$ for which the Lasso applied only on the subset of variables $\mb X_{S_\tau}$ where $S_\tau:= \{j : |\beta'_j| > \tau\}$ yields an estimate 
that has the necessary sign agreement with $\mbb\beta$. This then motivates Algorithm~\ref{alg:Lasso2} based on maximising over the candidate $p$-values obtained through different $\mbb\beta$ estimates derived from applying the Lasso to different subsets of the initial active set (see also \citet{Chatterjee2011} which introduces a related scheme). 

\begin{algorithm}
\caption{Modified Lasso RP tests (only used for Theorem
  \ref{thm:maximise_pval})}
\label{alg:Lasso2}
\begin{enumerate}
\item Let $\mbb\beta' = \hat{\mbb\beta}_\lambda(\sigma\mbb\varepsilon)$ be the Lasso estimate of $\mbb\beta$.
\item Let $s'=|\{j:\beta'_j \neq 0\}|$ and suppose $0<|\beta'_{j_s'}|\leq
  \cdots \leq |\beta'_{j_{1}}|$ are the non-zero components of $\mbb\beta'$
  arranged in order of non-decreasing magnitude. Define $\hat{S}^{(k)} =
  \{j_1,j_2, \ldots, j_k\}$.
\item For $k=1, \ldots, s'$ let $\check{\mbb\beta}^{(k)}$ be the Lasso estimate from regressing $\mb y$ on $\mb X_{\hat{S}^{(k)}}$. Further set $\check{\mbb\beta}^{(0)}=\mb 0$.
\item Using each of the $\check{\mbb\beta}^{(k)}$ in turn and
  $\check{\sigma}^{(k)}=\|\mb y - \mb X\check{\mbb\beta}^{(k)}\|_2/\sqrt{n}$, generate sets of
  residuals $\{\hat{\mb R}_\lambda(\check{\mbb\beta}^{(k)}, \check{\sigma}^{(k)}
  \mbb\zeta^{(b)})\}_{b=1}^B$ where the $\mbb\zeta^{(b)}$ are i.i.d.\ draws
  from $\mathcal{N}_n(\mb 0, \mb I)$. Use these to create corresponding $p$-values $Q_k$
  for RP tests based on \eqref{eq:MC_pval} or the method introduced in
  Section~\ref{sec:Combine}.
\item Output $Q=\max_{k=0,\ldots,s'} Q_k$ as the final approximate
  $p$-value. 
\end{enumerate}
\end{algorithm}
Note we do not recommend the use of Algorithm~\ref{alg:Lasso2} in practice; we only introduce it to facilitate theoretical analysis which sheds light on our proposed procedure Algorithm~\ref{alg:Lasso1}. Let $s=|S|$ and $s'=|\{j:\beta'_j \neq 0\}|$.
The theorem below gives conditions under which with high probability, $s' \geq s$ and residuals from responses generated around $\check{\mbb\beta}^{(s)}$ will equal the true residuals.  This then shows that the maximum $p$-value $Q$ will in general be a conservative $p$-value as it will always be at least as large as $Q_{s}$, on an event with high probability.

As well as a beta-min condition, the result requires some relatively mild assumptions on the design matrix. Let $\mathscr{C}(\xi, T)=\{\mb u:\|\mb u_{T^c}\|_1 \leq \xi\|\mb u_T\|_1, \, \mb u \neq \mb 0\}$. The restricted eigenvalue 
\citep{bickel07dantzig, Koltchinskii2009} is defined by
\begin{equation} \label{eq:RE}
 \phi(\xi) = \inf \bigg\{ \frac{\|\mb X \mb u\|_2 / \sqrt{n}}{\|\mb u\|_2}:  \mb u \in \mathscr{C}(\xi, S) \bigg\}.
\end{equation}
For a matrix $\mb M \in \R^{n \times p}$, $T \subset \{1,\ldots,p\}$ and $\xi>1$, the compatibility factor $\kappa(\xi, T, \mb M)$ \citep{van2009conditions} is given by
\begin{equation} \label{eq:compatibility}
\kappa(\xi, T, \mb M)=\inf \bigg\{\frac{\|\mb M \mb u\|_2/\sqrt{n}}{\|\mb u_T\|_1/|T|} : \mb u \in \mathscr{C}(\xi, T)\bigg\}.
\end{equation}
When either of the final two arguments are omitted, we shall take them to be $S$ and $\mb X$ respectively; the more general form is required in Section~\ref{sec:single}.
The sizes of $\kappa(\xi)$ and $\phi(\xi)$ quantify the ill-posedness of the the design matrix $\mb X$; we will require $\kappa^(\xi), \phi(\xi) > 0$ for some $\xi > 1$.
Note that in the random design setting where the rows of $\mb X$ are i.i.d.\ multivariate normal with the minimum eigenvalue of the covariance matrix bounded away from zero, the factors \eqref{eq:RE} and \eqref{eq:compatibility} can be thought of as positive constants in asymptotic regimes where $s\log(p)/n \to 0$. We refer the reader to \citet{van2009conditions} and \citet{Zhang2012} for further details.
\begin{thm} \label{thm:maximise_pval}
Suppose the data follows the Gaussian linear model \eqref{eq:lin_mod}. Let $\lambda = A\sqrt{2\log(p/\eta)/n}$ with $A>\sqrt{2}$ and $pe^{-s-2}>\eta > 0$. Suppose for $\xi>1$ that
\begin{align}
\frac{s\log(p/\eta)}{n\kappa^2(\xi)} \leq \frac{1}{A^2 (\xi+1)} \min \bigg(1-\frac{\sqrt{2}(\xi+1)}{A(\xi-1)}, \,\, \frac{1}{5} \bigg). \label{eq:A_cond}
\end{align}
Assume a beta-min condition
\begin{equation}
\min_{j \in S} |\beta_j| > 10\sqrt{2}A\xi\frac{\sigma \sqrt{s\log(p/\eta)}}{\phi^2(\xi)\sqrt{n}}. \label{eq:betamin}
\end{equation}
Then for all $x \in [0,1]$,
\begin{equation} \label{eq:valid}
\pr(Q \leq x) \leq x + \frac{2(1+r_{n-s})\eta}{\sqrt{\pi\log(p / \eta)}} +e^{-n/8}
\end{equation}
where $r_m \to 0$ as $m \to \infty$.
\end{thm}
Although the beta-min condition, which is of the form is of the form $\min_{j \in S} |\beta_j| \geq \text{const.} \times \sqrt{s\log(p)/n}$, may be regarded as somewhat strong, the conclusion is correspondingly strong: any RP method or collection of RP methods with arbitrary $\tilde{Q}$ for combining tests can be applied to the residuals and the result remains valid.
It is also worth noting however that the conditions are only required under the null.
For example, if the alternative of interest was that an additional variable $\mb z$ was related to the response after accounting for those in the original design matrix $\mb X$, no conditions on the relationship between $\mb X$ and $\mb z$ are required for the test to be valid.

More importantly though, the conditions are certainly not necessary for the conclusion to hold. The scaled residuals are a function of the fitted values and the response, and do not involve Lasso parameter estimates directly. Thus whilst duplicated columns in $\mb X$ could be problematic for inferential procedures relying directly on Lasso estimates such as the debiased Lasso \citep{zhangzhang14}, they pose no problem for RP tests. In addition, given a particular RP method, exact equality of the residuals would not be needed to guarantee a result of the form \eqref{eq:valid}.  

\subsubsection{Relevance of Theorem~\ref{thm:maximise_pval} to Algorithm~\ref{alg:Lasso1}} \label{sec:heuristic}
In the special case of testing for the significance of a single predictor described above, we have a much stronger result than Theorem~\ref{thm:maximise_pval} (see Theorems~\ref{thm:single_pval_null} and \ref{thm:single_pval_alt})
which shows that neither the beta-min condition nor the maximisation over candidate $p$-values of Algorithm~\ref{alg:Lasso2} is necessary for error control to hold. More generally, in our experiments we have found $Q_{s'}$ is usually equal to or close to the maximum $Q$ for large $B$ across a variety of settings. Thus selecting $Q_{s'}$ rather than performing the maximisation (which amounts to Algorithm~\ref{alg:Lasso1}) is able to deliver conservative error control as evidenced by the simulations in Section~\ref{sec:apps}.

A heuristic explanation for why the error is controlled is that typically the amount of signal remaining in $\hat{\mb R}_\lambda(\check{\mbb\beta}^{(k)}, \check{\sigma}\mbb\zeta)$ increases with $k$, simply because typically $\|\mb X \check{\mbb\beta}^{(k)}\|_2$ also increases  with $k$. This can result in the prediction error of a procedure applied to the various residuals decreasing with $k$ because the signal-to-noise ratios tend to be increasing; thus the $p$-values tend to increase with $k$.

In addition, when the Lasso performs well, we would expect residuals to contain very little signal, and any differences in the signals contained in $\hat{\mb R}_\lambda(\mbb\beta, \sigma\mbb\zeta)$ and $\hat{\mb R}_\lambda(\check{\mbb\beta}, \check{\sigma}\mbb\varepsilon)$ to be smaller still, particularly when  $\mb{X}\mbb\beta$ and $\mb{X}\check{\mbb\beta}$ are close. Typically the RP function will be insensitive to such small differences since they are unlikely to be too close to directions against which power is desired.
We now discuss the choices of $\check{\mbb\beta}$, $\lambda$ and $\check{\sigma}$ in Algorithm~\ref{alg:Lasso1}.

\subsubsection{Practical considerations} \label{sec:practice}

\paragraph{Choice of $\check{\mbb\beta}$.}
In view of the preceding discussion, it suffices for $\check{\mbb\beta}$ to satisfy a screening-type property:  we would like the support of $\check{\mbb\beta}$ to contain that of $\mbb\beta$. Though Theorem~\ref{thm:maximise_pval} suggests a fixed $\lambda$, since $\check{\mbb\beta}$ only needs to be computed once, we can use cross-validation. This is the perhaps the most standard way of producing an estimate that performs well for screening (see for example Section 2.5.1 of \citet{buhlmann2011statistics}).

If the folds for cross-validation are chosen at random, the estimate will have undesirable randomness beyond that of the data. We thus suggest taking many random partitions into folds and using an estimate based on a $\lambda$ that minimises the cross-validation error curve based on all of the folds used. In our simulations in Section~\ref{sec:apps} we partition the observations into 10 random folds a total of 8 times.



\paragraph{Choice of $\check{\sigma}$.}
The normalised RSS is perhaps the most natural choice for $\check{\sigma}^2$ (see also \citet{Reid2013}), though as Theorem~\ref{thm:lambda_fixed} suggests, the results are essentially unchanged when this is doubled or halved, for example.


\paragraph{Choice of $\lambda$ for the Lasso residuals.}
The choice of $\lambda$ should be such that with high probability, the resulting estimate contains the support of $\mbb\beta$ (see Theorem~\ref{thm:lambda_fixed}). 
Though Theorem~\ref{thm:maximise_pval} suggests taking $\lambda = A\sqrt{2\log(p)/n}$ for $A >\sqrt{2}$, the restriction on $A$ is an artefact of basing our result on oracle inequalities from \citet{Sun2012}, which place relatively simple conditions on the design. \citet{Sun2013} has a more involved theory which suggests a slightly smaller $\lambda$. We therefore use their method, the default in the \texttt{R} package \citet{scalreg}, as a convenient fixed choice of $\lambda$. 




\subsection{Testing the significance of individual predictors} \label{sec:single}
Here we consider the collection of null hypotheses $H_k :\beta_k=0$ and their corresponding alternatives that $\beta_k \neq 0$.  
Note that for this setting there are many other approaches that can perform the required tests.
Our aim here is to show that RP tests can be valid under weaker assumptions
than those laid out in Theorem~\ref{thm:maximise_pval}, and moreover that
the simpler approach of Algorithm~\ref{alg:Lasso1} can control type I
error.  

We begin with some notation. For $A_k := \{1,\ldots,p\}\setminus \{k\}$ and $\mb b \in \R^p$ let $\mb b_{-k}=\mb b_{A_k}$ and $\mb X_{-k} = \mb X_{A_k}$.
For each variable $k$, our RP method will be a least squares regression
onto a version of $\mb X_k$ that has been residualised with respect to $\mb
X_{-k}$.
Since in the high-dimensional setting $\mb X_{-k}$ will typically have full row rank, an OLS regression of $\mb X_k$ on $\mb X_{-k}$ will return the $\mb 0$-vector as residuals. Hence we will residualise $\mb X_k$ using the square-root Lasso:
\[
 \mbb\Psi_{k} = \argmin{\mb b \in \R^{p-1}} \{\|\mb X_k - \mb X_{-k} \mb b\|_2 /\sqrt{n} + \gamma\|\mb b\|_1\}.
\]
This RP method is closely related to the pioneering idea by 
  \citet{zhangzhang14} and similar to that of \citet{Ning2014}, who
consider using the regular Lasso (without the square-root parametrisation)
at each stage. If $\mb X_k$ were not residualised with respect to $\mb
X_{-k}$, and the regular Lasso were used, the resulting RP method would be similar to
that of \citet{voorman2014inference}. The work of \citet{Ren2015} studies an analogous procedure in the context of the Gaussian graphical model.

Let $\mb W_k$ be the residual $\mb X_k - \mb X_{-k}\mbb \Psi_{k}$. Note for each $k$ we may write
\[
\mb y = \mb X_{-k} \mbb\Theta_k + \beta_k \mb W_k + \sigma \mbb\varepsilon
\]
where $\mbb\Theta_k = \mbb\beta_{-k} + \beta_k \mbb\Psi_k \in \R^{p-1}$.
Let $\hat{\mbb\Theta}_k$ be the square-root Lasso regression of $\mb y$ on
to $\mb X_{-k}$ with tuning parameter $\lambda$. Our RP function will be
the RSS from OLS regression of the scaled Lasso residuals $(\mb y - \mb
X_{-k} \hat{\mbb\Theta}_k)/\|\mb y - \mb X_{-k} \hat{\mbb\Theta}_k\|_2$  on
to $\mb W_k$. Note this is an RP function even though it involves the
residualised version of $\mb X_k$, $\mb W_k$; the latter is simply a
function of $\mb X$. Equivalently, we can consider the test statistic
$T_k^2$ with $T_k$ 
defined by
\begin{align*}
T_k=\frac{\mb W_k^T(\mb y - \mb X_{-k} \hat{\mbb\Theta}_k)}{\|\mb W_k\|_2\|\mb y - \mb X_{-k} \hat{\mbb\Theta}_k\|_2/\sqrt{n}}.
\end{align*}
Note that $T_k$ is simply a regularised partial correlation between $\mb y$ and $\mb X_k$ given $\mb X_{-k}$.
The bootstrap version is
\begin{align*}
T^*_k=\frac{\mb W_k^T(\mb y^*_k - \mb X_{-k} \hat{\mbb\Theta}^*_k)}{\|\mb W_k\|_2\|\mb y^*_k - \mb X_{-k} \hat{\mbb\Theta}^*_k\|_2/\sqrt{n}},
\end{align*}
where $\mb y^*_k = \mb X_{-k}\hat{\mbb\Theta}_k +
\check{\sigma}\mbb\varepsilon^*$, $\mbb\varepsilon^* \sim \mathcal{N}_n(\mb
0, \mb I)$ and $ \hat{\mbb\Theta}^*_k$ is the Lasso regression of $\mb y^*$
on $\mb X_{-k}$. Here we will consider taking $\check{\sigma}=\|\mb y - \mb
X \hat{\mbb\beta}\|_2/\sqrt{n}$ where $\hat{\mbb\beta}$ is the square-root
Lasso regression of $\mb y$ on the full design matrix $\mb X$. 

As before, let $S$ be the support of $\mbb\beta$, which without loss of generality we will take to be $\{1,\ldots,s\}$, and also let $N=\{1,\ldots,p\} \setminus S$ be the set of true nulls. Assume $\mbb\varepsilon \sim \mathcal{N}_n(\mb 0, \mb I)$.
The following result shows that only a relatively mild compatibility condition is needed in order to ensure that the type I error is controlled. We consider an asymptotic regime with $n \to \infty$ where $\mbb\beta, \mb X$ and $p$ are all allowed to vary with $n$ though we suppress this in the notation. In the following we denote the cumulative distribution function of the standard normal by $\Phi$.
\begin{thm} \label{thm:single_pval_null}
Let $\lambda =A_1\sqrt{2\log(p)/n}$ for some constant $A_1>1$ and suppose that 
$s\sqrt{\log(p)^2/n}/\kappa^2(\xi, S)\to 0$
for some $\xi > (A_1+1)/(A_1-1)$. Let $\gamma=A_2\sqrt{2\log(p)/n}$ for some constant $A_2>0$. Define $\mathscr{B}=\{\mb b\in \R^p :\mb b_N=0\}$. Then
\begin{align}
\sup_{k \in N,\, \mbb\beta \in \mathscr{B}, \, x \in \R} |\pr(T_k \leq x) - \Phi(x)| &\to 0, \notag \\
\sup_{k \in N,\, \mbb\beta \in \mathscr{B}, \, x \in \R} |\pr(T^*_k \leq x | \mbb\varepsilon) - \Phi(x)| &\inprob  0. \notag
\end{align}
\end{thm}
We see that a bootstrap approach can control the type I error uniformly across the noise variables and $\mbb\beta \in \mathscr{B}$. We note that this result is for a fixed design $\mb X$ and does not require any sparsity assumptions on the inverse covariance matrix of a distribution that could have generated the rows of $\mb X$ \citep{Ning2014}, for example.  
\begin{thm} \label{thm:single_pval_alt}
Let $\lambda$ and $\gamma$ be as in Theorem~\ref{thm:single_pval_null}. Assume that for some $\xi$ with $\xi > (A_1+1)/(A_1-1)$ there is a sequence of sets $\{1,\dots,s-1\} \subseteq T \subset \{1,\ldots,p-1\}$ such that $|T|\sqrt{\log(p)^2/n}/\kappa^2(\xi, T, \mb X_{-k})\to 0$ and $\sqrt{\log(p)} \|\mbb\Theta_{k,T^c}\|_1 \to 0$ where $T^c= \{1,\ldots,p-1\}\setminus T$.
Further assume that $\beta_k \|\mb W_k\|_2/\sqrt{n} \to 0$. Define  $\mathscr{B}_k =\{ \mb b \in \mathscr{B} : b_k=\beta_k\}$. Then
\begin{align}
\sup_{\mbb\beta \in \mathscr{B}_k, x \in \R} \abs{\pr(T_k \leq x) - \Phi\Big(x - \beta_k\|\mb W_k\|_2\big/\sqrt{\sigma^2 + \beta_k^2\|\mb W_k\|_2^2/n}\Big)} &\to 0, \notag \\
\sup_{\mbb\beta \in \mathscr{B}_k, x \in \R} |\pr(T^*_k \leq x | \mbb\varepsilon) - \Phi(x)| &\inprob  0. \notag
\end{align}
\end{thm}
If $\mbb\Psi_k$ and hence $\mbb\Theta_k$ were sparse, we could take $T$ as the set of nonzeroes and the second condition involving $\|\mbb\Theta_{k,T^c}\|_1$ would be vacuous. This would be the case with high probability in the random design setting where $\mb X$ has i.i.d.\ Gaussian rows with sparse inverse covariance matrix \citep{optimalconf14}.
However, $\mbb\Theta_{k, T^c}$ can also have many small coefficients provided they have small $\ell_1$-norm.
The result above shows that the power of our method is comparable to the
proposals of \citet{zhangzhang14} and \citet{optimalconf14} based on the
debaised Lasso. If $\|\mb W_k\|_2=O(\sqrt{n})$ as would typically be the
case in the random design setting discussed above,
  we would have power tending to 1 if $\beta_k\to 0$ but
$\sqrt{n}|\beta_k|\to \infty$. 
Further results on power to detect nonlinearities are given in
  Section~\ref{sec:power} of the supplementary material.

The theoretical results do not suggest any real benefit from using the bootstrap as to test hypotheses we can simply compare $T_k$ to a standard normal distribution. However our experience has been that this can be slightly anti-conservative in certain settings. Instead, we propose to use the bootstrap to estimate the mean and standard deviation of the null distribution of the $T_k$ by computing the empirical mean $\hat{m}_k$ and standard deviation $\hat{v}_k$ of $B$ samples of $T_k^*$. Then we take as our $p$-values $2[1- \Phi\{|(T_k-\hat{m}_k)/\hat{v}_k|\}]$.

This construction of $p$-values appears to yield tests that very rarely have size exceeding their nominal level. Indeed in all our numerical experiments we found no evidence of this violation occurring. An additional advantage is that only a modest number of bootstrap samples is needed to yield the sort of low $p$-values that could fall below the threshold of a typical multiple testing procedure. We recommend choosing $B$ between 50 and 100.

\subsubsection{Computational considerations} \label{sec:single_comp}
Using our bootstrap approach for calibration presents a significant computational burden when it is applied to test for the significance of each of a large number of variables in turn. 
Some modifications to Algorithm~\ref{alg:Lasso1} can help to overcome this issue and allow this form of RP tests to be applied to typical high-dimensional data with large $p$.

Firstly rather than using cross-validation to choose $\lambda$ for computation of $\hat{\Theta}_k$, we recommend using the fixed $\lambda$ of \citet{Sun2013} (see also Section~\ref{sec:practice}). The tuning parameter $\gamma$ required to compute $\mb W_k$ can be chosen in the same way, and we also note that these nodewise regressions only need to be done once rather than for each bootstrap sample.
Great computational savings can be realised by first regressing $\mb y$ on
$\mb X$ to yield coefficients $\hat{\mbb\beta}$. Writing
$\hat{S}=\{k:\hat{\beta}_k\neq 0\}$, we know that for each $k \notin
\hat{S}$, $\hat{\mbb\Theta}_k = \hat{\mbb\beta}_{-k}$, so we only need to
compute $\hat{\mbb\Theta}_k$ for those $k$ in $\hat{S}$. The same logic can be
applied to computation of $\hat{\mbb\Theta}_k^*$ for the bootstrap
replicates.

We also remark that approaches for directly simulating Lasso estimates \citep{Zhou2014} may be used to produce simulated residuals. These have the potential to substantially reduce the computational burden; not just in the case of testing significance of individual predictors but for RP tests in general.

\section{Applications} \label{sec:apps}

\subsection{Low-dimensional nulls}
Here we return to the problem of testing for quadratic effects in the diabetes dataset used in the example of Figure~\ref{fig:multi_lambda}. 
In order to further investigate the power of the aggregate RP test constructed through Lasso regressions on a grid of 100 $\lambda$ values as described in Section~\ref{sec:Combine}, we created artificial signals from which we simulated responses. The signals (mean responses) were constructed by selecting at random $s$ of the quadratic terms and giving these coefficients generated using i.i.d.\ Unif$[-1, 1]$ random variables. The remaining coefficients for the variables were set to 0, so $s$ determined the sparsity level of the signal. Responses were generated by adding i.i.d.\ Gaussian noise to the signals, with variance chosen such that the $F$-test for the presence of quadratic effects has power 0.5 when the size is fixed at 0.05. We created 25 artificial signals at each sparsity level $s \in \{1, 4, 10, 20, 35, 54\}$. Note that the total number of possible quadratic effects was 54 (as one of the variables was binary), so the final sparsity level represents fully dense alternatives where we might expect the $F$-test to have good power. We note however that the average power of the $F$-test in the dense case rests critically on the form of the covariance between the generated quadratic coefficients, with optimality guarantees only in special circumstances (see Section 8 of \citet{goeman05testing}).
For the RP tests, we set the number of bootstrap samples $B$ to be 249.

We also compare the power of RP tests to the \emph{global test} procedure of \citet{goeman05testing}.
The results, shown in Figure~\ref{fig:F_test_diabetes}, suggest that RP
tests can outperform the $F$-test in a variety of settings, most notably
when the alternative is sparse, but also in dense settings. When there are
small effects spread out across many variables ($s \in \{35, 54\}$), the
global test tends to do best; indeed in such settings it is optimal. In the
sparser settings, RP tests perform better. 
\begin{figure}[!h]
\centering
\includegraphics[scale=0.4]{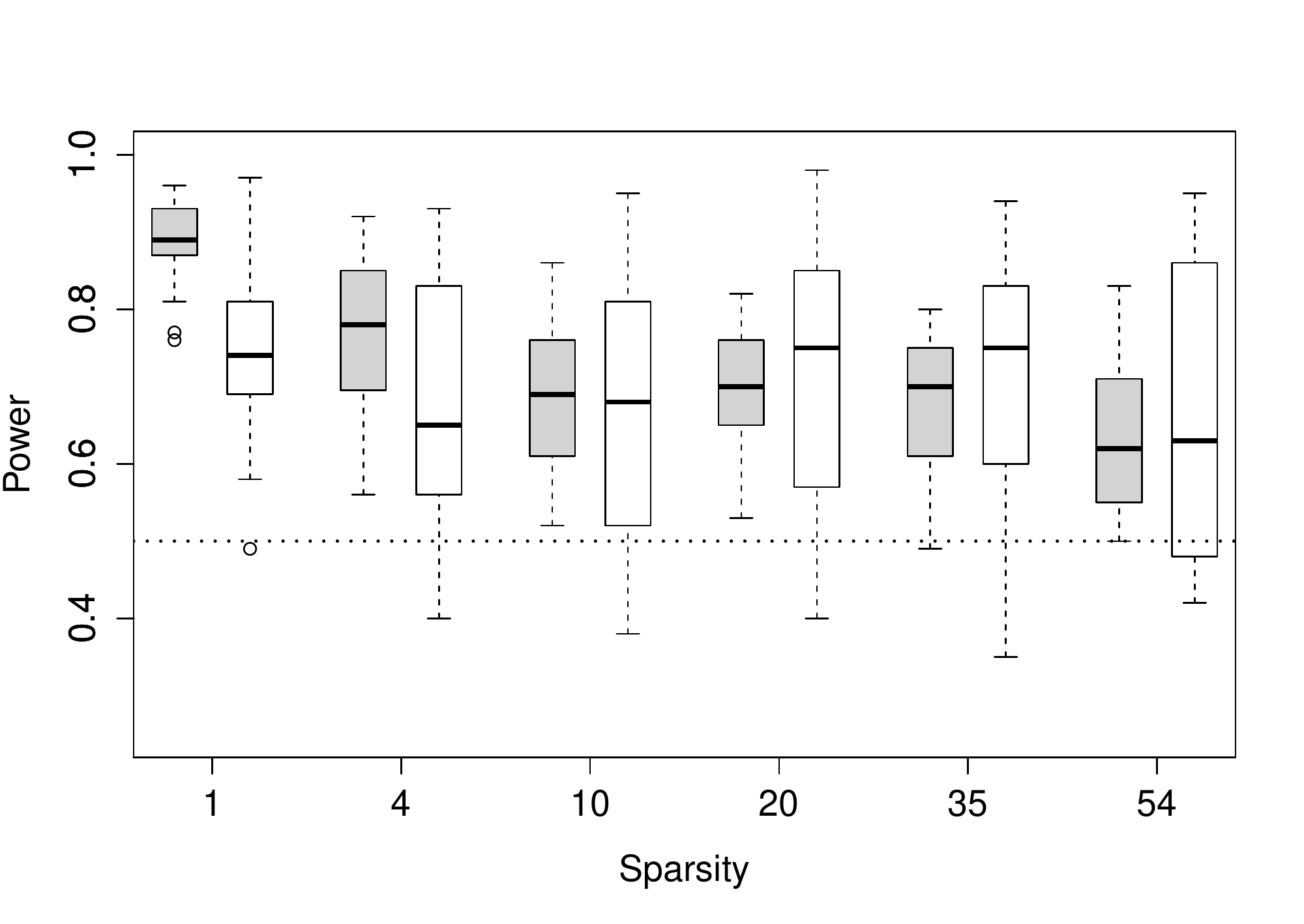}
\caption{Boxplots of the power of RP tests (grey) and the global test (white) across the 25 signals estimated through 100 repetitions, for each of the sparsity levels $s$; the power of the $F$-test is fixed at 0.5 and shown as a dotted line.\label{fig:F_test_diabetes}}
\end{figure}


\subsection{High-dimensional nulls}
In this section we report the results of using RP tests (Algorithm~\ref{alg:Lasso1}) tailored to detect particular alternatives on a variety of simulated examples where the null hypothesis is high-dimensional. We investigate both control of the type I error and the powers of the procedures.

Our examples are inspired by \citet{Dezeure2014}. We use $n \times p$ simulated design matrices with $p=500$ and $n=100$ except for the setting where we test for heteroscedasticity in which we increase $n$ to 300 in order to have reasonable power against these alternatives. The rows of the matrices are distributed as $\mathcal{N}_p(\mb 0, \mbb\Sigma)$ with  $\mbb\Sigma$ given by the three types described in Table~\ref{tab:sim}.

\begin{table}[!h]
\caption{\label{tab:sim}Generation of $\Sigma$.}
\centering
\begin{tabular}{rl}
 Toeplitz: & $\Sigma_{jk}=0.9^{|j-k|}$ \\ 
 Exponential decay: & $(\mbb\Sigma^{-1})_{jk}=0.4^{|j-k|/5}$ \\ 
 Equal correlation: & $\Sigma_{jk}=0.8$ if $j\neq k$ and 1 otherwise. \\ 
 \end{tabular}
 \end{table} 
 
In addition to the randomly generated design matrices, we also used a publicly available real design matrix from gene expression data of Bacillus Subtilis with $n=71$ observations and $p=4088$ predictors \citep{Buehlmann2014}. Similarly to \citet{Dezeure2014}, in order to keep the computational burden of the simulations manageable, we reduced the number of variables to $p=500$ by selecting only those with the highest empirical variance.
For each of the four design settings, we generated 25 design matrices (those from the real data were all the same). The columns of the design matrices were mean-centred and scaled to have $\ell_2$-norm $\sqrt{n}$.

In order to create responses under the null hypothesis, for each of these
100 design matrices, we randomly generated a vector of coefficients
$\mbb\beta$ as follows. We selected a set $S$ of 12 variables from
$\{1,\ldots,p\}$. We then assigned $\mbb\beta_{S^c}=\mb 0$ and each
$\beta_k$ with $k \in S$ was generated according to Unif$[-2,2]$
independently of other coefficients. This form of signal is similar to the
most challenging signal settings considered in
\citet{Dezeure2014} and also resembles the estimated signal from regression of the true response associated with the gene expression data on to the predictors using the Lasso or MCP \citep{zhang2010nearly}. Other constructions for generating the non-zero
  regression coefficients are considered in Section~\ref{sec:num} in the supplementary
  material.
Given $\mb X$ and $\mbb\beta$, we generated $r=100$ responses according to the linear model \eqref{eq:lin_mod} with $\sigma=1$.
Thus in total, here we evaluate the type I error control of our procedures
on over 100 data-generating processes.
The number of bootstrap samples $B$ used was
100 when testing for significance of individual predictors and fixed at 249
in all other settings. 

We now explain interpretation of the plots in Figures~\ref{fig:groups}--\ref{fig:hetero}; a description of Figure~\ref{fig:single} is given in Section~\ref{sec:apps_single}.
The top and bottom rows of each of Figures~\ref{fig:groups}--\ref{fig:hetero} concern settings under null and alternative hypotheses respectively. Thin red curves trace the empirical cumulative distribution functions (CDFs) of the $p$-values obtained using RP tests, whilst thin blue curves, where shown, represent the same for debiased Lasso-based approaches.
In all plots, thickened coloured curves are averages of their respective thin coloured curves; note these are averages over different simulation settings.

The black dashed line is the CDF of the uniform distribution; thus we would hope for the empirical CDFs to be close to this in the null settings (top rows), and rise above it in the bottom rows indicating good power.
Of course, even if all of the $p$-value distributions were stochastically larger than uniform so the type I error was always controlled, we would not expect their estimated distributions i.e.\ the empirical CDFs to always lie below the dashed line.
The black dotted curve allows us to assess type I error control across the simulation settings more easily.
It is constructed such that in each of the plots, were the type I error to be controlled exactly, we would expect on average 1 out of the 25 empirical CDFs for RP tests to escape above the region the line encloses. Thus several curves not completely enclosed under the dotted line in a given plot would indicate poor control of type I error. More precisely, the line is computed as follows. Let $q_\alpha(x)$ be the upper $\alpha$ quantile of a $\text{Bin}(B+1, x)/(B+1)$ distribution. Note this is the marginal distribution of $\hat{U}(x)$ where $\hat{U}$ is the empirical CDF of $B$ samples from the uniform distribution on $\{1/(B+1), 2/(B+1), \ldots, 1\}$. The curve then traces $q_\alpha(x)$ with $\alpha$ chosen such that
\[
\pr\big\{\,\max_{x \in [0,0.1]} (\hat{U}(x) - q_\alpha(x)) >0\big\}=1/25.
\]
We see that across all of the data-generating processes and for each of the three RP testing methods, it appears the size never exceeds the nominal level by a significant amount. Moreover the same holds for the additional 100 data-generating processes whose results presented in the supplementary material: the type I error is controlled well uniformly across all settings considered.

We now describe the particular RP tests used in Figures~\ref{fig:groups}--\ref{fig:hetero}, and the alternatives investigated, as well as the results shown in Figure~\ref{fig:single} concerning testing for the significance of individual predictors as detailed in Section~\ref{sec:single}.
\begin{figure}
    \centering
        \includegraphics[width=\textwidth]{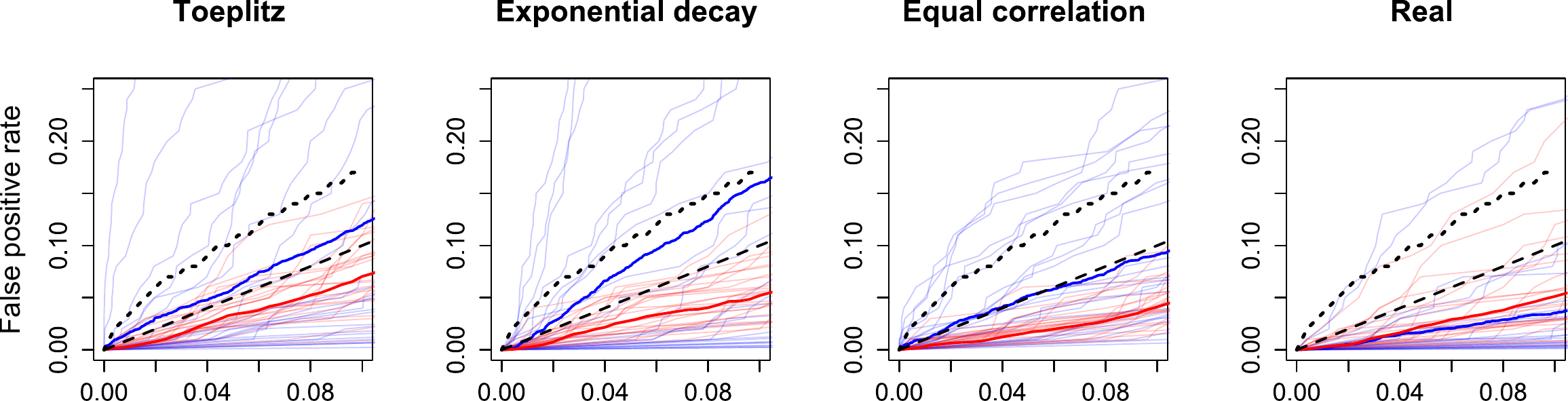}
    \vspace{0.2cm}
    
        \includegraphics[width=\textwidth]{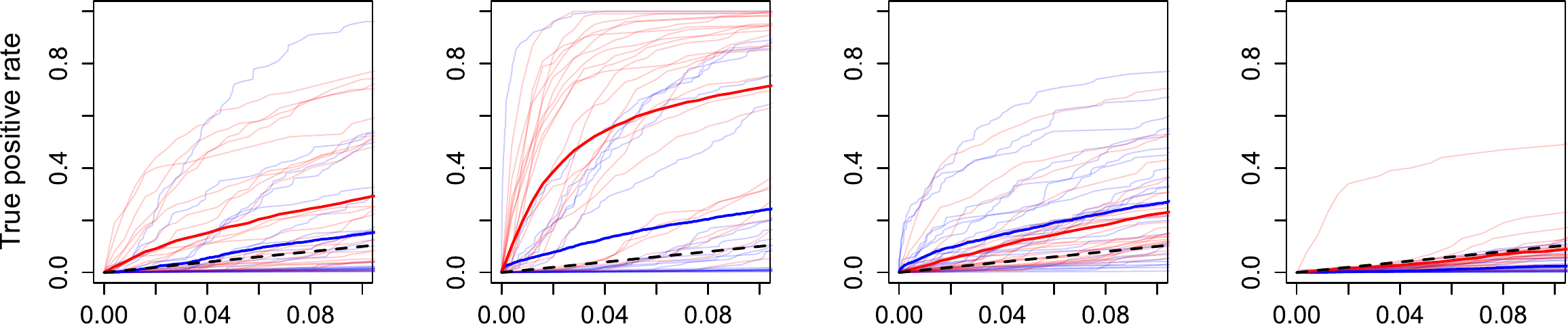}
    \caption{Testing significance of groups: the empirical distribution functions of the $p$-values from RP tests
      (red) and the debiased Lasso (blue) under the null (top row) and
      alternative (bottom row) respectively. The dashed line
        equals the 
        45 degree line 
        corresponding to the Unif$[0,1]$ distribution function, and
        the dotted curve is explained in the main text.
\label{fig:groups}}
\end{figure}
\begin{figure}
    \centering
        \includegraphics[width=\textwidth]{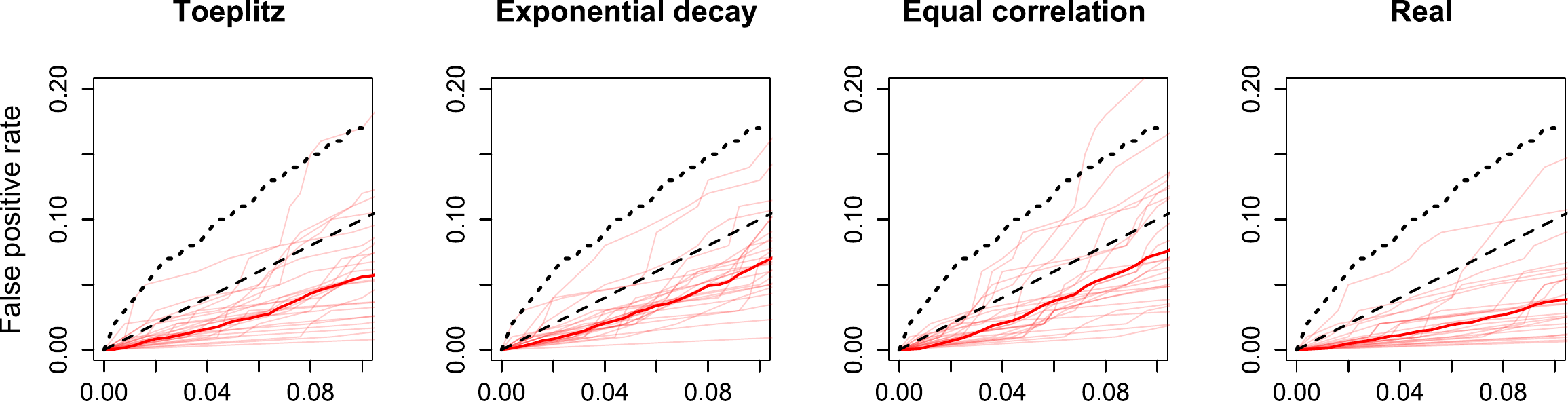}
    \vspace{0.2cm}
    
        \includegraphics[width=\textwidth]{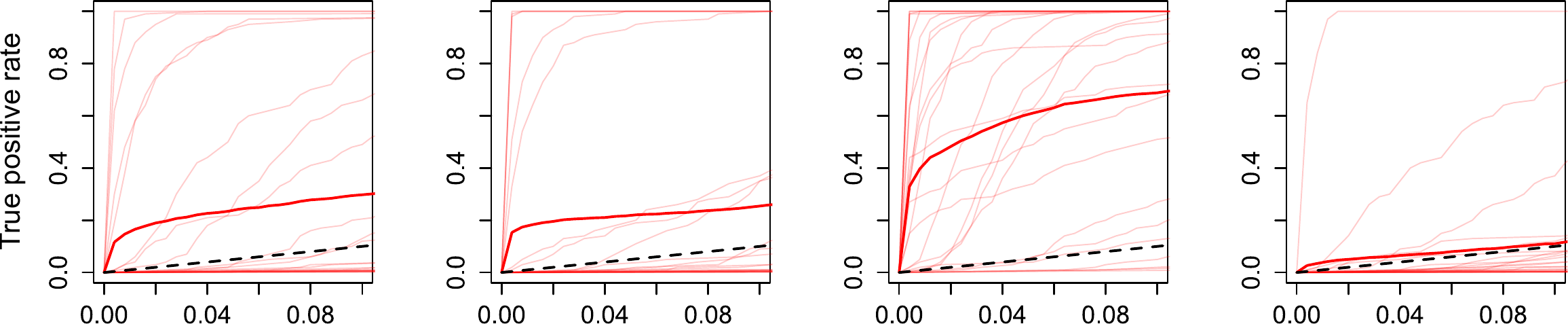}
    \caption{Testing for nonlinearity; the interpretation is similar to that of Figure~\ref{fig:groups}.\label{fig:nonlinear}}
\end{figure}
\begin{figure}
    \centering
        \includegraphics[width=\textwidth]{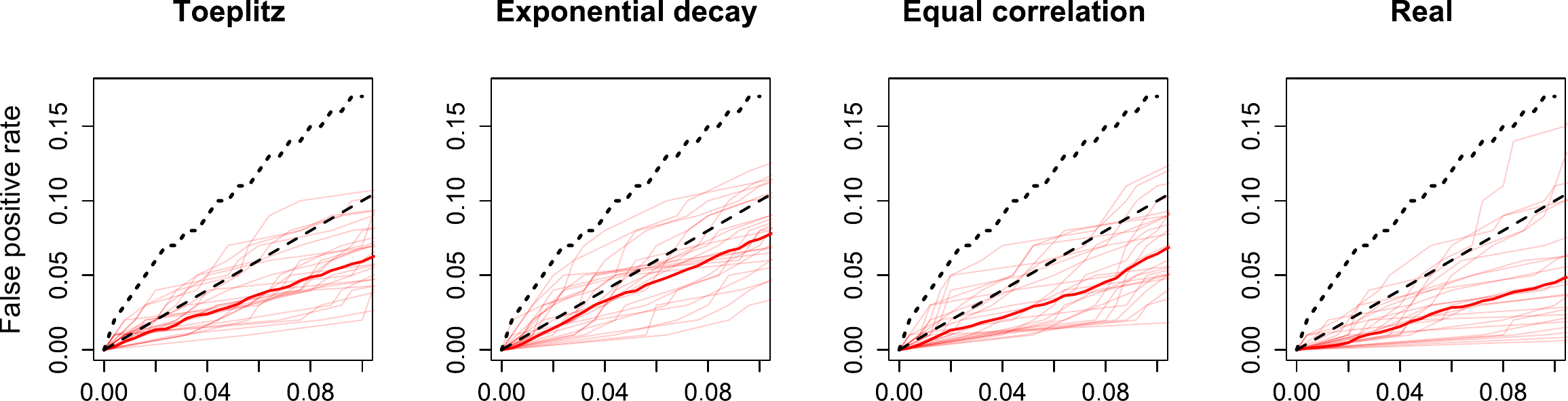}
    \vspace{0.2cm}
    
        \includegraphics[width=\textwidth]{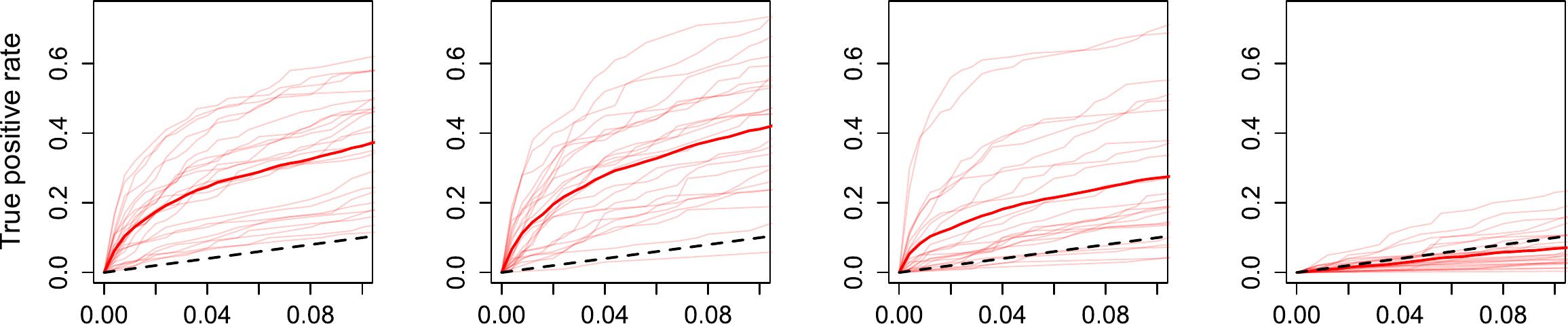}
    \caption{Testing for heteroscedasticity; the interpretation is similar to that of Figure~\ref{fig:groups}.\label{fig:hetero}}
\end{figure}
\begin{figure}
    \centering
        \includegraphics[width=\textwidth]{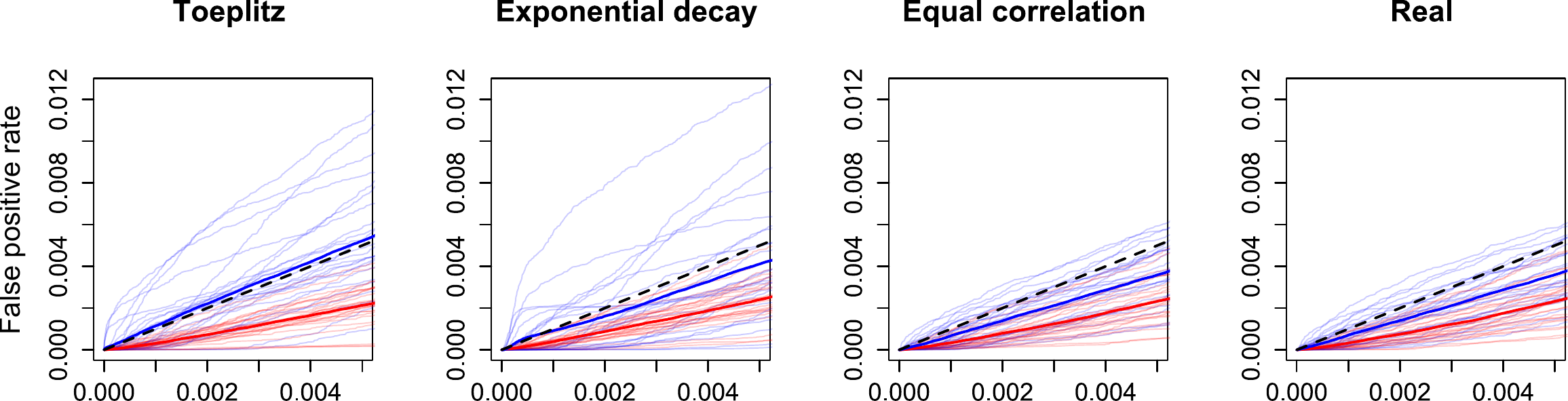}
    \vspace{0.2cm}
    
        \includegraphics[width=\textwidth]{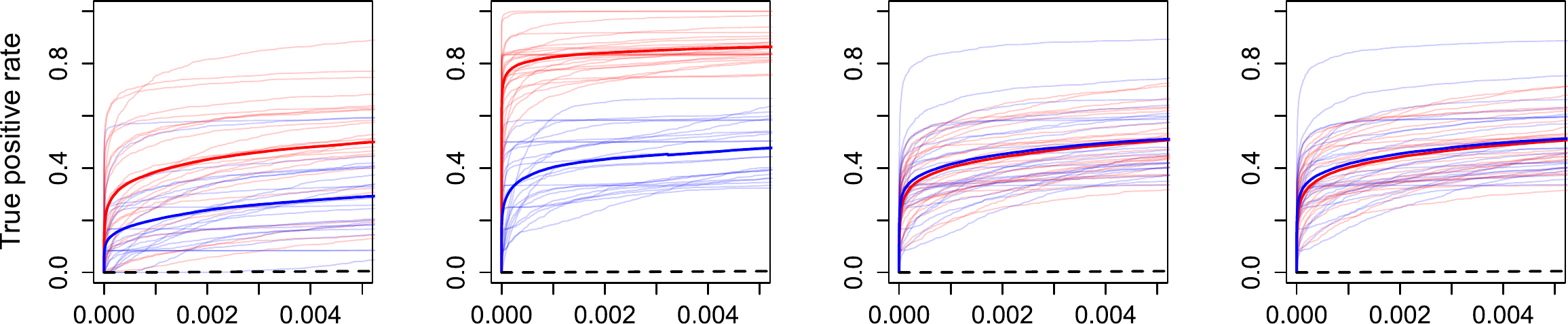}
    \caption{Testing individual variables: the plots give the proportion of $|S^c|=488$ null (top row) and $|S|=12$ true variables (bottom row) selected at various threshold levels with RP tests (red) and the debiased Lasso (blue).\label{fig:single}}
\end{figure}

\subsubsection{Groups}
We consider the problem of testing the null hypothesis $\mbb\beta_G=\mb 0$ within linear model \eqref{eq:lin_mod}. One approach is to regress each column of $\mb X_G$ on to $\mb X_{G^c}$ in turn using the square-root Lasso (c.f.\ Section~\ref{sec:single}), and consider a matrix of residuals $\tilde{\mb X} \in \R^{n \times |G|}$. We may then use Lasso regression on to $\tilde{\mb X}$ as our family of RP methods and combine the resulting test statistics as in Section~\ref{sec:Combine}.

We use this approach on our simulated data and the results are displayed in red in Figure~\ref{fig:groups}. For the null settings (top row) we took $G^c$ to be a randomly selected set of size $p/2$ containing $S$. Thus under the null, $\mbb\beta_{G^c}$ had 12 non-zero components whilst $\mbb\beta_G = \mb 0$. The alternatives, corresponding to the bottom row, also modify the signal such that $\mbb\beta_A$ is non-zero (in addition to $\mbb\beta_S$ being non-zero as was the case under the null) with coefficients generated in exactly the same way as for $\mbb\beta_S$ and $A$ being a randomly selected set of $12$ variables chosen from $G$.

The blue lines trace the empirical CDFs of $p$-values constructed using the debiased Lasso proposal of \citet{optimalconf14} and implemented in the \texttt{hdi} package \citet{Dezeure2014} for \texttt{R}. More specifically, we use the minimum of the $p$-values associated with each of the coefficients in $G$ (see Section 2.3 of \citet{optimalconf14}) as our test statistic, and calibrate this using the Westfall--Young procedure \citep{westfall93resampling} as explained in \citet{pb13}. This ensures that no power is lost due to correlations among the individual $p$-values, as would be the case with Bonferroni correction, for example. Remaining parameters were set to the defaults in the \texttt{hdi} package.

Although the sizes of the debiased Lasso-based tests averaged over the equal correlation design examples are very close to the nominal level, this is due to the several settings where the size exceeds the desired level being compensated for by other examples where the tests are more conservative. On the other hand, RP tests have slightly conservative type I error control across all the examples, and greater power among the Toeplitz and Exponential decay settings.

\subsubsection{Nonlinearity}
In order to test for nonlinearity, we consider an RP method based on Random Forest \citep{breiman01random}. We used the default settings for Random Forest as implemented by \citet{Liaw2002}, but rather than using a direct application to the residuals we apply it to the equicorrelation set: the set of variables with maximum absolute correlation with the residuals. This is invariably the set of variables selected in the initial Lasso fit, though in situations where the Lasso solution is not unique this will in general be a superset of the support of any Lasso solution. Using this smaller set of variables reduces the computational burden of a Random Forest fit, and also gives the test greater power in situations where the variables contributing to the nonlinear signal also feature in sparse linear approximations to the truth. Applying a Random Forest to the entire set of variables may have slightly greater power when this is not the case, but would have greatly diminished power in the more natural situations where this this holds.
Rather than using the RSS from the Random Forest fits as our proxy for prediction error, we use the out of bag error. This has the advantage of being more insensitive to the size of the equicorrelation set and tends to result in greater power.

To create the nonlinear signal for the alternative settings, we randomly divided $S$ into four groups of three. Each variable $x$ was transformed via a sigmoid composed with a random affine mapping as below:
\[
x \mapsto [1 + \exp\{-5(a + bx)\}]^{-1}.
\]
Here $a, b \in \mathcal{N}(0,1)$ independently. The transformed variables in each group were multiplied together, and a linear combination of these resulting products with Unif$[-1,1]$ generated coefficients formed the nonlinear component of the signal. This nonlinear signal was then scaled such that the residuals from an OLS fit to the variables in $S$ had an empirical variance of 2, and finally added to the linear signal.

The results displayed in Figure~\ref{fig:nonlinear} show that RP tests are able to deliver reasonable power in many of the settings considered, though the real design examples appear to be particularly challenging.

\subsubsection{Heteroscedasticity}
As testing for heteroscedasticity in a high-dimensional setting is rather challenging, here we increase the number of observations for the simulated design settings to $n=300$ in order to have reasonable power against the alternative. The data-generation procedure under the null was left unchanged. In order to generate vectors of variances for the alternative settings, we randomly selected 3 variables from $S$ and formed a linear combination of these variables with Unif$[-2,2]$ coefficients. A constant was then added, so the minimum component was 0.01, and finally the vector was scaled so the average of its components was 1. This vector then determined the variance of normal errors added to the signal.

To detect this heteroscedasticity, we used a family of RP methods given by Lasso regression of the absolute values of the residuals onto the equicorrelation set. The results are shown in Figure~\ref{fig:hetero}.
RP tests are able to deliver reasonable power in the the simulated design settings, but do struggle to detect the heteroscedasticity with the real design which has a lower number of observations ($n=71$). 

\subsubsection{Testing significance of individual predictors} \label{sec:apps_single}
Figure~\ref{fig:single} shows the results of using RP tests as described in Section~\ref{sec:single} to test hypotheses $H_k:\beta_k=0$. The red curves give the average proportions of false (top row) and true positives (bottom row) that would be selected given $p$-value thresholds varying along the $x$-axis. Thus for example in order to obtain the expected number of false positives selected at a given threshold, the $y$ values should be multiplied by $p-|S|=488$.
The blue curves display the same results for the debaised Lasso as implemented in the \texttt{hdi} package. The dashed 45 degree line gives the expected proportion of false positives that would be incurred by an exact test.

We see that even at the low $p$-value thresholds particularly relevant for multiple testing correction, RP tests give consistent error control whilst also delivering superior or equal power. Such error control effectively requires accurate knowledge of the extreme tails of the null distribution of the test statistics. We see here that the debiased Lasso approach is not always able to achieve this in the Toeplitz and Exponential decay settings, and indeed error control for multiple testing is rare among the currently available methods \citep{Dezeure2014}.


\section{Discussion} \label{sec:discuss}
The RP testing methodology introduced in this work allows us to treat model
checking as a prediction problem: that of fitting any (prediction) function to the scaled residuals from OLS or
  Lasso. This makes the problem of testing
goodness of fit amenable to the entire range of prediction methods that
have been developed across statistics and machine learning. We have
  investigated here RP tests for detecting significant single or groups of
  variables, heteroscedasticity, or deviations from linearity, and we
  expect that effective RP methods can also be found for testing for
  correlated errors, heterogeneity and other sorts of departures from the
  standard Gaussian linear model. Related ideas should be applicable
  to test for model misspecification in high-dimensional 
generalised linear models, for example.


\bibliographystyle{abbrvnat}
{

}
\end{cbunit}

\begin{cbunit}
\newpage
\setcounter{page}{1}
\section*{Supplementary material}
This supplementary material is organised as follows. Section~\ref{sec:power} contains results on the power of the RP tests approach for detecting nonlinearity.

In Section~\ref{sec:non-Gaussian} we discuss how the RP tests methodology can be extended to test for null hypotheses of linear models with non-Gaussian errors, and present numerical results in support of our proposed scheme. Additional numerical results to complement those of Section~\ref{sec:apps} in the main paper are presented in Section~\ref{sec:num}.
In Section~\ref{sec:Interp} we provide some brief comments on the interpretation of $p$-values derived from RP tests.
Finally the proofs of all of the results in the main paper, we well as those stated in Section~\ref{sec:power}, are collected in Section~\ref{sec:proofs}. Note that all equations numbered 1--10 are in the main paper.
\appendix
\section{Power of Lasso RP tests} \label{sec:power}
In this section we briefly discuss the power of RP tests for detecting nonlinearity. Suppose the response is generated according to
\begin{equation*}
\mb y = \mb X\mbb\beta + \mb f + \sigma\mbb\varepsilon,
\end{equation*}
where $\mbb\beta \in \R^p$ is a sparse vector with $S=\{j:\beta_j \neq 0\}$, $s=|S|$ and as before $\mbb\varepsilon \sim \mathcal{N}_n(\mb 0,  \mb I)$. The nonlinear term $\mb f$ is to be thought of as a vector of function evaluations of some nonlinear function: $f_i = f(\mb x_{i, S})$ where $f:\R^{|S|} \to \R$, though this is not assumed in the sequel.

As in Section~\ref{sec:Lasso} of the main paper, here we require that the columns of $\mb X$ have been scaled to have $\ell_2$-norm $\sqrt{n}$.
To facilitate theoretical analysis, we will assume $\check{\mbb\beta}$ is a Lasso estimate with fixed $\lambda=A_1\sqrt{2\log(p)/n}$ and $A_1>1$, rather than with the tuning parameter selected by cross-validation as in Algorithm~\ref{alg:Lasso1}.
Furthermore, we will also take this to be the tuning parameter used in the construction of the Lasso scaled residuals $\hat{\mb R} = \hat{\mb R}_\lambda(\mbb\beta, \mb f + \sigma \mbb\varepsilon)$. Let the bootstrap scaled residuals be $\hat{\mb R}^* = \hat{\mb R}_\lambda(\check{\mbb\beta}, \check{\sigma}\mbb\zeta)$.
We will also take this $\lambda$ to be the tuning parameter used in the construction of the Lasso scaled residuals $\hat{\mb R} = \hat{\mb R}_\lambda(\mbb\beta, \mb f + \sigma \mbb\varepsilon)$.
Let the bootstrap scaled residuals derived from $\check{\mbb\beta}$ and $\check{\sigma}:=\|\mb y - \mb X\check{\mbb\beta}\|_2/\sqrt{n}$ be $\hat{\mb R}^* := \hat{\mb R}_\lambda(\check{\mbb\beta}, \check{\sigma}\mbb\zeta)$ where $\mbb\zeta \in \mathcal{N}_n (\mb 0, \mb I)$.

To quantify the potential power of RP tests, we define
\[
\mbb\psi_\gamma = \argmin{\mb b \in \R^p} \{\|\mb f  -\mb X\mb b\|_2/\sqrt{n} + \gamma\| \mb b\|_1\}
\]
and let $\mb w_\gamma = \mb f - \mb X\mbb\psi_\gamma$ be the nonlinear signal $\mb f$ residualised with respect to $\mb X$.
As in Section~\ref{sec:single} we consider an asymptotic regime where $p$, $\mb X$, $\mbb\beta$, $S$ and $\mb f$ can all change as $n \to \infty$, though we suppress this in the notation. Also, as in Theorem~\ref{thm:single_pval_null}, let $\mathscr{B} = \{\mb b \in \R^p: \mb b_{S^c}=0\}$.

The result, which follows from Theorem~\ref{thm:single_pval_alt} and its proof, shows that whilst the true residuals are positively correlated with the residualised signal $\mb w_\gamma$, the bootstrap residuals are not.
\begin{cor} \label{cor:power}
Suppose $\|\mb f\|_2/\sqrt{n} \to 0$ and for some $\gamma$ we have $\sqrt{n}\gamma\|\mbb\psi_{\gamma, {S^c}}\|_1\to 0$ and $\|\mb f\|_2\gamma = o(\sqrt{\log(p)})$. Assume there exists $\xi > (A_1+1)/(A_1-1)$ with $s\gamma \sqrt{\log(p)}/\kappa^2(\xi, S)\to 0$. We have
\begin{align*}
 \sup_{\mbb\beta \in \mathscr{B}, x \in \R} \abs{\pr(\mb w_\gamma^T\hat{\mb R} / \|\mb w_\gamma\|_2 \leq x) - \Phi\Big(x - \|\mb w_\gamma\|_2 \big/\sqrt{\sigma^2 + \|\mb w_\gamma\|^2_2/n}\Big)} &\to 0, \\
 \sup_{\mbb\beta \in \mathscr{B}, x \in \R} |\pr(\mb w_\gamma^T\hat{\mb R}^* / \|\mb w_\gamma\|_2 \leq x | \mbb\varepsilon) - \Phi(x)| \inprob 0. 
\end{align*}
\end{cor}

An interesting application of the result above is quantification of the power to detect interactions, as we now discuss.
 Consider a random design setting where $\mb X$ is a scaled version
 of a
 matrix $\mb Z$ whose rows $\mb z_i$ are independent with $\mb z_i \sim
 \mathcal{N}_p(\mb 0, \mbb\Sigma)$ and $\Sigma_{jj}=1$ for all $j$. That is we have $\mb X_k = \sqrt{n} \mb Z_k / \|\mb Z_k\|_2$. Let
 $f_i = \mb z_{i,S}^T \mbb \Theta \mb z_{i,S}$ where without loss of
 generality, $\mbb\Theta \in \R^{s\times s}$ is a symmetric matrix. Thus
 the nonlinear component of the signal is a quadratic function of the
 variables in $S$. As before, we will consider asymptotics where $n, p \to \infty$ and now also $\mbb\Sigma \in \R^{p \times p}$ will change as $p \to \infty$, though we suppress this in the notation.
 We will assume a restricted eigenvalue-type condition on
the sequence of covariance matrices $\mbb\Sigma$:
let
\begin{equation*}
\phi_0 (\xi) = \inf\bigg\{\frac{\|\mbb\Sigma^{1/2}\mb u\|_2}{\|\mb u \|_2} : \mb u \in \mathscr{C}(\xi, S) \bigg\}
\end{equation*}
and assume that for $\xi > (A_1+1)/(A_1-1)$, we have $\phi_0(\xi)>\phi>0$ as $n \to \infty$. Note that this is weaker than assuming the minimum eigenvalue of $\mbb\Sigma$ is bounded away from zero, for example.
 
 \begin{thm} \label{thm:interactions}
 Suppose $n^{1/3} \E(f_1^2) \to 0$
  and $s\log(p)/n^{1/3} \to 0$. We have
\begin{align*}
 \sup_{\mbb\beta \in \mathscr{B}, x \in \R} \abs{\pr(\mb f^T\hat{\mb R} / \|\mb f\|_2 \leq x | \mb Z) - \Phi\big(x - \|\mb f\|_2 /\sqrt{\sigma^2 + \|\mb f\|^2_2/n}\big)} &\inprob 0, \\
 \sup_{\mbb\beta \in \mathscr{B}, x \in \R} |\pr(\mb f^T\hat{\mb R}^* / \|\mb f\|_2 \leq x | \mbb\varepsilon, \mb Z) - \Phi(x)| \inprob 0. 
\end{align*}
 \end{thm}
Note that the theorem allows for $\E(\|\mb f\|_2^2) = n\E(f_1^2) \to \infty$, though we do need $\E(f_1^2) \to 0$.
 We see that the nonlinear signal is positively correlated with the true Lasso residuals, but not with the bootstrap residuals. Thus the nonlinear signal is present in the true Lasso residuals, and in principle can be detected by a suitable RP method.
 



\section{Non-Gaussian errors} \label{sec:non-Gaussian}
Although the null hypothesis that the Gaussian linear model
\eqref{eq:lin_mod} is correct is often of interest, one may wish to
consider a larger null hypothesis that allows for non-Gaussian
errors. Theorem~\ref{thm:maximise_pval} cannot easily be extended to this
setting as it allows for arbitrary (collections of) RP functions to be
used, including those that might directly test for normality. We do not
pursue this further here but note that knowing errors are non-Gaussian can
be helpful for designing a different objective function to use with
$\ell_1$ penalisation that may be more efficient for estimation. We also
note that one could in principle extend the results of
Theorems~\ref{thm:single_pval_null} and \ref{thm:single_pval_alt} to allow
for non-Gaussian error distributions under the null. The result for a
single variable follows via the central limit theorem, but the uniformity
of variables in $S^c$ requires the deep results of
\citet{Chernozhukov2014}. 

Nevertheless, it is desirable that a test for e.g.\ nonlinearity should not
reject more often when a sparse linear model with non-normal errors
holds. When non-Gaussian errors must be included in the null hypothesis, we
recommend taking the simulated errors $\mbb\zeta^{(b)}$ to be a sample
with replacement from the original scaled residuals $\hat{\mb R}$. 

Figures~\ref{fig:groups_exp} and \ref{fig:nonlinear_exp} are identical to Figures~\ref{fig:groups} and \ref{fig:nonlinear} but with exponential errors rather than Gaussian errors used in all simulations. Similarly Figures~\ref{fig:groups_t} and \ref{fig:nonlinear_t} use $t$-distributed errors with 3 degrees of freedom scaled to have variance 1. We use the nonparametric bootstrap approach described above. We see that type I error is very well controlled for RP tests across all the settings with the power also competitive.

\begin{figure}
    \centering
        \includegraphics[width=\textwidth]{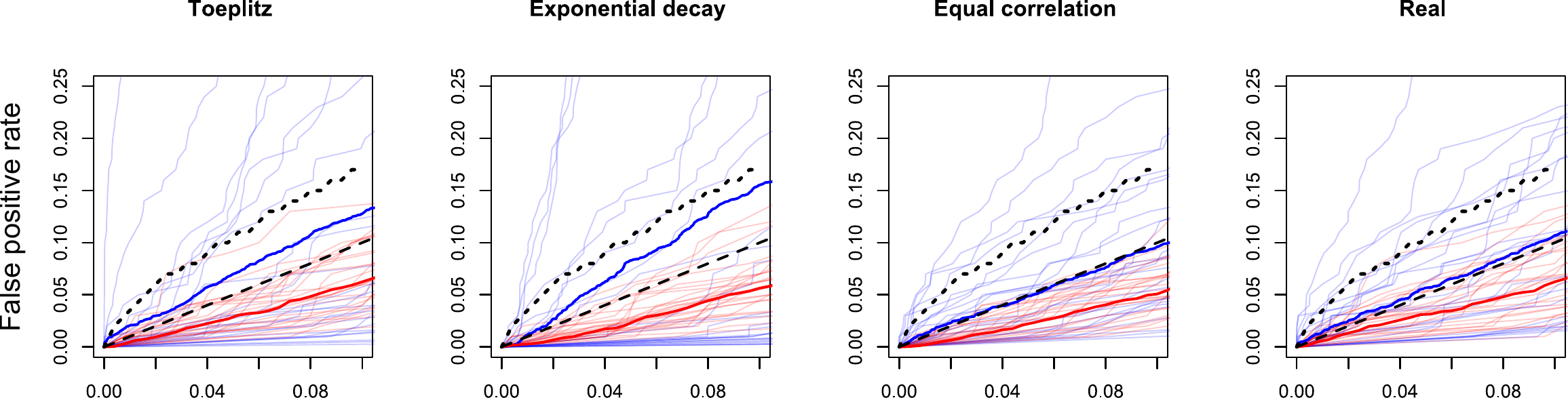}
    \vspace{0.2cm}
    
        \includegraphics[width=\textwidth]{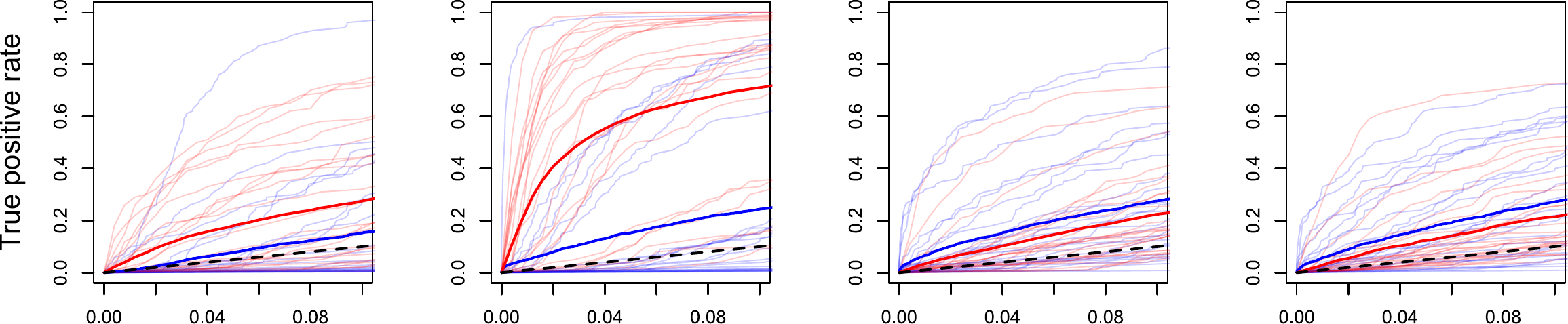}
    \caption{Testing significance of groups with exponential errors; the interpretation is similar to that of Figure~\ref{fig:groups}.
\label{fig:groups_exp}}
\end{figure}

\begin{figure}
    \centering
        \includegraphics[width=\textwidth]{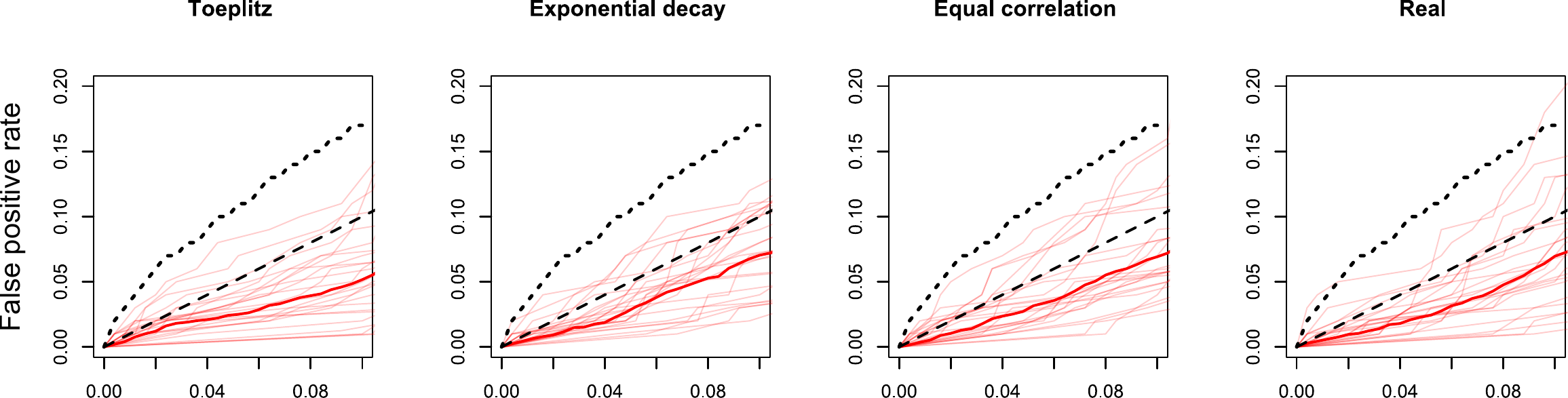}
    \vspace{0.2cm}
    
        \includegraphics[width=\textwidth]{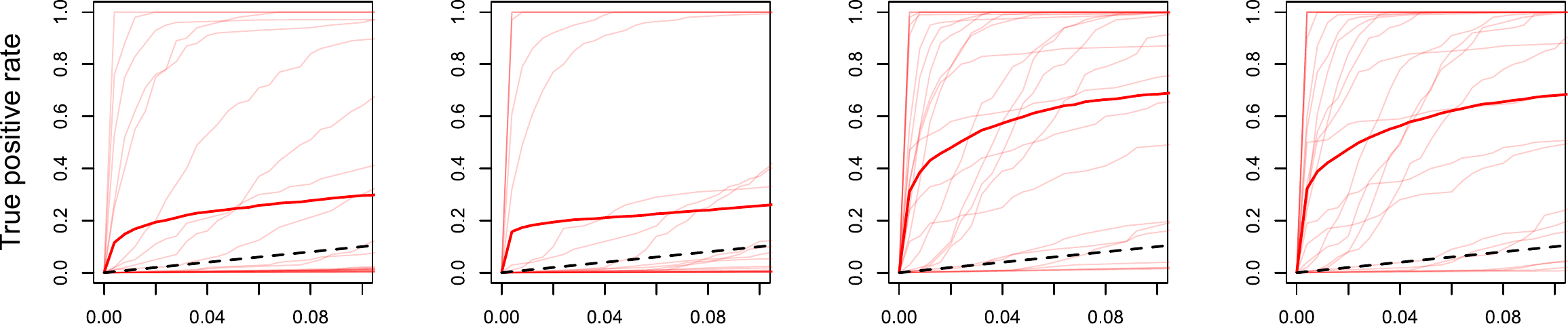}
    \caption{Testing for nonlinearity with exponential errors; the interpretation is similar to that of Figure~\ref{fig:nonlinear}.\label{fig:nonlinear_exp}}
\end{figure}

\begin{figure}
    \centering
        \includegraphics[width=\textwidth]{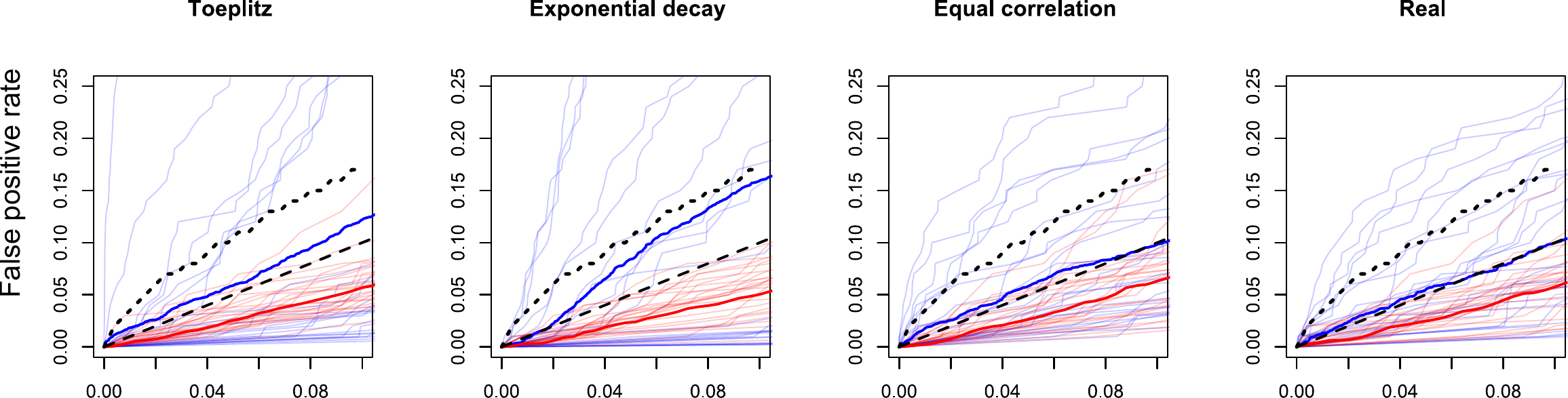}
    \vspace{0.2cm}
    
\includegraphics[width=\textwidth]{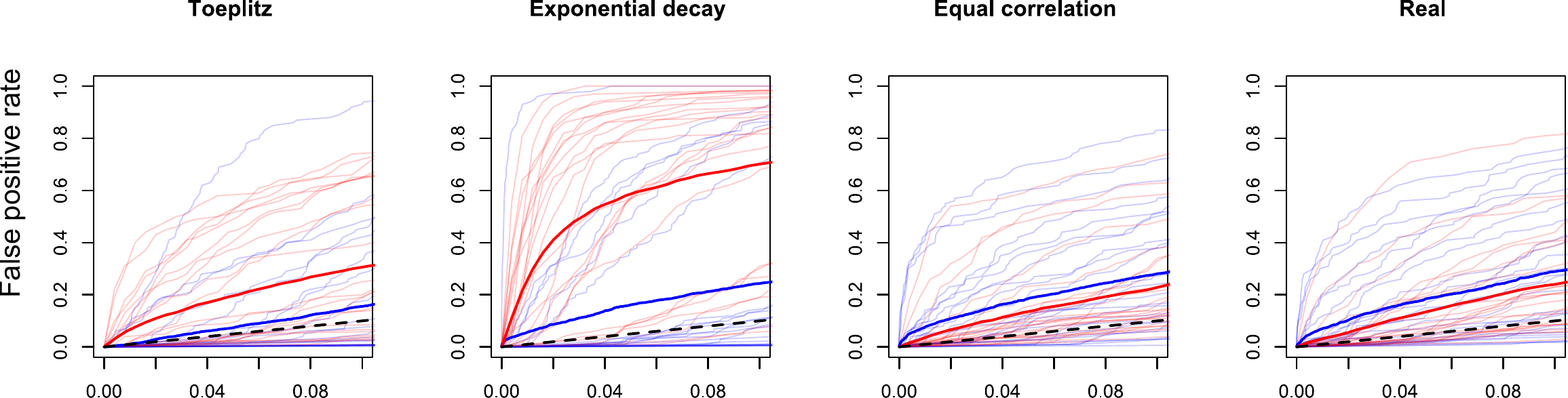}
    \caption{Testing significance of groups with $t$-distributed errors with 3 degrees of freedom; the interpretation is similar to that of Figure~\ref{fig:groups}.
\label{fig:groups_t}}
\end{figure}

\begin{figure}
    \centering
        \includegraphics[width=\textwidth]{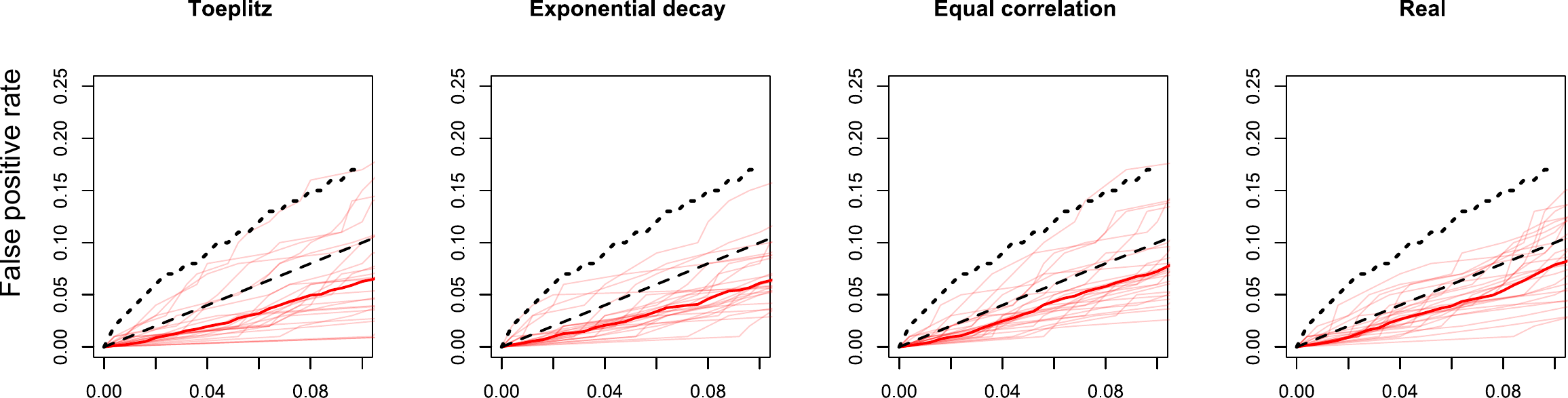}
    \vspace{0.2cm}
    
        \includegraphics[width=\textwidth]{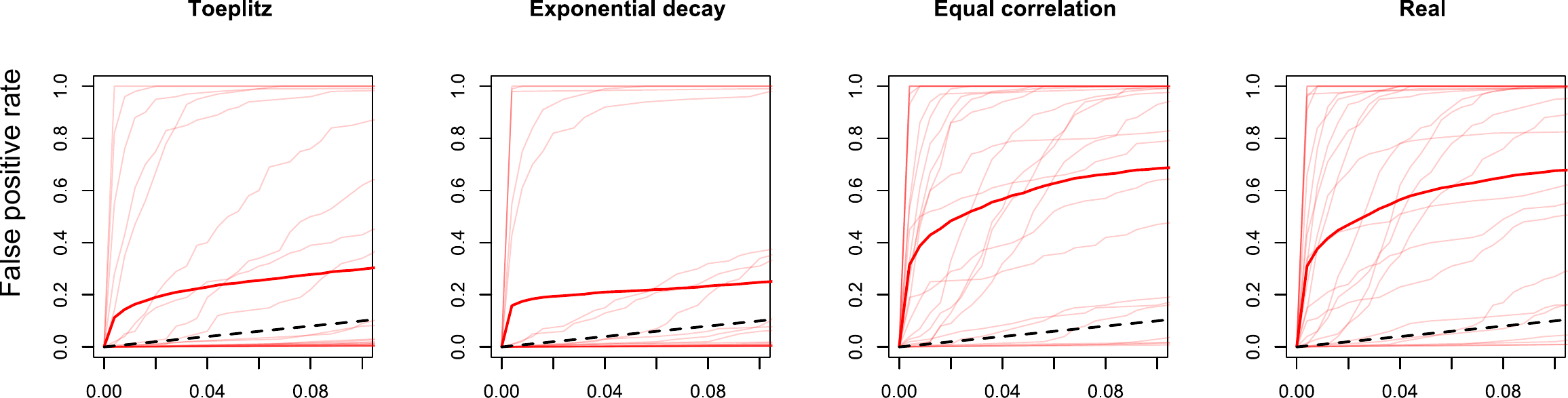}
    \caption{Testing for nonlinearity with $t$-distributed errors with 3 degrees of freedom; the interpretation is similar to that of Figure~\ref{fig:nonlinear}.\label{fig:nonlinear_t}}
\end{figure}

\section{Additional numerical results} \label{sec:num}
In this section we present additional numerical results of the same format as those in Section~\ref{sec:apps} in the main paper, but where the nonzero coefficients given active variables $S=\{j_1,\ldots,j_{12}\}$ are chosen as follows:
\begin{equation} \label{eq:coef_decay}
\beta_{j_k} \propto \frac{1}{\sqrt{k}}.
\end{equation}
The coefficients are then scaled to have an $\ell_1$-norm of 12.
Modulo these modifications, Figures~\ref{fig:groups_decay} and \ref{fig:nonlinear_decay} are exactly analogous to Figures~\ref{fig:groups} and \ref{fig:nonlinear} respectively.

We see the results are very much in line with those of Section~\ref{sec:apps}. Note that the reduced power of RP tests compared to the debiased Lasso in the equal correlation and real design settings of Figure~\ref{fig:groups_decay} is due to the poor calibration of the debiased approach, which here tends to greatly exceeds its nominal level. 
\begin{figure}
    \centering
        \includegraphics[width=\textwidth]{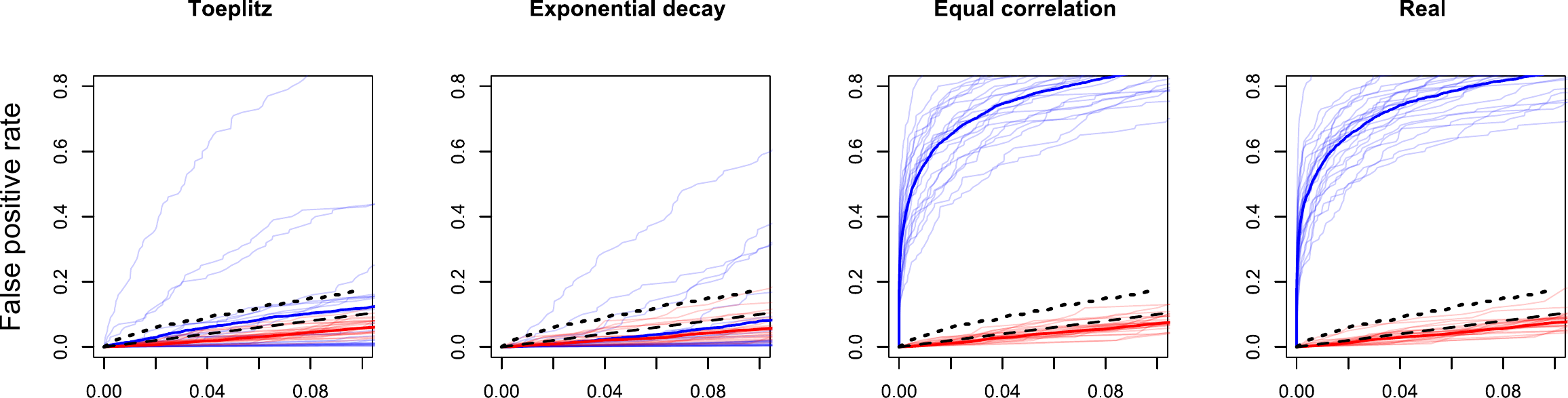}
    \vspace{0.2cm}
    
        \includegraphics[width=\textwidth]{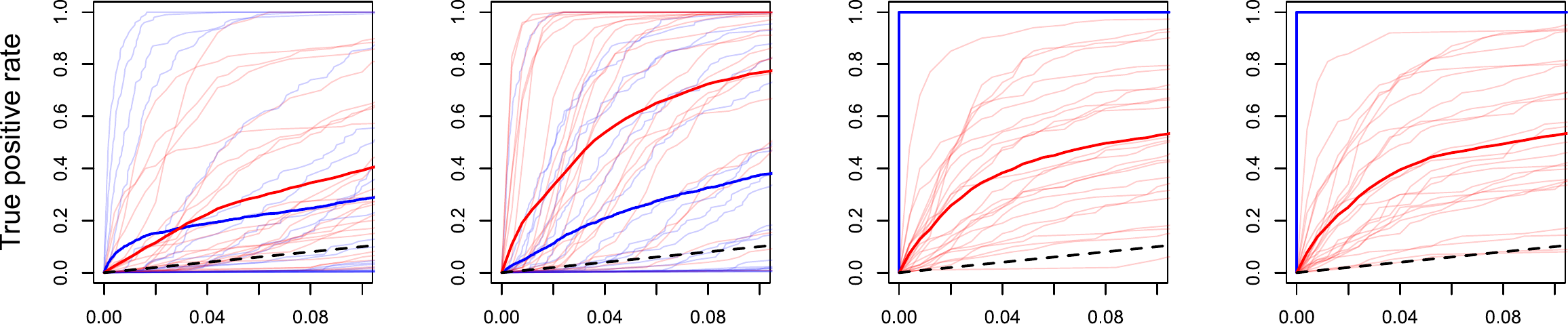}
    \caption{Testing significance of groups when coefficients are chosen according to \eqref{eq:coef_decay}; the interpretation is similar to that of Figure~\ref{fig:groups}.
\label{fig:groups_decay}}
\end{figure}
\begin{figure}
    \centering
        \includegraphics[width=\textwidth]{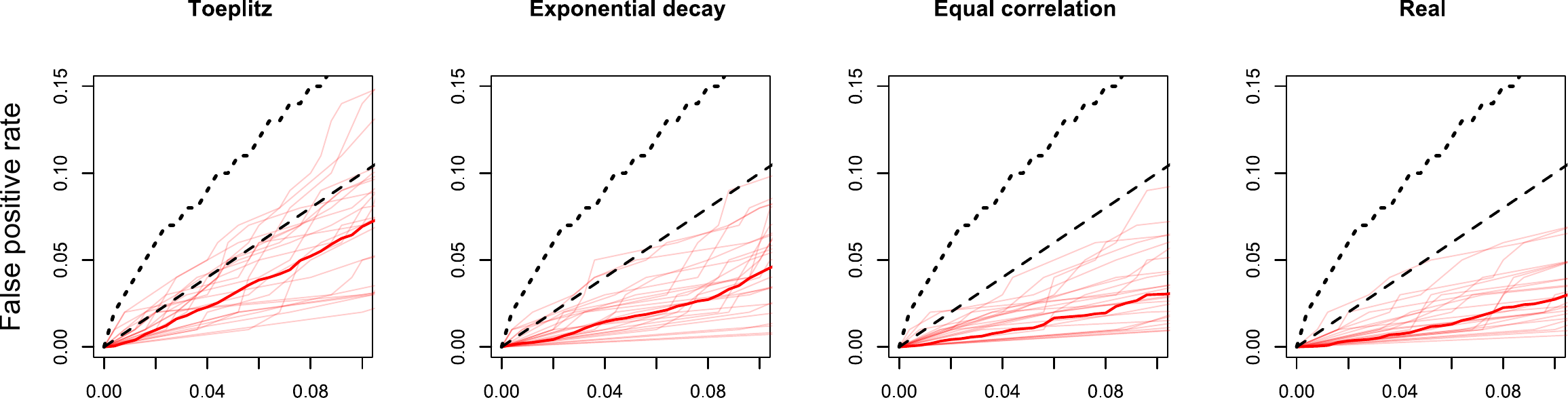}
    \vspace{0.2cm}
    
        \includegraphics[width=\textwidth]{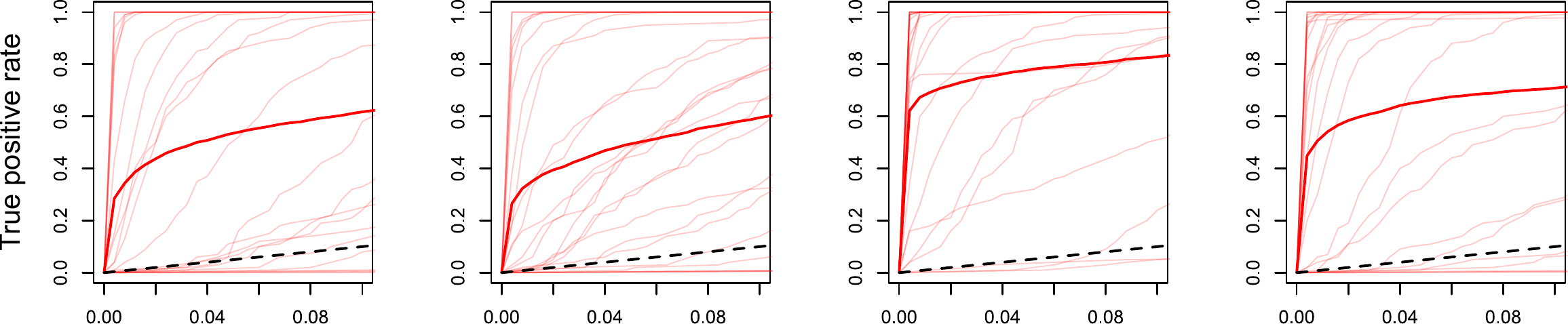}
    \caption{Testing for nonlinearity; the interpretation is similar to that of Figure~\ref{fig:nonlinear}.\label{fig:nonlinear_decay}}
\end{figure}

\section{Uses of RP tests and interpretation of $p$-values} \label{sec:Interp}
One use of RP tests is as a form of reassurance that Lasso-based
inference is safe to use on the data at hand: a large $p$-value indicates a
lack of evidence for the Lasso having poor performance.
A small value on the other hand suggests a sparse linear model is inappropriate.

As we explain in this work however, the RP testing framework is general enough to include tests for significance of groups or individual predictors. It is important to note though that as with all $p$-values, low values in these cases can only indicate inadequacy of
the null models considered;
extra assumptions are required to draw conclusions, such as that a variable is significant,
on the basis of having observed a low $p$-value.

%

\section{Proofs} \label{sec:proofs}
In the proofs that follow, we will let $c_1,\ldots,c_4 \geq 0$ denote constants, which may change from line to line.
\subsection{Proof of Theorem~\ref{thm:lambda_fixed}}
In the following we suppress the dependence of $\hat{\mbb\beta}_\lambda(\mbb\beta,\sigma\mbb\zeta)$ on $\lambda$ for notational simplicity.
  We know that every Lasso solution $\hat{\mbb\beta} (\mbb\beta, \sigma\mbb \zeta)$ is characterised by the KKT conditions
  \[
   \frac{1}{\sqrt{n}}\frac{\mb X^T [\mb X\{\mbb\beta - \hat{\mbb\beta} (\mbb\beta, \sigma\mbb \zeta)\} + \sigma\mbb \zeta]}{\|\mb X\{\mbb\beta - \hat{\mbb\beta} (\mbb\beta, \sigma\mbb \zeta)\} + \sigma\mbb \zeta\|_2} = \lambda \hat{\mbb\nu},
  \]
 where $ \norms{\hat{\mbb\nu}}_\infty \leq 1$ and $\hat{\mbb\nu}_{\hat{S}} = \sgn(\hat{\mbb\beta}_{\hat{S}}(\mbb\beta, \sigma\mbb \zeta))$ with $\hat{S}= \{j:\hat{\beta}_j(\mbb\beta, \sigma\mbb \zeta) \neq 0\}$.
 
 Now picking a particular Lasso solution $\hat{\mbb\beta}(\mbb\beta, \sigma\mbb \zeta)$ in the case where it is not unique, let 
 \[
 \tilde{\mbb\beta}(\mbb \zeta) = \check{\mbb\beta} + \frac{\check{\sigma}}{\sigma}\{\hat{\mbb\beta}(\mbb\beta, \sigma\mbb \zeta) - \mbb\beta\}. 
 \]
 Note that when $\mbb \zeta \in \Lambda_{\lambda,t}$, the upper bound on $\check{\sigma}$
 ensures that we have $\sgn(\tilde{\mbb\beta}(\mbb \zeta)) = \sgn(\hat{\mbb\beta}(\mbb\beta, \sigma\mbb \zeta))$. Next observe that
 \begin{equation} \label{eq:res_eq}
  \mb X\{\check{\mbb\beta} - \tilde{\mbb\beta}(\mbb \zeta)\} + \check{\sigma} \mbb \zeta = -\frac{\check{\sigma}}{\sigma}\mb X\{\hat{\mbb\beta}(\mbb\beta, \sigma\mbb \zeta) - \mbb\beta\} +\check{\sigma}\mbb \zeta = \frac{\check{\sigma}}{\sigma}[\mb X\{\mbb\beta - \hat{\mbb\beta}(\mbb\beta, \sigma\mbb \zeta)\} +\sigma\mbb \zeta],
 \end{equation}
 so
  \[
  \frac{1}{\sqrt{n}} \frac{\mb X^T [\mb X\{\check{\mbb\beta} - \tilde{\mbb\beta} (\mbb \zeta)\} + \check{\sigma}\mbb \zeta]}{\|\mb X\{\check{\mbb\beta} - \tilde{\mbb\beta} (\mbb \zeta)\} + \check{\sigma}\mbb \zeta\|_2} = \frac{1}{\sqrt{n}}\frac{\mb X^T [\mb X\{\mbb\beta - \hat{\mbb\beta} (\mbb\beta, \sigma\mbb \zeta)\} + \sigma\mbb \zeta]}{\|\mb X\{\mbb\beta - \hat{\mbb\beta} (\mbb\beta, \sigma\mbb \zeta)\} + \sigma\mbb \zeta\|_2}= \lambda \hat{\mbb\nu}(\mbb\beta, \mbb \zeta).
  \]
  But this shows that $\tilde{\mbb\beta}(\mbb \zeta)$ satisfies the KKT conditions for $\hat{\mbb\beta}(\check{\mbb\beta} , \check{\sigma}\mbb \zeta)$. Since Lasso fitted values are unique, we must have $\mb X \tilde{\mbb\beta}(\mbb \zeta) = \mb X \hat{\mbb\beta}(\check{\mbb\beta} , \check{\sigma}\mbb \zeta)$. Now substituting into \eqref{eq:res_eq} finally shows that $\hat{\mb R}_\lambda(\check{\mbb\beta}, \check{\sigma}\mbb \zeta)=\hat{\mb R}_\lambda(\mbb \beta, \sigma\mbb \zeta)$ as required. 

\subsection{Results from \citet{Sun2012}}
The proofs of Theorems~\ref{thm:maximise_pval}--\ref{thm:interactions}
make use of Theorem 2 and Corollary 1 in \citet{Sun2012}. We re-state a
subset of these results here for convenience. We have modified the notation
in \citet{Sun2012} in order to avoid clashes with our own
notation. Furthermore, we have replaced the sign-restricted cone invertibility factor $F_2(\xi, S)$ (equation 21 in \citet{Sun2012}) with its lower bound $F_2(\xi, S) \geq \phi^2(\xi)/(1+\xi)$ \citep{Zhang2012}.

Consider the linear model setup of \eqref{eq:lin_mod} though without any assumptions on the distribution of $\mbb\varepsilon$ initially. For $\xi>1$ and $\lambda>0$, define
\begin{equation} \label{eq:mu_def}
\mu(\lambda, \xi, \mbb\beta, \mb X)=(\xi+1)\min_T\inf_{0<\nu<1}\max \bigg[ \frac{\|\mbb\beta_{T^c}\|_1}{\nu}, \, \frac{\lambda|T|/\{2(1-\nu)\}}{\kappa^2\{(\xi+\nu)/(1-\nu), T, \mb X\}} \bigg],
\end{equation}
where the minimum is over $T \subset \{1,\ldots,p\}$ and $T^c=\{1,\ldots,p\}\setminus T$.
Further let $\tilde{\sigma}=\|\mb y-\mb X\mbb\beta\|_2/\sqrt{n}$ and $\tau^2=\tau^2(\tilde{\sigma}, \lambda, \xi, \mbb\beta, \mb X)=\lambda \mu(\tilde{\sigma}\lambda, \xi, \mbb\beta, \mb X)/\tilde{\sigma}$.
Writing $\hat{\mbb\beta}=\hat{\mbb\beta}_{\lambda}(\mbb\beta, \sigma\mbb\varepsilon)$, define $\hat{\sigma}=\|\mb y - \mb X\hat{\mbb\beta}\|_2/\sqrt{n}$.
\begin{thm}[Theorem 2 and Corollary 1 in \citet{Sun2012}] \label{thm:Sun}
Let
\begin{equation} \label{eq:Lambda_Sun}
\Lambda_1 = \bigg\{\mbb\zeta \in \R^n :\frac{\|\mb X^T \mbb\zeta\|_\infty}{\sqrt{n}\|\mbb\zeta\|_2} \leq (1-\tau^2)\lambda\frac{\xi-1}{\xi+1} \bigg\}.
\end{equation}
When $\mbb\varepsilon \in \Lambda_1$,
\begin{align*}
\max(1-\hat{\sigma}/\tilde{\sigma}, 1-\tilde{\sigma}/\hat{\sigma}) &\leq \tau^2, \\
\|\hat{\mbb\beta}-\mbb\beta\|_1 &\leq \mu(\tilde{\sigma}\lambda, \xi, \mbb\beta, \mb X)/(1-\tau^2) \\
\|\hat{\mbb\beta}-\mbb\beta\|_2 &\leq \frac{2\tilde{\sigma}\xi \sqrt{s}\lambda}{(1-\tau^2)\phi^2(\xi)}.
\end{align*}
\end{thm}

\begin{lem} \label{lem:t_tail}
Let $n, p, m \in \mathbb{N}$ with $n\geq m \geq 3$. Let $\mbb\varepsilon \sim \mathcal{N}_m(\mb 0, \mb I)$.
Let $\mb a \in \R^m$ and suppose $0 < \eta < pe^{-(n-m)-2}$. Then
\begin{equation*}
\pr(\mb a^T \mbb\varepsilon / \|\mbb\varepsilon\|_2 > \sqrt{2\log(p/\eta)/n}\, \|\mb a\|_2) \leq \frac{1}{p}\frac{(1+r_m)\eta}{\sqrt{\pi\log(p/\eta)}}
\end{equation*}
where $r_m \to 0$ as $m \to \infty$.
\end{lem}
\begin{proof}
We follow the proof of part (ii) of Theorem 2 of \citet{Sun2012}. Let $z =\mb a^T \mbb\varepsilon / (\|\mb a\|_2 \|\mbb\varepsilon\|_2)$. Then $z/\sqrt{(1-z^2)/(m-1)}$ follows a $t$-distribution with $m-1$ degrees of freedom. The only change we need to make in the aforementioned proof is to note that
\[
\eta < pe^{-(n-m)-2} \Rightarrow m-2 - \log(p/\eta) \geq n-2\log(p/\eta).
\]
and modify (A8) in \citet{Sun2012} appropriately.
\end{proof}

\subsection{Proof of Theorem~\ref{thm:maximise_pval}}
Without loss of generality assume $S=\{1,\ldots,s\}$. Let $\lambda_0 =
2\sqrt{\log(p/\eta)/n}$ and let $\mb P$ denote the orthogonal projection on to $\mb X_S$. Define the following subsets of $\R^n$:
\begin{align*}
\Lambda_1 &= \bigg\{ \mbb\zeta : \max_j \frac{|\mb X_j^T (\mb I - \mb P)\mbb\zeta|}{\|(\mb I - \mb P)\mbb\zeta\|_2} \leq \lambda_0 \|(\mb I - \mb P)\mb X_j\|_2/\sqrt{2} \bigg\} \\
\Lambda_2 &= \bigg\{ \mbb\zeta : \max_j \frac{|\mb X_j^T \mb P \mbb\zeta|}{\|\mbb\zeta\|_2} \leq \lambda_0 \| \mb P \mb X_j\|_2/\sqrt{2} \bigg\} \\
\Lambda_3 &= \{\mbb\zeta : \|\mbb\zeta\|_2/\sqrt{n} \leq\sqrt{2}\}.
\end{align*}
Let $\Lambda = \Lambda_1 \cap \Lambda_2 \cap \Lambda_3$ and let $\Omega = \{\mbb\varepsilon \in \Lambda\}$.
Lemma~\ref{lem:t_tail} shows that on the event $\Omega$
we have the following properties.
\begin{enumerate}[(i)]
\item $\sgn(\mbb\beta'_S)= \sgn(\mbb\beta_S)$ and $\min_{j \in S} \beta'_j / \beta_j > 1 - 1/(2\sqrt{2})$.
\item $\hat{S}^{(s)}=S$ and $\sgn(\check{\mbb\beta}^{(s)}) = \sgn(\mbb\beta_S)$.
\item $2\sqrt{2} \sigma \min_{j \in S} \check{\beta}^{(s)}_j/\beta_j > \check{\sigma}^{(s)}$.
\end{enumerate}
In the following, we suppress dependence on $\lambda$. Let
$\check{\mbb\beta} \in \R^p$ be $\check{\mbb\beta}^{(s)}$ with $p-s$ zeroes
added: $\check{\mbb\beta} = (\check{\mbb\beta}^{(s)}, 0, \ldots,0)$ and let $\check{\sigma}=\check{\sigma}^{(s)}$.
Note that by Theorem~\ref{thm:lambda_fixed}, on $\Omega$ we know that $\hat{\mb R}(\check{{\mbb\beta}}, \check{\sigma} \mbb\zeta) = \hat{\mb R}(\mbb\beta, \sigma \mbb\zeta)$ for all $\mbb\zeta \in \Lambda$. Moreover, Lemma~\ref{lem:indep} shows that conditional on $\Omega$, $\hat{\mb R}(\mbb\beta, \sigma \mbb\varepsilon)$ and $\mb X\check{\mbb\beta} / \check{\sigma} $ are independent.
Write $\hat{\mb R}^{(b)} = \hat{\mb R}_\lambda(\check{{\mbb\beta}}, \check{\sigma} \mbb\zeta^{(b)})$ for $b=1,\ldots, B$ and let $\hat{\mb R}^{(0)}=\hat{\mb R}_\lambda(\mbb\beta, \sigma \mbb\varepsilon)$. We see that conditional on $\Omega$, $\{\hat{\mb R}^{(b)}\}_{b=0}^B$ are independent. Also $\hat{\mb R}^{(0)}\ind_{\{\mbb\varepsilon \in \Lambda\}}, \{\hat{\mb R}^{(b)} \ind_{\{\mbb\zeta^{(b)} \in \Lambda\}}\}_{b=1}^B$ are independent and identically distributed.

Let $\tilde{Q}_b$, $b=0, \ldots, B$ be derived  as in
Section~\ref{sec:Combine} by applying the function $\tilde{Q}$ to appropriate functions of scaled residuals $\{\hat{\mb R}^{(b)}\}_{b=0}^B$. Recall that $Q= \sum_{b=0}^B \ind_{\{\tilde{Q}_b \geq \tilde{Q}_0\}}/(B+1)$. Let $R = \{b : \mbb\zeta^{(b)} \in \Lambda\}$ and note that letting $\delta = \pr(\Omega^c)$ we have $|R| \sim \text{Bin}(1-\delta, B)$. We have
\begin{align*}
\pr(Q \leq x) &\leq \pr(Q \leq x, \Omega)  + \delta \\
&= (1-\delta) \E \{\pr(Q \leq x | R,\Omega)\} + \delta.
\end{align*}
Lemma~\ref{lem:event_prop} gives the required bound on $\delta$; it only remains to show the first term on the RHS is at most $x$. From the above, conditional on $R$ and the event $\Omega$, $\{\tilde{Q}_0, \{\tilde{Q}_b\}_{b \in R}\}$ are exchangeable. Now there can be at most $\floor{x(B+1)}$ values $b \in \{0,\ldots,B\}$ with $\sum_{b' \neq b} \ind_{\{\tilde{Q}_{b'} \geq \tilde{Q}_b\}}/(B+1) \leq x$.  This entails that
\begin{align*}
(|R|+1)\pr(Q \leq x | R, \Omega) &\leq \floor{x(B+1)} \\
\pr(Q \leq x | R, \Omega) &\leq \frac{x(B+1)}{|R|+1}.
\end{align*}
Therefore
\begin{align*}
(1-\delta) \E \{\pr(Q \leq x | R, \Omega)\} &\leq x (1-\delta) \E \bigg(\frac{B+1}{|R|+1}\bigg) \\
&= x \sum_{r=0}^B (1-\delta)^{r+1}\delta^{B-r} \binom{B}{r} \frac{B+1}{r+1} \\
&= x\sum_{r=1}^{B+1} (1-\delta)^{r}\delta^{B+1-r} \binom{B+1}{r} \\
&= x(1 - \delta^{B+1})\leq x.
\end{align*}

\subsection{Proofs of Theorems~\ref{thm:single_pval_null} and \ref{thm:single_pval_alt}} \label{sec:single_proofs}
The proofs of Theorems~\ref{thm:single_pval_null} and \ref{thm:single_pval_alt} rest on the following decomposition of $T_k$ and an analogous one for $T_k^*$:
\begin{align} \label{eq:T_k_decomp}
T_k= \frac{1}{\hat{\sigma}_k} \norms{\mb W_k}_2 \beta_k + \frac{\sigma}{\hat{\sigma}_k} Z_k + \frac{1}{\hat{\sigma}_k}\delta_k
\end{align}
where
\[
\hat{\sigma}_k = \|\mb y - \mb X_{-k} \hat{\mbb\Theta}_k\|_2/\sqrt{n}.
\]
The first term in \eqref{eq:T_k_decomp} is zero for $k \in N$ and
\begin{align}
 \delta_k &= \frac{\mb W_k^T}{\|\mb W_k\|_2} \mb X_{-k}(\mbb\Theta_k - \hat{\mbb\Theta}_k), \\
 Z_k &= \frac{\mb W_k^T \mbb\varepsilon}{\|\mb W_k\|_2} \sim \mathcal{N}(0, 1). \label{eq:def_Z}
\end{align}
The main term we have to control is $\delta_k$. The result of \citet{Sun2012} shows that $\|\mbb\Theta_k - \hat{\mbb\Theta}_k\|_1$ is small with high probability. Next appealing to the KKT conditions for the square-root Lasso, we obtain 
\begin{equation}
\|\mb X_{-k}^T \mb W_k\|_\infty /\|\mb W_k\|_2 \leq \sqrt{n} \gamma. \label{eq:KKT_W}
\end{equation}
H\"older's inequality then gives
\begin{equation} \label{eq:holder}
|\delta_k| \leq  \sqrt{n} \gamma \|\mbb\Theta_k - \hat{\mbb\Theta}_k\|_1,
\end{equation}
which forms the basis of the proofs.
\subsubsection{Proof of Theorem~\ref{thm:single_pval_null}}
In view of Lemma~\ref{lem:convergence} applied with $\mathcal{F}_n$ simply a constant, for the first part we need only show that $\sup_{\mbb\beta \in \mathscr{B}, k \in N} \delta_k \inprob 0$ and $\sup_{\mbb\beta \in \mathscr{B}, k \in N} |\hat{\sigma}_k -\sigma| \inprob 0$. 
First note that
\begin{equation*}
\sup_{\mbb\beta \in \mathscr{B}, k \in N} \mu(\lambda, \xi, \mbb\beta_{-k}, \mb X_{-k}) \leq \sup_{\mbb\beta \in \mathscr{B}} \mu(\lambda, \xi, \mbb\beta, \mb X) \to 0
\end{equation*}
in view of Lemma~\ref{lem:mu_bd} and consequently (as clearly $\lambda \to 0$)
\begin{align*}
\sup_{\mbb\beta \in \mathscr{B}, k \in N} \tau^2(\tilde{\sigma}, \lambda, \xi, \mbb\beta_{-k}, \mb X_{-k}) \leq \sup_{\mbb\beta \in \mathscr{B}} \tau^2(\tilde{\sigma}, \lambda, \xi, \mbb\beta, \mb X) \to 0.
\end{align*}
Now writing $\tau^2 =\sup_{\mbb\beta \in \mathscr{B}}\tau^2(\tilde{\sigma}, \lambda, \xi, \mbb\beta, \mb X)$ for convenience, we see that for $n$ sufficiently large it must be the case that $(1-\tau^2)\lambda(\xi-1)/(\xi+1) > \sqrt{2\log(p)/n}$. By Theorem~\ref{thm:Sun} there is a sequence of events with probability tending to 1 on which we have
\begin{gather}
\sup_{\mbb\beta \in \mathscr{B}, k \in N} \max \bigg(1 - \frac{\hat{\sigma}_k}{\tilde{\sigma}}, 1 - \frac{\tilde{\sigma}}{\hat{\sigma}_k}\bigg) \leq \tau^2, \label{eq:sigma_bd1}\\
\sup_{\mbb\beta \in \mathscr{B}, k \in N} \| \hat{\mbb\Theta}_k - \mbb\Theta_{k}\|_1 \leq \frac{(\xi+1)\tilde{\sigma}\lambda s}{(1-\tau^2) \kappa^2(\xi, S)} \label{eq:Theta_bd1}, \\
\sup_{\mbb\beta \in \mathscr{B}}\max \bigg(1 - \frac{\check{\sigma}}{\tilde{\sigma}}, 1 - \frac{\tilde{\sigma}}{\check{\sigma}}\bigg) \leq \tau^2 \label{eq:sigma_bd2},
\end{gather}
where $\tilde{\sigma}=\sigma\|\mbb\varepsilon\|_2/\sqrt{n}$. From Lemma~\ref{lem:chisq} we know that $\tilde{\sigma} \inprob \sigma$, so in particular we have $\sup_{\mbb\beta \in \mathscr{B}, , k \in N} |\hat{\sigma}_k -\sigma| \inprob 0$. Thus also, on a sequence of events with probability tending to 1, applying \eqref{eq:holder} to \eqref{eq:Theta_bd1} we have
\begin{align*}
\sup_{\mbb\beta \in \mathscr{B}, k \in N} |\delta_k| \leq c_1\sigma  \frac{\sqrt{\log(p)} \lambda s}{\kappa^2(\xi, S)} \to 0,
\end{align*}
which completes the proof of the first part.

Turning to the bootstrap results, we know that that on a sequence of events of the form $\Omega_n=\{\mbb\varepsilon \in \Delta_n\}$ with probability tending to 1, we have
\begin{gather*}
\sup_{\mbb\beta \in \mathscr{B}, k \in N} \sqrt{n}\gamma\|\hat{\mbb\Theta}_k -\mbb\Theta_k\|_1 \geq \sup_{\mbb\beta \in \mathscr{B}, k \in N} \sqrt{n}\gamma\|\hat{\mbb\Theta}_{k, N_k}\|_1 \to 0 \; \text{ and } \;\sup_{\mbb\beta \in \mathscr{B}}|\check{\sigma} -\sigma|\to 0.
\end{gather*}
Here $N_k$ corresponds to the set of noise components of $\mbb\Theta_k$.
Thus on $\Omega_n$, by Lemma~\ref{lem:mu_bd}, we have
\begin{equation} \label{eq:mu_boot_bd}
\sup_{\mbb\beta \in \mathscr{B}, k \in N} \sqrt{n}\gamma\mu(c_1\lambda, \xi^*, \hat{\mbb\Theta}_k, \mb X_{-k}) \to 0
\end{equation}
for any fixed $c_1 >0$ and some $\xi^* > (A_1+1)/(A_1-1)$.
Let $\delta^*_k$ and $\hat{\sigma}^*_k$ be the bootstrap equivalents of $\delta_k$ and $\hat{\sigma}_k$ respectively.
By applying Lemma~\ref{lem:convergence} now with $\mathcal{F}_n = \mbb\varepsilon$, we see that it is enough to show that on $\Omega_n$, for all $\eta > 0$ we have
\begin{align*}
\pr(\sup_{\mbb\beta \in \mathscr{B}, k \in N} |\delta^*_k| >\eta | \mbb\varepsilon) \to 0 \; \text{ and } \; \pr(\sup_{\mbb\beta \in \mathscr{B}, k \in N} |\hat{\sigma}^*_k - \sigma| >\eta | \mbb\varepsilon) \to 0.
\end{align*}
For this it is sufficient to exhibit a sequence $\Omega^*_n=\{\mbb\varepsilon^* \in \Delta^*_n\}$ whose probability tends to 1 such that on $\Omega_n \cap \Omega^*_n$ we have
\begin{align}
\sup_{\mbb\beta \in \mathscr{B}, k \in N} |\delta^*_k| &\to 0 \label{eq:delta_tend2} \\
\sup_{\mbb\beta \in \mathscr{B}, k \in N} |\hat{\sigma}^*_k - \sigma| &\to 0. \label{eq:sigma_tend2}
\end{align}
Let $\tilde{\sigma}^* = \check{\sigma}\|\mbb\varepsilon^*\|_n/\sqrt{n}$.
Define
\[
\tau_*^2 = \sup_{\mbb\beta \in \mathscr{B}, k \in N} \tau^2(\tilde{\sigma}^*, \lambda, \xi^*, \hat{\mbb\Theta}_k,\mb X_{-k}).
\]
Provided $\tau_*^2 \to 0$
we have $(1-\tau_*^2)\lambda(\xi^*-1)/(\xi^*+1) > \sqrt{2\log(p)/n}$ for $n$ sufficiently large. This gives us the equivalent of \eqref{eq:sigma_bd1} and \eqref{eq:Theta_bd1} for $n$ sufficiently large:
\begin{gather}
\sup_{\mbb\beta \in \mathscr{B}, k \in N} \max \bigg(1 - \frac{\hat{\sigma}^*_k}{\tilde{\sigma}^*}, 1 - \frac{\tilde{\sigma}^*}{\hat{\sigma}^*_k}\bigg) \leq \tau_*^2, \label{eq:sigma_bd3} \\
\sup_{\mbb\beta \in \mathscr{B}, k \in N} \| \hat{\mbb\Theta}^*_k - \mbb\Theta^*_k\|_1 \leq c_2 \sup_{\mbb\beta \in \mathscr{B}, k \in N} \mu(\tilde{\sigma}^* \lambda, \xi^*, \hat{\mbb\Theta}_k, \mb X_{-k}). \label{eq:Theta_bd2}
\end{gather}
By Lemma~\ref{lem:chisq}, conditional on $\mbb\varepsilon$, $\sup_{\mbb\beta \in \mathscr{B}}|\tilde{\sigma}^* - \check{\sigma}|\inprob 0$. Thus $\tau_*^2$ tending to 0 and therefore also \eqref{eq:sigma_bd3} and \eqref{eq:Theta_bd2} occur on a sequence of events with probability tending to 1, $\Omega_n^* \cap \Omega_n$, where $\Omega_n^*$ is of the form $\{\mbb\varepsilon^* \in \Delta^*_n\}$. As on $\Omega_n^* \cap \Omega_n$, $\sup_{\mbb\beta \in \mathscr{B}}|\check{\sigma} - \sigma|\to 0$, \eqref{eq:sigma_bd2} gives \eqref{eq:sigma_tend2}. 
Applying \eqref{eq:holder} to \eqref{eq:Theta_bd2} and \eqref{eq:mu_boot_bd} then gives \eqref{eq:delta_tend2}, which completes the proof.

\subsubsection{Proof of Theorem~\ref{thm:single_pval_alt}}
The proof proceeds similarly to that of Theorem~\ref{thm:single_pval_null}. For the first result, we use Lemma~\ref{lem:convergence} with $\mathcal{F}_n$ simply constant.
Thus it suffices to show $\sup_{\mbb\beta \in \mathscr{B}_k}|\hat{\sigma}_k - \sigma| \inprob 0$, $\sup_{\mbb\beta \in \mathscr{B}_k}|\delta_k| \inprob 0$ and $\sup_{\mbb\beta \in \mathscr{B}_k} |\beta_k|\|\mb W_k\|_2|{\sigma'_k}^{-1} - \hat{\sigma}_k^{-1}| \inprob 0$ where $\sigma'_k = \sqrt{\sigma^2 + \|\mb W_k\|_2^2\beta_k^2/n}$. 
To this end, first note that by Lemma~\ref{lem:mu_bd}, for some $\xi' > (A_1+1)/(A_1-1)$ we have $\log(p) \sup_{\mbb\beta \in \mathscr{B}_k} \mu(\lambda, \xi', \mbb\Theta_k, \mb X_{-k}):= \log(p)\mu_k \to 0$.
Let 
\[
\tilde{\sigma}_k=\frac{1}{\sqrt{n}}\|y - \mb X_{-k}\mbb\Theta_k\|_2 = \frac{1}{\sqrt{n}}\|\sigma\mbb\varepsilon + \beta_k\mb W_k\|_2.
\]
Then
\begin{align*}
\tilde{\sigma}_k^2 = \frac{1}{n}\bigg(\sigma^2\|\mbb\varepsilon\|_2^2 + \beta_k^2\|\mb W_k\|_2^2 + 2\sigma\beta_k\|\mb W_k\|_2 Z_k\bigg)
\end{align*}
where $Z_k$ is defined as in \eqref{eq:def_Z}. Since $\beta_k\|\mb W_k\|_2 /\sqrt{n} \to 0$, we have that $\tilde{\sigma}_k \inprob \sigma$ by Lemma~\ref{lem:chisq}.
For later use we also note that
\begin{align} \label{eq:sigma_tilde}
\sqrt{n}(\tilde{\sigma}_k^2 - {\sigma'_k}^2) = \sqrt{n}(\sigma^2\|\mbb\varepsilon\|^2_2/n-1) +2\sigma \beta_k \|\mb W_k\|_2 Z_k / \sqrt{n} =O_P(1).
\end{align}
by the central limit theorem and as $\beta_k \|\mb W_k\|_2 / \sqrt{n} \to 0$.
Thus we have $\tau_k^2 := \sup_{\mbb\beta \in \mathscr{B}_k} \tau^2(\tilde{\sigma}_k, \lambda, \xi', \mbb\Theta_k, \mb X_{-k}) \inprob 0$.
Note that by \eqref{eq:KKT_W},
\[
\frac{1}{n}\beta_k\|\mb X_{-k}^T \mb W_k\|_\infty \leq \frac{1}{\sqrt{n}}\beta_k \gamma \|\mb W_k\|_2 = o\sqrt{\frac{\log(p)}{n}}.
\]
Therefore
\begin{align*}
&\lim_{n \to \infty} \pr\bigg(\frac{\|\mb X_{-k}^T(\sigma\mbb\varepsilon +\beta_k \mb W_k)\|_\infty/n}{\tilde{\sigma}_k} \geq (1-\tau_k^2)\lambda\frac{\xi-1}{\xi+1}\bigg) \\
\geq
&\lim_{n \to \infty} \pr\bigg(\frac{\|\mb X_{-k}^T\mbb\varepsilon\|_\infty/n}{\|\mbb\varepsilon\|_2/\sqrt{n}} \leq \sqrt{2\log(p)/n}\bigg) = 1.
\end{align*}
Thus by Theorem~\ref{thm:Sun} we have that on a sequence of events $\Omega_n=\{\mbb\varepsilon \in \Delta_n\}$ with probability tending to 1, $|\tilde{\sigma}_k - \sigma| \to 0$, $\sqrt{n}\sup_{\mbb\beta \in \mathscr{B}_k}|\hat{\sigma}_k -\tilde{\sigma}_k| \to 0$, and $\log(p) \sup_{\mbb\beta \in \mathscr{B}_k}\|\hat{\mbb\Theta}_k - \mbb\Theta_k\|_1 \to 0$. The latter in conjunction with \eqref{eq:holder} shows $\sup_{\mbb\beta \in \mathscr{B}_k}|\delta_k| \inprob 0$. That $\sup_{\mbb\beta \in \mathscr{B}_k}|\beta_k| \|\mb W_k\|_2|{\sigma'_k}^{-1} - \hat{\sigma}_k^{-1}| \inprob 0$ follows from \eqref{eq:sigma_tilde}.

Now we derive the result concerning the bootstrap test statistic.
Note that on $\Omega_n$,
\begin{equation} \label{eq:boot_est}
\sqrt{n}\gamma \sup_{\mbb\beta \in \mathscr{B}_k} \|\hat{\mbb\Theta}_{k,T^c}\|_1 \leq \sqrt{n}\gamma \{\|\mbb\Theta_{k,T^c}\|_1 + \sup_{\mbb\beta \in \mathscr{B}_k}\|\hat{\mbb\Theta}_k - \mbb\Theta_k\|_1\} \to 0,
\end{equation}
and so also $\sqrt{n}\gamma\sup_{\mbb\beta \in \mathscr{B}_k}\mu(c_1\lambda, \xi^*, \hat{\mbb\Theta}_k, \mb X_{-k}) \to 0$ for any fixed $c_1 >0$ and some $\xi^* > (A_1+1)/(A_1-1)$. The rest of the proof then proceeds exactly as the proof for the bootstrap statistic in Theorem~\ref{thm:single_pval_null}.

\subsection{Proof of Corollary~\ref{cor:power}}
The proof is essentially identical to that of Theorem~\ref{thm:single_pval_alt}, but with $\beta_k\mb W_k$ replaced by $\mb w_\gamma$ and $\mbb\Theta_k$ replaced by $\mbb\beta + \mbb\psi_\gamma$.

\subsection{Proof of Theorem~\ref{thm:interactions}}
By Corollary 1 of \citet{Raskutti2010}, with probability tending to 1 we
have $\kappa(\xi)\geq \phi(\xi)>\phi/8>0$. By Lemma~\ref{lem:inter_bound}
and the fact that $\sqrt{\E(f_1^2)} = o(n^{-1/6})$ we see that $\|\mb f\|_2 = o_P(n^{1/3})$. Indeed, we have $\pr(\|\mb f\|_2/n^{1/3} \leq n^{1/6} \sqrt{2E(f_1^2)}) \to 1$ as $n \to \infty$. Also, with probability tending to 1,
\[
\frac{1}{n}\|\mb X^T\mb f\|_\infty \leq c_1\frac{\sqrt{\log(p)}}{n^{1/3}} \cdot \frac{1}{\sqrt{n}}\|\mb f\|_2 \leq c_1\frac{\sqrt{\log(p)}}{n^{1/3}} \cdot o(n^{-1/6}) = o(\sqrt{\log(p)/n}),
\]
i.e.\ $\|\mb X^T\mb f\|_\infty/n = o_P(\sqrt{\log(p)/n})$.
We now temporarily make the dependence of $\mb Z$ and $p$ on $n$ explicit by writing $\mb Z^{(n)}$ and $p_n$ respectively, in order to explain the structure of the argument to follow. From the above, we have a sequence of sets $\Lambda_n \subseteq \R^{n \times p_n}$ for which $\pr(\mb Z^{(n)} \in \Lambda_n) \to 1$, and on the respective events $\|\mb f\|_2 = o(n^{1/3})$ and $\frac{1}{n}\|\mb X^T\mb f\|_\infty = o(\sqrt{\log(p)/n})$ (uniformly).
Let $\hat{\sigma}=\|\mb y - \mb X\mbb\beta'\|_2/\sqrt{n}$. In relation to Corollary~\ref{cor:power}, we will take $\gamma=c_1\sqrt{\log(p)}/n^{1/3}$ sufficiently large such that the Lasso regression of $\mb f$ on $\mb X$ produces the zero vector.
By Lemma~\ref{lem:convergence}, it suffices to show that for each $\eta > 0$,
\begin{align}
\sup_{\mb Z^{(n)} \in \Lambda_n}\pr\big(\sup_{\mbb\beta \in \mathscr{B}} |\hat{\sigma} - \sigma|>\eta \;|\,\mb Z^{(n)}\big) &\to 0 \notag \\
\sup_{\mb Z^{(n)} \in \Lambda_n} \pr\big(\sup_{\mbb\beta \in \mathscr{B}}  \|\mb X^T \mb f\|_\infty\|\mbb\beta - \mbb\beta'\|_1/\|\mb f\|_2>\eta \;|\, \mb Z^{(n)}\big) &\to 0 \notag\\
\sup_{\mb Z^{(n)} \in \Lambda_n} \pr\big(\sup_{\mbb\beta \in \mathscr{B}} \|\mb f\|_2|\sigma'^{-1}-\hat{\sigma}^{-1}|>\eta \; |\,\mb Z^{(n)}\big) & \to 0. \label{eq:fbd2}
\end{align}
Here $\sigma' := \sqrt{\sigma^2 + \|\mb f\|_2^2/n}$.
The remainder of the argument to arrive at the first result is essentially identical to that in Theorem~\ref{thm:single_pval_alt}, with $\mb f$ playing the role of $\beta_k \mb W_k$, and with the probabilities being conditional on $\mb Z^{(n)}$. The only difference is that in place of \eqref{eq:sigma_tilde} (which leads to the equivalent of \eqref{eq:fbd2}), we have
\begin{align*}
n^{1/3} |\tilde{\sigma}^2 - \sigma'^2| = |n^{1/3}\sigma^2(\|\mbb\varepsilon\|_2^2/n-1) +2\sigma \mb f^T \mbb\varepsilon/n^{2/3}|,
\end{align*}
where $\tilde{\sigma}=\|\sigma\mbb\varepsilon +\mb f\|_2/\sqrt{n}$.
Since for all $\mb Z^{(n)} \in \Lambda_n$, $\|\mb f\|_2 = o(n^{1/3})$ uniformly, it is straightforward to show that 
\begin{align*}
\sup_{\mb Z^{(n)} \in \Lambda_n} \pr\big(\sup_{\mbb\beta \in \mathscr{B}} \|\mb f||_2|\sigma'-\tilde{\sigma}|>\eta \; |\,\mb Z^{(n)}\big) & \to 0.
\end{align*}
for any $\eta > 0$, which then leads to \eqref{eq:fbd2}.

The bootstrap result is simpler. We know there is a sequence of events depending only on $(\mb Z, \mbb\varepsilon)$ on which
\[
\sup_{\mbb\beta \in \mathscr{B}}\|\mbb\beta - \check{\mbb\beta}\|_1 \leq c_2 \sigma s\sqrt{\log(p)/n} / \phi^2(\xi).
\]
Thus on the same sequence of events
\begin{align*}
\sup_{\mbb\beta \in \mathscr{B}} \sqrt{n} \cdot c_1\sqrt{\log(p)}/n^{1/3} \cdot \mu(\lambda, \xi, \check{\mbb\beta}, \mb X) \to 0,
\end{align*}
from which the result follows by arguing along the lines of the second part of the proof of Theorem~\ref{thm:single_pval_null}.
\subsection{Technical lemmas}
\begin{lem} \label{lem:event_prop}
Consider the setup of Theorem~\ref{thm:maximise_pval} and its proof. Recall that
\begin{align*}
\Lambda_1 &= \bigg\{ \mbb\zeta : \max_j \frac{|\mb X_j^T (\mb I - \mb P)\mbb\zeta|}{\|(\mb I - \mb P)\mbb\zeta\|_2} \leq \lambda_0 \|(\mb I - \mb P)\mb X_j\|_2/\sqrt{2} \bigg\}, \\
\Lambda_2 &= \bigg\{ \mbb\zeta : \max_j \frac{|\mb X_j^T \mb P \mbb\zeta|}{\|\mbb\zeta\|_2} \leq \lambda_0 \| \mb P \mb X_j\|_2/\sqrt{2} \bigg\}, \\
\Lambda_3 &= \{\mbb\zeta : \|\mbb\zeta\|_2/\sqrt{n} \leq\sqrt{2}\},\\
\Lambda &= \Lambda_1 \cap \Lambda_2 \cap \Lambda_3,
\end{align*}
with $\lambda_0=2\sqrt{\log(p/\eta)/n}$.
On the event $\Omega = \{\mbb\varepsilon \in \Lambda\}$
we have the following properties.
\begin{enumerate}[(i)]
\item $\sgn(\mbb\beta'_S)= \sgn(\mbb\beta_S)$ and $\min_{j \in S} \beta'_j / \beta_j > 1 - 1/(2\sqrt{2})$.
\item $\hat{S}^{(s)}=S$ and $\sgn(\check{\mbb\beta}^{(s)}) = \sgn(\mbb\beta_S)$.
\item $2\sqrt{2} \sigma \min_{j \in S} \check{\beta}^{(s)}_j/\beta_j > \check{\sigma}^{(s)}$.
\end{enumerate}
Furthermore
\begin{equation} \label{eq:pr_Omega}
\pr(\Omega) \geq 1 - \frac{2(1+r_{n-s})\eta}{\sqrt{\pi\log(p / \eta)}} -e^{-n/8}.
\end{equation}
\end{lem}
\begin{proof}
First we bound $\pr(\Omega)$. From Lemma~\ref{lem:t_tail} and the union bound we have $\pr(\mbb\varepsilon \in \Lambda_2) \leq (1+r_n)\eta/\sqrt{\pi\log(p/\eta)}$. Next note that
\[
\frac{\mb X_j^T(\mb I - \mb P)\mbb\varepsilon}{\|(\mb I - \mb P)\mbb\varepsilon\|_2\|(\mb I - \mb P)\mb X_j\|_2} \eqdist \frac{\mb a^T \mbb\zeta}{\|\mbb\zeta\|_2\|\mb a\|_2}
\]
where $\mb a \in \R^{n-s}$ and $\eqdist$ denotes equality in distribution. Thus Lemma~\ref{lem:t_tail} gives $\pr(\mbb\varepsilon \in \Lambda_1) \leq (1+r_{n-s})\eta/\sqrt{\pi\log(p/\eta)}$. Finally, Lemma~\ref{lem:chisq} gives
$\pr(\mbb\varepsilon \in \Lambda_3) \geq  1- e^{-n/8}$. 

Now turning to (i), observe that as $\|(\mb I - \mb P) \mbb\zeta\|_2 \leq \|\mbb\zeta\|_2$ and
\[
 (\|\mb P \mb X_j\|_2 + \|(\mb I - \mb P)\mb X_j\|_2)/\sqrt{2} \leq \|\mb X_j\|_2= \sqrt{n},
\]
we have
\[
 \Omega \subseteq \bigg\{\frac{\|\mb X^T\mbb\varepsilon\|_\infty}{\sqrt{n}\|\mbb\varepsilon\|_2} \leq  \lambda_0 \bigg\}.
\]
Since
\begin{equation*} 
\frac{1}{2}\frac{(\xi+1)\lambda s}{\kappa^2(\xi, S)} \geq \mu(\lambda, \xi, \mbb\beta, \mb X),
\end{equation*}
\eqref{eq:A_cond} ensures that $(1-\tau^2)\lambda(\xi-1)/(\xi+1) \geq \lambda_0$. Note also that \eqref{eq:A_cond} gives $\tau^2 \leq 1/5$. Thus by Theorem~\ref{thm:Sun}, the fact that $\mbb\varepsilon \in \Lambda_3$ and \eqref{eq:betamin} we have 
\begin{equation} \label{eq:beta_l2}
\|\mbb\beta'-\mbb\beta\|_2 \leq \frac{2\tilde{\sigma} \xi \sqrt{s} \lambda }{(1-\tau^2)\phi^2(\xi)} \leq \frac{\tilde{\sigma}}{4\sigma} \min_{j \in S} |\beta_j| \leq \frac{1}{2\sqrt{2}}\min_{j \in S} |\beta_j|,
\end{equation}
where $\tilde{\sigma}=\|\mbb\varepsilon\|_2/\sqrt{n}$. This shows $\sgn(\mbb\beta'_S)=\sgn(\mbb\beta_S)$. Next
\[
\min_{j \in S} \frac{\beta'_j}{\beta_j} \geq 1 - \frac{\|\mbb\beta'-\mbb\beta\|_2}{\min_{j \in S}|\beta_j|} \geq 1 - \frac{1}{2\sqrt{2}},
\]
which shows (i).
Now
\begin{align*}
\min_{j \in S} |\beta'_j| - \max_{j \in S^c} |\beta'_j| & \geq \min_{j \in S} |\beta_j| - \max_{j \in S} |\beta_j -\beta'_j| - \max_{j \in S^c} |\beta'_j| \\
& > \min_{j \in S} |\beta_j| - \sqrt{2}\|\mbb\beta-\mbb\beta'\|_2 >0.
\end{align*}
Thus $\hat{S}^{(s)}=S$. By Theorem~\ref{thm:Sun}, the $\ell_2$ bound \eqref{eq:beta_l2} is also satisfied by $\check{\mbb\beta}^{(s)}$, which then shows (ii). Also note that $\check{\sigma}^{(s)} \leq \tilde{\sigma}/(1-\tau^2)$ by Theorem~\ref{thm:Sun}. From \eqref{eq:beta_l2} we have
\begin{align*}
2\sqrt{2} \frac{\sigma}{\check{\sigma}^{(s)}} \min_{j\in S} \frac{\check{\beta}^{(s)}_j}{\beta_j} &\geq 2\sqrt{2} \frac{\sigma}{\check{\sigma}^{(s)}} \bigg(1 - \frac{\tilde{\sigma}}{4\sigma}\bigg) \\
&\geq 2 \frac{\tilde{\sigma}}{\check{\sigma}^{(s)}} - \frac{\tilde{\sigma}}{\sqrt{2}\check{\sigma}^{(s)}}.
\end{align*}
We see that the RHS is at least 1 when $\tau^2 \leq (\sqrt{2}-1)/(2\sqrt{2}-1)$, which then shows (iii) as $(\sqrt{2}-1)/(2\sqrt{2}-1) < 1/5$.
\end{proof}

\begin{lem} \label{lem:indep}
Consider the setup of Theorem~\ref{thm:maximise_pval} and its proof. Conditional on the event $\Omega$, $\hat{\mb R}(\mbb\beta, \sigma \mbb\varepsilon)$ and $(\check{\mbb\beta}^{(s)}, \check{\sigma}^{(s)})$ are independent.
\end{lem}
\begin{proof}
Write $v_N(\mbb\zeta) = N\|(\mb I -\mb P)\mbb\zeta\|_2$. First we claim that on $\Omega$, $\hat{\mb R}(\mbb\beta, \sigma \mbb\varepsilon)$ depends only on $(\mb I - \mb P)\mbb\varepsilon / v_1(\mbb\varepsilon)$. To this end, we argue that for $\mbb\zeta \in \Lambda$ with $v_N(\mbb\zeta) > 1/(2\sqrt{2})$,
\begin{align}
\hat{\mb R}(\mbb\beta, \sigma \mbb\zeta) &= \hat{\mb R}(\mbb\beta + \sigma\mbb\delta(\mbb\zeta), \sigma(\mb I - \mb P)\mbb\zeta) \notag\\
&= \hat{\mb R}(\mbb\beta, \sigma(\mb I - \mb P)\mbb\zeta) \label{eq:R_eq1}\\
&= \hat{\mb R}(\mbb\beta, \sigma(\mb I - \mb P)\mbb\zeta/v_N(\mbb\zeta)), \label{eq:R_eq2}
\end{align}
where $\mbb\delta(\mbb\zeta) = ((\mb X_S^T \mb X_S)^{-1}\mb X_S^T\mbb\zeta,
0, \ldots, 0) \in \R^p$. Note the validity of the above inequalities would prove the initial claim for all $\mbb\zeta \in \Lambda$ with $v_N(\mbb\zeta) > 1/(2\sqrt{2})$ for each fixed $N$. However, since $\cup_{N>0} \{\mbb\zeta : v_N(\mbb\zeta) > 1/(2\sqrt{2})\}=\R^n$ the initial claim would then have to be true for all $\mbb\zeta \in \Lambda$.
We now set about proving these equalities.
The first equality is clear by
definition of $\hat{\mb R}$. Turning to the second \eqref{eq:R_eq1}, let
$\mbb\zeta \in \Lambda$ and denote by
  $\|\cdot\|$ the operator norm. From Lemma~\ref{lem:event_prop} we have
\begin{align*}
\sigma \|\mbb\delta(\mbb\zeta)\|_\infty &= \sigma \|(\mb X_S^T \mb X_S)^{-1}\mb X_S^T\mbb\zeta \|_\infty \\
&\leq \sigma (\|\mbb\zeta\|_2/\sqrt{n}) \|(\mb X_S^T \mb X_S)^{-1}\| \,\, \sqrt{s}\| \mb X_S^T \mb P \mbb\zeta\|_\infty / (\|\mbb\zeta\|_2/\sqrt{n}) \\
&\leq \frac{2\sigma \sqrt{s\log(p/\eta)/n}}{\phi^2(\xi)} \leq \min_{j \in S} |\beta_j|/10,
\end{align*}
using the facts that $\xi > 1$ and $A > \sqrt{2}$ and \eqref{eq:betamin} in the final line. Note that by Lemma~\ref{lem:event_prop}, $\Lambda \subseteq \Lambda_{\lambda, t}$ with the latter defined in Theorem~\ref{thm:lambda_fixed} and $t=1-1/(2\sqrt{2})$. Now clearly if $\mbb\zeta \in \Lambda$ then also $(\mb I - \mb P)\mbb\zeta \in \Lambda$. An application of Theorem~\ref{thm:lambda_fixed} then gives the desired equality \eqref{eq:R_eq1}.
Equality \eqref{eq:R_eq2} then follows from a further application of Theorem~\ref{thm:lambda_fixed} noting that $v_N(\mbb\zeta) < 1/(1-t)= 2\sqrt{2}$ by assumption.

Next we examine $\check{\mbb\beta}^{(s)}$ and $\check{\sigma}^{(s)}$. Note that on $\Omega$, the former is simply the Lasso estimate from regressing on $\mb X_S$ and the latter is the resulting normalised root-RSS. Write $\check{\mbb\beta} = \check{\mbb\beta}^{(s)}$ and $\check{\sigma} = \check{\sigma}^{(s)}$.
The least squares part of the Lasso objective decomposes as
\[
\|\mb X_S\mbb\beta_S  + \mbb\varepsilon - \mb X_S \mb b\|_2 = \{\|\mb X_S\mbb\beta_S  + \mb P\mbb\varepsilon - \mb X_S \mb b\|_2^2 + \|(\mb I - \mb P)\mbb\varepsilon\|_2^2\}^{1/2}.
\]
Thus it is clear that the fitted values $\mb X_S \check{\mbb\beta}$ do not depend on $(\mb I - \mb P)\mbb\varepsilon / v_1(\mbb\varepsilon)$. This then implies that $\check{\sigma}$ does not depend on $(\mb I - \mb P)\mbb\varepsilon / v_1(\mbb\varepsilon)$ since it is determined by $\|\mbb\varepsilon\|_2^2 = \|\mb P\mbb\varepsilon\|_2^2 + \|(\mb I- \mb P)\mbb\varepsilon\|_2^2$ and $\mb P\mbb\varepsilon$.

Now observe that
\begin{align*}
\mb P \mbb\varepsilon &\ci \big(\|(\mb I - \mb P)\mbb\varepsilon\|_2, \, (\mb I - \mb P)\mbb\varepsilon / \|(\mb I - \mb P)\mbb\varepsilon\|_2\big) \\
\|(\mb I - \mb P)\mbb\varepsilon\|_2 &\ci (\mb I - \mb P)\mbb\varepsilon / \|(\mb I - \mb P)\mbb\varepsilon\|_2,
\end{align*}
so  $\mb P \mbb\varepsilon,\, \|(\mb I - \mb P)\mbb\varepsilon\|_2, \, (\mb
I - \mb P)\mbb\varepsilon / \|(\mb I - \mb P)\mbb\varepsilon\|_2$ are
jointly independent. Let $E_1 = \{\hat{\mb R}(\mbb\beta,
\sigma\mbb\varepsilon) \in B_1\}$, $E_2 = \{(\mb X_S \check{\mbb\beta}, \,
\check{\sigma}) \in B_2\}$, where $B_1 \subseteq \R^n$ and $B_2 \subseteq \R^{n+1}$ are arbitrary Borel sets. Let $\Omega_k = \{\mbb\varepsilon \in \Lambda_k\}$, $k=1,2,3$. From the above
\begin{align*}
\pr(E_1, E_2 | \Omega_1, \Omega_2, \Omega_3) = \frac{\pr(E_1, \Omega_1)}{\pr(\Omega_1)} \frac{\pr(E_2, \Omega_2, \Omega_3)}{\pr(\Omega_2,\Omega_3)} = \pr(E_1|\Omega_1)\pr(E_2|\Omega_2, \Omega_3).
\end{align*}
The conditional probabilities on the RHS remain unchanged if we modify the conditioning event to be $(\Omega_1, \Omega_2, \Omega_3)$ since $E_1 \ci (\Omega_2,\Omega_3)$ and $\E_2 \ci \Omega_1$. This completes the proof.
\end{proof}

\begin{lem} \label{lem:mu_bd}
Given a sequence of collections of matrices $\mb M_{k,n} \in \R^{n \times p_n}$, $k \in N_n$, $n=1,2,\ldots$,
 suppose there exists a sequence of collections of sets $S_{k, n} \subseteq \{1,\ldots,p_n\}$, tuning parameters $\lambda_n = A\sqrt{\log(p)/n}$ and $\xi > c$ for constants $A, c>0$ such that the follow holds:
\begin{align*}
\sup_{k \in N_n}\frac{|S_{k,n}| \sqrt{\log(p_n)^2/n}}{\kappa^2(\xi, S_{k,n}, \mb M_{k,n})} \to 0.
\end{align*}
Moreover, suppose the collection of sequences $\mbb\beta_{k,n} \in \R^{p_n}$ is such that
\[\sup_{k \in N_n} \sqrt{\log(p_n)}\|\mbb\beta_{n, S_{k, n}^c}\|_1 \to 0.
\]
Then there exists a $\xi' > c$ such that $\sqrt{\log(p_n)}\mu(c_1\lambda_n, \xi', \mbb\beta_{k, n}, \mb M_{k,n})\to 0$ for any $c_1>0$ (where the function $\mu$ is defined in \eqref{eq:mu_def}).
\end{lem}
\begin{proof}
Let $\xi' = (\xi + c)/2 > c$ and let $\nu' > 0$ be given by $(\xi' + \nu')/(1-\nu') =\xi$. Then we have
\begin{align*}
& \sup_{k \in N_n} \sqrt{\log(p_n)}\mu(c_1\lambda_n, \xi', \mbb\beta_{k,n}, \mb M_{k,n}) \\
\leq &(\xi' +1) \sqrt{\log(p_n)} \max_{k \in N_n} \max\bigg[\frac{\|\mbb\beta_{k, n, S_{k,n}^c}\|_1}{\nu'},\, \frac{c_1\lambda_n|S_{k,n}|/\{2(1-\nu')\}}{\kappa^2(\xi, S_{k,n}, \mb M_{k,n})} \bigg] \to 0.
\end{align*}
\end{proof}

\begin{lem} \label{lem:chisq}
Let $Z_n \sim \chi^2_n$. We have the following tail bounds
\citep[pg. 29]{Boucheron2013}:
\begin{align*}
\pr(Z_n > n + 2\sqrt{n\gamma} + 2\gamma) \leq e^{-\gamma},
\end{align*}
whence taking $\gamma=n/8$,
\begin{align*}
\pr(\sqrt{Z_n/n} > \sqrt{2}) < \pr[Z_n > n\{1+2(1/\sqrt{8}+1/8)\}] \leq e^{-n/8}.
\end{align*}
\end{lem}

\begin{lem} \label{lem:convergence}
Let $T_{k,n} \in \R$, $k \in N_n$, $n=1,2,\ldots$ be a collection of random variables which we may decompose as 
\[
T_{k,n} = a_{k,n} Z_{k,n} + b_{k,n}
\]
where each $Z_{k,n}$ is identically distributed with continuous distribution function $F$. Suppose further that each $Z_{k,n}$ is independent of the random elements $\mathcal{F}_n$, and for all $\delta > 0$, $\sup_{k \in N_n}\pr(|b_{k,n}-d_{k,n}|>\delta|\mathcal{F}_n) \inprob 0$  for some random variables $d_{k,n}\in \R$ that are functions of $\mathcal{F}_n$, and $\sup_{k \in N_n} \pr( |a_{k,n}-1|>\delta|\mathcal{F}_n) \inprob 0$. Then
\[
\sup_{k \in N_n} \sup_{x \in \R} |\pr(T_{k,n} \leq x|\mathcal{F}_n) - F(x-d_{k,n})| \inprob 0.
\]
\end{lem}
\begin{proof}
First note that for all $\epsilon > 0$, there exists $\delta >0$ such that for all $c_1, c_2 \in [-\delta, \delta]$,
\begin{equation} \label{eq:unif_cont}
\sup_{x \in \R}\abs{F\bigg(\frac{x+c_1}{1+c_2}\bigg) - F(x)} <\epsilon.
\end{equation}
Indeed, the function $G(x, c_1, c_2) :=F\{(x+c_1)/(1+c_2)\}$ is uniformly continuous on $\R \times [-\delta', \delta']^2$ for $0<\delta'<1$ sufficiently small (since $F$ is uniformly continuous and compositions of uniformly continuous functions are continuous) and the LHS of \eqref{eq:unif_cont} is $\sup_{x \in \R} |G(x, c_1, c_2)-G(x, 0, 0)|$.

Hence given $\epsilon >0$, let $\delta > 0$ be such that the LHS of \eqref{eq:unif_cont} is at most $\epsilon/2$.
Then
\begin{align*}
\sup_{x \in \R} |\pr(T_{k,n} \leq x|\mathcal{F}_n) - F(x-d_{k,n})| &= \sup_{x \in \R} |\pr(Z_{k,n} \leq (x + d_{k,n}-b_{k,n})/a_{k,n}|\mathcal{F}_n) - F(x)| \\
&\leq \sup_{x \in \R}\sup_{c_1, c_2 \in [-\delta, \delta]} |\pr(Z_{k,n} \leq (x + c_1)/(1+c_2)|\mathcal{F}_n) - F(x)|\\
&\qquad + \pr(|b_{k,n}-d_{k,n}|>\delta|\mathcal{F}_n) + \pr(|a_{k,n}-1|>\delta|\mathcal{F}_n) \\
&\leq \epsilon/2 + \pr(|b_{k,n}-d_{k,n}|>\delta|\mathcal{F}_n) + \pr(|a_{k,n}-1|>\delta|\mathcal{F}_n).
\end{align*}
Thus
\begin{align*}
& \pr\{\sup_{k \in N_n} \sup_{x \in \R} |\pr(T_{k,n} \leq x|\mathcal{F}_n) - F(x-d_{k,n})| > \epsilon\} \\
&\leq \pr\{\sup_{k \in N_n}\pr(|b_{k,n}-d_{k,n}|>\delta|\mathcal{F}_n) >\epsilon/4\} +  \pr\{\sup_{k \in N_n}\pr(|a_{k,n}-1|>\delta|\mathcal{F}_n)>\epsilon/4\}\\
&\to 0 \qquad \text{as } n \to \infty.
\end{align*}
\end{proof}

\begin{lem} \label{lem:inter_bound}
Consider the setup of Theorem~\ref{thm:interactions}. There exists constants $c_1, c_2 >0$ such that for all $t<n^{1/4}$
\begin{align}
\pr(1-tn^{-1/4} \leq \|\mb f\|_2^2/\E(\|\mb f\|_2^2) \leq 1+tn^{-1/4}) &\leq c_1e^{-c_2t^2}, \label{eq:bd1}\\
\pr\bigg( \frac{1}{\sqrt{n}}\frac{\|\mb X^T \mb f\|_\infty}{\|\mb f\|_2} \leq c_1\frac{\sqrt{\log(p)}}{n^{1/3}}\bigg) \to 0. \label{eq:lim2}
\end{align}
\end{lem}
\begin{proof}
Although in Theorem~\ref{thm:interactions} certain assumptions are placed on $\E(f_1^2)$, since the probabilities above do not depend on $\E(f_1^2)$ here we may assume $\E(f_1^2)=1$.
Fix $j \in \{1,\ldots,p\}$. Using the eigendecomposition $\mb P \mb D \mb P^T = \mbb\Sigma_{S,S}^{1/2}\mbb{S,S}^{1/2}$ (where $\mb D$ is diagonal and $\mb P$ is orthogonal) and writing $\eqdist$ for equality in distribution, we have
\[
(f_1, Z_{1j}) = (\mb z_{1,S}^T \mb B \mb z_{1,S},\, Z_{1j}) = (\mb z_{1,S}^T \Sigma_{S,S}^{-1/2} \mb P \mb D \mb P^T  \Sigma_{S,S}^{-1/2}\mb z_{1,S},\, Z_{1j}) \eqdist (\mb u^T \mb D \mb u, v)
\]
where $(\mb u, v)$ is multivariate Gaussian with $\mb u \sim \mathcal{N}_s(\mb 0, \mb I)$ and $\Var(v)=1$.
Let the diagonal entries of $\mb D$ be $\mbb\theta$.

First we show \eqref{eq:bd1}. Note that
\begin{align}
\E(f_1^2)=\E\Big\{\Big(\sum_j u_j^2 \theta\Big)^2\Big\} &= \sum_j \theta_j^2\E(u_j^4) + \sum_{j \neq k}\theta_j\theta_k\E(u_j^2)\E(u_k^2) \notag\\
 &= 2\|\mbb\theta\|_2^2 + \Big(\sum_j \theta_j\Big)^2 =1. \label{eq:expect}
\end{align}
Thus in particular, $\|\mbb\theta\|_\infty \leq 1/\sqrt{2}$, $\|\mbb\theta\|_2^2 \leq 1/2$ and $\big|\sum_j \theta_j\big| \leq 1$.

Now with a view to applying Lemma~\ref{lem:tail_bd} below, consider
\begin{align*}
\E \exp \Big\{\frac{1}{4}\Big| \Big(\sum_j u_j^2 \theta_j\Big)^2\Big|^{1/2}\Big\} &\leq  \E\exp\Big( \sum_j u_j^2 \theta_j/4\Big) +\E\exp\Big( -\sum_j u_j^2 \theta_j/4\Big)\notag\\
&= \prod_j (1-\theta_j/2)^{-1/2} + \prod_j (1+\theta_j/2)^{-1/2}
\end{align*}
using the fact that $u_j^2 \sim \chi^2_1$.
Note that
\[
(1-\theta_j/2)^{-1} = 1 + \frac{\theta_j}{2} + \frac{\theta_j^2}{4(1-\theta_j/2)} \leq 1 + \frac{\theta_j}{2} + \theta_j^2.
\]
Thus, by the AM--GM inequality we have 
\begin{align*}
\prod_j (1-\theta_j/2)^{-1} &\leq \prod_j (1 + \theta_j/2 + \theta_j^2) \\
&\leq \bigg( \frac{1}{s} \sum_{j} (1 + \theta_j/2 + \theta_j^2) \bigg)^s \\
&\leq (1+1/s)^s< e,
\end{align*}
using \eqref{eq:expect}. Similarly, $\prod_j (1+\theta_j/2)^{-1} < e$. Putting things together, we see that
\begin{align}
\E \exp(|f_1^2 - \E(f_1^2)|^{1/2}/4) \leq e^{1/4}\E \exp(|f_1|/4) <e^{1/4} \sqrt{2e}. \label{eq:f_1_bd}
\end{align}
Lemma~\ref{lem:tail_bd} then immediately gives \eqref{eq:bd1}.

To show \eqref{eq:lim2}, we first obtain a tail bound for $|\mb f^T \mb Z_j|/n$. Observe that if $(w_1, w_2, w_3)$ is multivariate Gaussian with zero-mean, then since $(w_1, w_2, w_3) \eqdist -(w_1, w_2, w_3)$, we have 
\[
\E(w_1w_2w_3)=\E\{(-w_1)(-w_2)(-w_3)\}=-\E(w_1w_2w_3)=0.
\]
Then, by H\"older's inequality
\begin{align*}
\E\exp(|f_1 Z_{1j}|^{2/3}/4) &= \E\exp\Big(\Big |\sum_j v u_j^2\theta_j \Big|^{2/3}\big/ 4\Big) \\
&\leq \sum_{r=0}^\infty \E\bigg(\Big|\sum_j u_j^2\theta_j/4\Big|^{2r/3}|v/2|^{2r/3} \bigg)\Big/r! \\
&\leq \sum_{r=0}^\infty [\E\{(v^2/4)^{r}\}]^{1/3} \bigg\{\E\Big(\Big|\sum_j u_j^2\theta_j/4\Big|^{r}\Big) \bigg\}^{2/3}\Big/ r! \\
&\leq \sum_{r=0}^\infty \E\Big(\Big|\sum_j u_j^2\theta_j/4\Big|^{r}\Big) \Big/ r! + \sum_{r=0}^\infty \E\{(v^2/4)^{r}\} /r! \\
&= \E\exp(|f_1|/4) + \E(v^2/4) \leq \sqrt{2e} + \sqrt{2},
\end{align*}
using \eqref{eq:f_1_bd} in the final line. Note $\|\mb f\|_2/\sqrt{n} \inprob 1$ by \eqref{eq:bd1}, and $\|\mb Z_j\|/\sqrt{n} \inprob 1$ by Lemma~\ref{lem:chisq}. Thus by Lemma~\ref{lem:tail_bd}, $\pr(|\mb f^T \mb X_j|/(\sqrt{n}\|\mb f\|_2) \geq t) \leq c_1\exp(-c_2n^{2/3}t^2)$ for $t \in [0, 1]$ and some constants $c_1, c_2>0$. Thus for $c_3$ with $c_2c_3^2-1 >0$,
\begin{align*}
\pr(\|\mb Z^T \mb f\|_\infty/(\sqrt{n}\|\mb f\|_2) \geq c_3\sqrt{\log(p)}/n^{1/3}) \leq c_1 p\exp(-c_2c_3^2\log(p))=c_1 p^{-(c_2c_3^2-1)} \to 0
\end{align*}
as $p \to \infty$ and $\log(p)/n^{2/3} \to 0$.
\end{proof}

\begin{lem}[Lemma B.4 of \citet{Hao2014}] \label{lem:tail_bd}
Let $W_1, \ldots, W_n$ be independent random variables with zero mean such that $\E(\exp(c_1|W_i|^\alpha)) \leq c_2$ for constants $c_1, c_2>0$ and $\alpha \in (0,1]$. Then there exist constants $c_3,c_4 > 0$ such that for $t \in [0,1]$,
\[
\pr\bigg( \frac{1}{n}\bigg| \sum_{i=1}^n W_i\bigg| \geq t \bigg) \leq c_3 \exp(-c_4 n^{\alpha}t^2).
\]
\end{lem}

\bibliographystyle{abbrvnat}

\begin{thebibliography}{48}
\providecommand{\natexlab}[1]{#1}
\providecommand{\url}[1]{\texttt{#1}}
\expandafter\ifx\csname urlstyle\endcsname\relax
  \providecommand{\doi}[1]{doi: #1}\else
  \providecommand{\doi}{doi: \begingroup \urlstyle{rm}\Url}\fi

\bibitem[Belloni et~al.(2011)Belloni, Chernozhukov, and Wang]{Belloni2011}
A.~Belloni, V.~Chernozhukov, and L.~Wang.
\newblock Square-root lasso: pivotal recovery of sparse signals via conic
  programming.
\newblock \emph{Biometrika}, 98\penalty0 (4):\penalty0 791--806, 2011.

\bibitem[Bickel et~al.(2009)Bickel, Ritov, and Tsybakov]{bickel07dantzig}
P.~Bickel, Y.~Ritov, and A.~Tsybakov.
\newblock {Simultaneous Analysis of {Lasso} and {Dantzig} selector}.
\newblock \emph{Ann.\ Statist.}, 37:\penalty0 1705--1732, 2009.

\bibitem[Breiman(2001)]{breiman01random}
L.~Breiman.
\newblock {Random Forests}.
\newblock \emph{Machine Learning}, 45:\penalty0 5--32, 2001.

\bibitem[B{\"u}hlmann(2013)]{pb13}
P.~B{\"u}hlmann.
\newblock {Statistical significance in high-dimensional linear models}.
\newblock \emph{Bernoulli}, 19:\penalty0 1212--1242, 2013.

\bibitem[B{\"u}hlmann and {van de Geer}(2011)]{buhlmann2011statistics}
P.~B{\"u}hlmann and S.~{van de Geer}.
\newblock \emph{{Statistics for High-Dimensional data: Methods, Theory and
  Applications}}.
\newblock Springer, 2011.

\bibitem[B{\"u}hlmann and van~de Geer(2015)]{pbvdg15}
P.~B{\"u}hlmann and S.~van~de Geer.
\newblock High-dimensional inference in misspecified linear models.
\newblock \emph{Electron.\ J.\ Statist.}, 9:\penalty0 1449--1473, 2015.

\bibitem[B{\"u}hlmann et~al.(2014)B{\"u}hlmann, Kalisch, and
  Meier]{Buehlmann2014}
P.~B{\"u}hlmann, M.~Kalisch, and L.~Meier.
\newblock High-dimensional statistics with a view toward applications in
  biology.
\newblock \emph{Annual Review of Statistics and Its Application}, 1:\penalty0
  255--278, 2014.

\bibitem[Camponovo(2014)]{Camponovo2014}
L.~Camponovo.
\newblock On the validity of the pairs bootstrap for lasso estimators.
\newblock \emph{Biometrika, to appear}, 2014.

\bibitem[Chatterjee and Lahiri(2010)]{Chatterjee2010}
A.~Chatterjee and S.~Lahiri.
\newblock Asymptotic properties of the residual bootstrap for lasso estimators.
\newblock \emph{Proc.\ Am.\ Math.\ Soc.}, 138\penalty0
  (12):\penalty0 4497--4509, 2010.

\bibitem[Chatterjee and Lahiri(2011)]{Chatterjee2011}
A.~Chatterjee and S.~N. Lahiri.
\newblock Bootstrapping lasso estimators.
\newblock \emph{J.\ Am.\ Statist.\ Ass.}, 106\penalty0
  (494):\penalty0 608--625, 2011.

\bibitem[Davison and Hinkley(1997)]{Davison1997}
A.~C. Davison and D.~V. Hinkley.
\newblock \emph{Bootstrap methods and their application}, volume~1.
\newblock Cambridge university press, 1997.

\bibitem[Dezeure et~al.(2015)Dezeure, B{\"u}hlmann, Meier, and
  Meinshausen]{Dezeure2014}
R.~Dezeure, P.~B{\"u}hlmann, L.~Meier, and N.~Meinshausen.
\newblock {High-dimensional Inference: Confidence intervals, p-values and
  R-Software hdi}.
\newblock \emph{Statistical Science}, 30:\penalty0 533--558, 2015.

\bibitem[Efron and Tibshirani(1994)]{Efron1994}
B.~Efron and R.~J. Tibshirani.
\newblock \emph{An introduction to the bootstrap}.
\newblock CRC press, 1994.

\bibitem[Efron et~al.(2004)Efron, Hastie, Johnstone, and
  Tibshirani]{efron04least}
B.~Efron, T.~Hastie, I.~Johnstone, and R.~Tibshirani.
\newblock {Least Angle Regression}.
\newblock \emph{Ann.\ Statist.}, 32:\penalty0 407--451, 2004.


\bibitem[Goeman et~al.(2006)Goeman, van~de Geer, and {van
  Houwelingen}]{goeman05testing}
J.~J. Goeman, S.~A. van~de Geer, and H.~C. {van Houwelingen}.
\newblock {Testing against a high-dimensional alternative}.
\newblock \emph{J.\ R.\ Statist.\ Soc.\ B},
  68:\penalty0 477--493, 2006.

\bibitem[Javanmard and Montanari(2014)]{jamo13b}
A.~Javanmard and A.~Montanari.
\newblock {Confidence intervals and hypothesis testing for high-dimensional
  regression}.
\newblock \emph{J.\ Mach.\ Learn.\ Res.}, 15:\penalty0 2869--2909,
  2014.

\bibitem[Koltchinskii(2009)]{Koltchinskii2009}
V.~Koltchinskii.
\newblock The dantzig selector and sparsity oracle inequalities.
\newblock \emph{Bernoulli}, 15\penalty0 (3):\penalty0 799--828, 2009.

\bibitem[Liaw and Wiener(2002)]{Liaw2002}
A.~Liaw and M.~Wiener.
\newblock Classification and regression by randomforest.
\newblock \emph{R News}, 2\penalty0 (3):\penalty0 18--22, 2002.
\newblock URL \url{http://CRAN.R-project.org/doc/Rnews/}.

\bibitem[Lockhart et~al.(2014)Lockhart, Taylor, Tibshirani, and
  Tibshirani]{covtest14}
R.~Lockhart, J.~Taylor, R.~J. Tibshirani, and R.~Tibshirani.
\newblock {A significance test for the lasso}.
\newblock \emph{Ann.\ Statist.}, 42:\penalty0 413--468, 2014.

\bibitem[Meinshausen(2015)]{Meinshausen2014}
N.~Meinshausen.
\newblock Group bound: confidence intervals for groups of variables in sparse
  high dimensional regression without assumptions on the design.
\newblock \emph{J.\ R.\ Statist.\ Soc.\ B}, 77\penalty0 (5):\penalty0 923--945, 2015.

\bibitem[Meinshausen and B{\"u}hlmann(2006)]{meinshausen04consistent}
N.~Meinshausen and P.~B{\"u}hlmann.
\newblock {High dimensional graphs and variable selection with the Lasso}.
\newblock \emph{Ann.\ Statist.}, 34:\penalty0 1436--1462, 2006.

\bibitem[Meinshausen and B{\"u}hlmann(2010)]{meinshausen2008ss}
N.~Meinshausen and P.~B{\"u}hlmann.
\newblock {Stability selection (with discussion)}.
\newblock \emph{J.\ R.\ Statist.\ Soc.\ B},
  72:\penalty0 417--473, 2010.

\bibitem[Meinshausen et~al.(2009)Meinshausen, Meier, and
  B{\"u}hlmann]{meinshausen09pvalues}
N.~Meinshausen, L.~Meier, and P.~B{\"u}hlmann.
\newblock {P-values for high-dimensional regression}.
\newblock \emph{J.\ Am.\ Statist.\ Ass.}, 104:\penalty0
  1671--1681, 2009.

\bibitem[Nan and Yang(2014)]{Nan2014}
Y.~Nan and Y.~Yang.
\newblock Variable selection diagnostics measures for high-dimensional
  regression.
\newblock \emph{J.\ Computnl Graph.\ Statist.}, 23\penalty0
  (3):\penalty0 636--656, 2014.

\bibitem[Ning and Liu(2014)]{Ning2014}
Y.~Ning and H.~Liu.
\newblock A general theory of hypothesis tests and confidence regions for
  sparse high dimensional models.
\newblock \emph{arXiv preprint arXiv:1412.8765}, 2014.

\bibitem[{R Development Core Team}(2005)]{R}
{R Development Core Team}.
\newblock \emph{{R: A language and environment for statistical computing}}.
\newblock R Foundation for Statistical Computing, Vienna, Austria, 2005.
\newblock URL \url{http://www.r-project.org}.

\bibitem[Reid et~al.(2016)Reid, Tibshirani, and Friedman]{Reid2013}
S.~Reid, R.~Tibshirani, and J.~Friedman.
\newblock A study of error variance estimation in lasso regression.
\newblock \emph{Statistica Sinica, to appear}, 2016.

\bibitem[Ren et~al.(2015)Ren, Sun, Zhang, and Zhou]{Ren2015}
Z.~Ren, T.~Sun, C.-H.~Zhang, and H.~H.~Zhou.
\newblock Asymptotic normality and optimalities in estimation of large Gaussian graphical models.
\newblock \emph{Ann.\ Statist.}, 43:\penalty0 991--1026, 2015.

\bibitem[Shah and Samworth(2013)]{shah2013}
R.~D. Shah and R.~J. Samworth.
\newblock Variable selection with error control: another look at stability
  selection.
\newblock \emph{J.\ R.\ Statist.\ Soc.\ B}, 75\penalty0 (1):\penalty0 55--80, 2013.

\bibitem[Sun(2013)]{scalreg}
T.~Sun.
\newblock \emph{scalreg: Scaled sparse linear regression}, 2013.
\newblock URL \url{https://CRAN.R-project.org/package=scalreg}.
\newblock R package version 1.0.

\bibitem[Sun and Zhang(2012)]{Sun2012}
T.~Sun and C.-H.~Zhang.
\newblock Scaled sparse linear regression.
\newblock \emph{Biometrika}, 99\penalty0 (4):\penalty0 879--898, 2012.

\bibitem[Sun and Zhang(2013)]{Sun2013}
T.~Sun and C.-H.~Zhang.
\newblock Sparse matrix inversion with scaled lasso.
\newblock \emph{The J.\ Mach.\ Learn.\ Res.}, 14\penalty0
  (1):\penalty0 3385--3418, 2013.

\bibitem[Tibshirani(1996)]{tibshirani96regression}
R.~Tibshirani.
\newblock Regression shrinkage and selection via the lasso.
\newblock \emph{J.\ R.\ Statist.\ Soc.\ B},
  58:\penalty0 267--288, 1996.

\bibitem[Tibshirani(2013)]{Tibshirani2013}
R.~J.~Tibshirani.
\newblock The lasso problem and uniqueness.
\newblock \emph{Electron.\ J.\ Statist.}, 7:\penalty0 1456--1490,
  2013.

\bibitem[{van de Geer} and B{\"u}hlmann(2009)]{van2009conditions}
S.~{van de Geer} and P.~B{\"u}hlmann.
\newblock {On the conditions used to prove oracle results for the lasso}.
\newblock \emph{Electron.\ J.\ Statist.}, 3:\penalty0 1360--1392,
  2009.

\bibitem[van~de Geer et~al.(2014)van~de Geer, B{\"u}hlmann, Ritov, and
  Dezeure]{optimalconf14}
S.~van~de Geer, P.~B{\"u}hlmann, Y.~Ritov, and R.~Dezeure.
\newblock {On asymptotically optimal confidence regions and tests for
  high-dimensional models}.
\newblock \emph{Ann.\ Statist.}, 42:\penalty0 1166--1202, 2014.

\bibitem[Voorman et~al.(2014)Voorman, Shojaie, and
  Witten]{voorman2014inference}
A.~Voorman, A.~Shojaie, and D.~Witten.
\newblock Inference in high dimensions with the penalized score test.
\newblock \emph{arXiv preprint arXiv:1401.2678}, 2014.

\bibitem[Wasserman and Roeder(2009)]{wasserman2009high}
L.~Wasserman and K.~Roeder.
\newblock {High dimensional variable selection}.
\newblock \emph{Ann.\ Statist.}, 37:\penalty0 2178--2201, 2009.

\bibitem[Westfall and Young(1993)]{westfall93resampling}
P.~Westfall and S.~Young.
\newblock \emph{{Resampling-based multiple testing: Examples and methods for
  p-value adjustment}}.
\newblock John Wiley \& Sons, 1993.

\bibitem[Zhang(2010)]{zhang2010nearly}
C.-H.~Zhang.
\newblock {Nearly unbiased variable selection under minimax concave penalty}.
\newblock \emph{Ann.\ Statist.}, 38:\penalty0 894--942, 2009.

\bibitem[Zhang and Zhang(2014)]{zhangzhang14}
C.-H. Zhang and S.~S. Zhang.
\newblock {Confidence intervals for low dimensional parameters in high
  dimensional linear models}.
\newblock \emph{J.\ R.\ Statist.\ Soc.\ B},
  76:\penalty0 217--242, 2014.

\bibitem[Zhang and Zhang(2012)]{Zhang2012}
C.-H. Zhang and T.~Zhang.
\newblock A general theory of concave regularization for high-dimensional
  sparse estimation problems.
\newblock \emph{Statistical Science}, 27\penalty0 (4):\penalty0 576--593, 2012.

\bibitem[Zhao and Yu(2006)]{zhao05model}
P.~Zhao and B.~Yu.
\newblock {On Model Selection Consistency of Lasso}.
\newblock \emph{J.\ Mach.\ Learn.\ Res.}, 7:\penalty0 2541--2563,
  2006.

\bibitem[Zhou(2014)]{Zhou2014}
Q.~Zhou.
\newblock Monte carlo simulation for lasso-type problems by estimator
  augmentation.
\newblock \emph{J.\ Am.\ Statist.\ Ass.}, 109\penalty0
  (508):\penalty0 1495--1516, 2014.

\bibitem[Zhou(2015)]{Zhou2015}
Q.~Zhou.
\newblock Uncertainty quantification under group sparsity.
\newblock \emph{arXiv preprint arXiv:1507.01296}, 2015.

\end{thebibliography}

\begin{thebibliography}{48}
\providecommand{\natexlab}[1]{#1}
\providecommand{\url}[1]{\texttt{#1}}
\expandafter\ifx\csname urlstyle\endcsname\relax
  \providecommand{\doi}[1]{doi: #1}\else
  \providecommand{\doi}{doi: \begingroup \urlstyle{rm}\Url}\fi

\bibitem[Boucheron et~al.(2013)Boucheron, Lugosi, and Massart]{Boucheron2013}
S.~Boucheron, G.~Lugosi, and P.~Massart.
\newblock \emph{Concentration inequalities: A nonasymptotic theory of
  independence}.
\newblock OUP Oxford, 2013.

\bibitem[Chernozhukov et~al.(2014)Chernozhukov, Chetverikov, and
  Kato]{Chernozhukov2014}
V.~Chernozhukov, D.~Chetverikov, and K.~Kato.
\newblock Central limit theorems and bootstrap in high dimensions.
\newblock \emph{arXiv preprint arXiv:1412.3661}, 2014.

\bibitem[Hao and Zhang(2014)]{Hao2014}
N.~Hao and H.~H. Zhang.
\newblock Interaction screening for ultrahigh-dimensional data.
\newblock \emph{J.\ Am.\ Statist.\ Ass.}, 109\penalty0
  (507):\penalty0 1285--1301, 2014.


\bibitem[Raskutti et~al.(2010)Raskutti, Wainwright, and Yu]{Raskutti2010}
G.~Raskutti, M.~J. Wainwright, and B.~Yu.
\newblock Restricted eigenvalue properties for correlated gaussian designs.
\newblock \emph{The J.\ Mach.\ Learn.\ Res.}, 11:\penalty0
  2241--2259, 2010.

\bibitem[Sun and Zhang(2012)]{Sun2012}
T.~Sun and C.-H. Zhang.
\newblock Scaled sparse linear regression.
\newblock \emph{Biometrika}, 99\penalty0 (4):\penalty0 879--898, 2012.

\bibitem[Zhang and Zhang(2012)]{Zhang2012}
C.-H. Zhang and T.~Zhang.
\newblock A general theory of concave regularization for high-dimensional
  sparse estimation problems.
\newblock \emph{Statistical Science}, 27\penalty0 (4):\penalty0 576--593, 2012.


\end{thebibliography}
{

}
\end{cbunit}

\end{document}